\documentclass[letterpaper, 12pt]{article}

\usepackage[english]{babel}
\usepackage[utf8]{inputenc}
\usepackage[T1]{fontenc}

\usepackage[a4paper,top=3cm,bottom=2cm,left=3cm,right=3cm,marginparwidth=1.75cm]{geometry}

\usepackage{amsmath, amssymb, amsthm}
\usepackage[dvipsnames]{xcolor}
\usepackage{graphicx}
\usepackage{amsfonts}
\usepackage{optidef}
\usepackage{dsfont}
\usepackage{booktabs}
\usepackage{adjustbox}
\usepackage{subcaption}
\usepackage{tikz-cd}

\usepackage{natbib}
\newtheorem{thm}{Theorem}[section]
\newtheorem{lem}[thm]{Lemma}
\newtheorem{assumption}[thm]{Assumption}
\newtheorem{prop}[thm]{Proposition}

\newtheorem{cor}[thm]{Corollary}

\theoremstyle{definition}
\newtheorem{defn}[thm]{Definition}
\newtheorem{ex}[thm]{Example}
\newtheorem{remark}[thm]{Remark}

\numberwithin{equation}{section}

\newcommand{\mr}{\mathbb{R}}
\newcommand{\wh}{\widehat}
\newcommand{\lag}{\langle}
\newcommand{\rag}{\rangle}
\newcommand{\op}{o_p}
\newcommand{\Op}{O_p}
\newcommand{\one}{\mathds{1}}

\newcommand{\sub}{\subseteq}
\newcommand{\rootn}{\sqrt{n}}
\newcommand{\negrootn}{n^{-1/2}}

\newcommand{\convp}{\overset{p}{\to}}
\newcommand{\convd}{\Rightarrow}
\newcommand{\normal}{\mathcal{N}}
\newcommand{\mc}{\mathcal}

\newcommand{\indep}{\perp \!\!\! \perp}
\newcommand{\est}{\wh{\theta}}

\newcommand{\thetatrue}{\theta_0}

\newcommand{\hksd}{\sigma}

\newcommand{\Xn}{X_{1:n}}

\newcommand{\Di}{D_i}

\newcommand{\Dj}{D_j}

\newcommand{\en}{E_n}

\newcommand{\ei}{e_i}

\newcommand{\bi}{b_i}

\newcommand{\ai}{a_i}

\newcommand{\corrg}{\corr_G}

\newcommand{\dig}{D_{ig}}
\newcommand{\djg}{D_{jg}}

\newcommand{\aig}{a_{ig}}
\newcommand{\ajg}{a_{jg}}

\DeclareMathOperator*{\argmin}{argmin}
\DeclareMathOperator*{\argmax}{argmax}

\DeclareMathOperator{\kl}{KL}

\DeclareMathOperator{\supp}{Supp}

\DeclareMathOperator{\bern}{Bernoulli}
\DeclareMathOperator{\unif}{Unif}

\DeclareMathOperator{\cov}{Cov}
\DeclareMathOperator{\corr}{Corr}

\DeclareMathOperator{\var}{Var}

\DeclareMathOperator{\linearspan}{span}

\DeclareMathOperator{\image}{Image}

\DeclareMathOperator{\disp}{Disp}

\DeclareMathOperator{\releff}{Efficiency}

\usepackage{enumitem}
\usepackage{booktabs}
\usepackage[colorlinks=true, allcolors=blue]{hyperref}
\usepackage{setspace}
\usepackage[format=plain,font={small,stretch=1},labelfont=bf,textfont=up]{caption}

\newcommand{\figwfull}{\textwidth}       % 1x3 strips from the decided regens: fig:multivariate-uniform-draws, fig:rs_projections, fig:net_setup
\newcommand{\figwquad}{0.9\textwidth}    % 2x2 result grids: fig:frontier_efficiency, fig:net_efficiency
\newcommand{\figwproj}{0.85\textwidth}   % eigenspace 2x2: fig:lhs_projections
\newcommand{\figwstrip}{0.8\textwidth}   % sparse 1x3 strips: fig:intro-preview, fig:multivariate-net-lattice

\newcommand{\transportsym}{\Pi_k}
\newcommand{\objn}{Q_n}
\newcommand{\objest}{\wh Q_n}
\newcommand{\eps}{\epsilon}

\newcommand{\sfn}{s}

\newcommand{\snbar}{\bar s_n}

\newcommand{\yi}{Y_i}
\newcommand{\yj}{Y_j}

\newcommand{\done}{D_1}
\newcommand{\dtwo}{D_2}

\newcommand{\ef}{E_F}
\newcommand{\varf}{\var_F}
\newcommand{\covf}{\cov_F}

\newcommand{\suppf}{\mc D}
\newcommand{\suppfpre}{\mc D_{pre}}

\newcommand{\covg}{\cov_G}
\newcommand{\si}{\sfn_i}
\newcommand{\sj}{\sfn_j}
\newcommand{\sig}{\sfn_{ig}}
\newcommand{\sjg}{\sfn_{jg}}
\newcommand{\sitilde}{\tilde \sfn_i}
\newcommand{\vark}{\var_k}

\newcommand{\match}{\tau}
\newcommand{\matchvar}{\pi}

\newcommand{\gs}{G^c}

\newcommand{\sgbar}{\bar s_g}

\newcommand{\eg}{E_G}
\newcommand{\eh}{E_h}

\newcommand{\giid}{G_{iid}}

\newcommand{\dispg}{\disp_G}
\newcommand{\thetan}{\theta_n}
\newcommand{\thetablp}{\theta_\text{BLP}}

\newcommand{\inv}{^{-1}}

\newcommand{\matchcoeffk}{Q_k}

\newcommand{\pii}{\pi_i}

\newcommand{\opone}{\op(1)}

\newcommand{\ri}{r_i}
\newcommand{\rig}{r_{ig}}
\newcommand{\rjg}{r_{jg}}

\newcommand{\evalm}{\lambda_m}
\newcommand{\eval}{\lambda}
\newcommand{\lf}{\lag}
\newcommand{\rf}{\rag_F}
\newcommand{\ltwonot}{L_0^2(F)}

\newcommand{\rim}{\ri^m}
\newcommand{\rigm}{\rig^m}
\newcommand{\rjgm}{\rjg^m}
\newcommand{\grs}{G_{R}}
\newcommand{\hm}{h_m}
\newcommand{\viid}{v_{iid}}
\newcommand{\vmatch}{v_\Delta}

\newcommand{\vg}{v_{g}}

\newcommand{\betan}{\beta_n}
\newcommand{\sih}{\si^H}

\newcommand{\sh}{s^H}

\newcommand{\varp}{\var_P}
\newcommand{\ep}{E_P}

\newcommand{\ci}{c_i}

\newcommand{\weightm}{w_m}
\newcommand{\mui}{\mu_i}
\newcommand{\muig}{\mu_{ig}}
\newcommand{\mujg}{\mu_{jg}}

\newcommand{\varn}{\var_n}
\newcommand{\covn}{\cov_n}

\newcommand{\elltwo}{L^2}
\newcommand{\ehist}{E_{hist}}
\newcommand{\phist}{P_{hist}}
\newcommand{\whist}{w_{hist}}

\newcommand{\ecomp}{E_c}
\newcommand{\eeven}{E_{even}}
\newcommand{\eodd}{E_{odd}}
\newcommand{\podd}{P_{odd}}
\newcommand{\peven}{P_{even}}
\newcommand{\ecyclic}{E_{c}}
\newcommand{\eacyclic}{E_{a}}
\newcommand{\pcyclic}{P_{c}}
\newcommand{\pacyclic}{P_{a}}
\newcommand{\ynbar}{\bar Y_n}

\newcommand{\tstar}{T^*}
\DeclareMathOperator{\TV}{TV}
\newcommand{\etatv}{\eta_{\TV}}
\newcommand{\maxraten}{R_n}
\newcommand{\infdispg}{\mathrm{ID}_G}
\newcommand{\supdispg}{\mathrm{SD}_G}
\newcommand{\smom}{\mc S_{\mathrm{m}}}
\newcommand{\sreg}{\mc S_{\mathrm{r}}}

\newcommand{\pg}{{\Pr}_G}
\newcommand{\imb}{\mc I}

\date{\today}
\title{Coupling Designs for Randomized Experiments with Complex Treatments\thanks{\linespread{1}\selectfont We are grateful to seminar and conference participants at Boston College, Emory, EuroCIM, Columbia, Harvard, Northwestern, NYU, University of Pittsburgh, Princeton, and Stanford for comments that greatly improved the paper.}}
\author{%
Max Cytrynbaum\thanks{Department of Economics, Yale University.} \and
Fredrik S{\"a}vje\thanks{Department of Economics, Uppsala University.}%
}

\usepackage{etoolbox}
\AtBeginEnvironment{abstract}{\singlespacing}
\AtBeginEnvironment{thebibliography}{\singlespacing}

\begin{document}
\begingroup
\renewcommand{\thefootnote}{\fnsymbol{footnote}}
\maketitle
\endgroup
\setcounter{footnote}{0}

\begin{abstract}
We describe a new family of experimental designs that extends the principle of stratified randomization to settings with continuous, constrained multivariate, and other irregular treatment spaces.
Our approach is to first match units into homogeneous groups, then use Monte Carlo couplings to assign within-group treatments to be highly dispersed over the treatment space.
We show that ensuring similar units receive dissimilar treatments improves estimation efficiency.
The efficiency gains are proportional to the product of dispersion and match quality, where dispersion measures how spread out the assignments are relative to independent randomization.
We develop a new spectral analysis showing how efficiency depends on alignment between the smoothness and shape of the estimator's influence function and the coupling's principal directions.
We illustrate these designs with examples from development, behavioral, and labor economics.  
In particular, our empirical application uses data from a real experiment allocating savings monitors using their position within village social networks.
\end{abstract}

\clearpage

\section{Introduction}\label{section:introduction}

Consider assigning treatments $\Di$ from a distribution $F$ to units $i \in [n]$ in a randomized experiment.
It is well known that for simple discrete distributions $F$, experimenters can improve estimation efficiency by stratifying on covariates in the experimental design.
For example, if $\Di \in \{0, 1\}$ and $F = \bern(1/2)$, a matched pairs design would match units with similar baseline covariates $X_i \approx X_j$ into pairs, then set $\Di=1$ and $\Dj=0$ or vice-versa, with equal probability.
Doing so balances covariates between treatment and control groups, improving precision for standard treatment-effect estimators \citep{bruhn2009,bai2023efficiency}.
Similarly, if there are $k$ treatments and $F = \unif([k])$, one can improve efficiency by matched $k$-tuples randomization: match units into groups of $k$ with similar baseline covariates, then randomly assign one unit to each treatment \citep{cochran1957experimental}.

Moving beyond these well understood situations, consider a researcher interested in the effect of cash grants $d \in [0, u]$ on future household consumption $\yi(d)$.
They want to estimate features of the average dose-response $d \mapsto n\inv \sum_i \yi(d)$ for grant amounts $d \in [0, u]$.
Identification of the full curve requires continuous randomization of the treatment, e.g.\ with $F = \unif[0, u]$.
However, this choice of distribution makes conventional stratification impossible, since there are infinitely many treatment levels.

One simple design for this treatment space is just to assign treatments independently, but this can be highly inefficient.
Alternatively, one could first discretize the treatment, say into $k=20$ treatment levels, then apply matched $k$-tuples randomization.
However, such discretization generally requires changing the causal estimand, since the dose-response is non-identified at intermediate points.
Even leaving such identification issues aside, match quality deteriorates rapidly as $k$ increases, reducing the efficiency gains from stratified randomization.
Indeed, in moderately sized experiments, it can even be challenging to find well-matched pairs of $k=2$ units for relatively low-dimensional covariates, and will be much more difficult to find well-matched groups of $k=20$ units.

The challenge of improving efficiency in experiments with complex treatments is not limited to continuous treatments.
Experiments with irregular, multivariate, or constrained treatments arise in a variety of applications.
Examples include cash transfer experiments in development \citep{haushofer2016}, tax and information experiments in behavioral economics \citep{reesjones2019,masatlioglu2023}, product experiments on digital platforms \citep{sahni2020}, and correspondence studies of discrimination \citep{evsyukova2025,bertrand2004emily}.
In these settings the treatment is often too rich for conventional stratification to randomize efficiently.
We develop a new family of \emph{coupling designs} to balance covariates in such applications, without forced discretization.
We make three main contributions in this paper:
\begin{enumerate}
	\item In Sections~\ref{section:overview} and \ref{section:coupling-designs}, we introduce a new family of coupling designs that enables efficient randomization in experiments with complex treatments, illustrating our method with a variety of applications.
	The key idea is to first match units into homogeneous groups, then use coupling techniques to assign treatments within these groups to be highly dispersed over the treatment space.
    % (a coupling is a joint distribution for the treatments $(\Di)_{i=1}^k$ within a group of $k$ units, with each marginal fixed at $F$).
	We construct such couplings for treatment spaces $\suppf \sub \mr^m$ by combining techniques from the Monte Carlo integration literature with tools from optimal transport.

	\item In Section~\ref{section:dispersion-match-quality}, we introduce the concepts of dispersion and match quality, showing that the efficiency gain from a coupling design is proportional to the product of these two key forces.
    %, exactly in a benchmark parametric model and eigenspace by eigenspace in general.
    Section~\ref{section:efficiency-analysis} contains our main theoretical results. 
    We develop a novel spectral analysis that shows how efficiency depends on both the smoothness and shape of the estimator's influence function as well as the principal directions of the implemented coupling.
	Section~\ref{section:coupling-analysis} applies this analysis to compare various choices of coupling designs.

	\item In Section~\ref{section:asymptotics}, we develop asymptotic theory for the family of coupling designs.
	We show asymptotic normality of a large family of estimators under coupling design randomization, and develop conservative variance estimators enabling asymptotically valid inference on general causal parameters.
\end{enumerate}

We illustrate how the designs work in practice with an application in Section~\ref{section:illustrative-applications}.
The application is grounded in an experiment in which social monitors are assigned to savers based on their position in village social networks \citep{breza2019}.

\subsection{Related Literature}

Experimental design for causal inference problems goes back to at least \citet{fisher1926}, and has been an active area of research in statistics and econometrics for decades \citep{athey2017survey,bai2025primer}.
%Recent surveys include \citet{athey2017survey} and \citet{bai2025primer}.
Several approaches to improving efficiency in randomized experiments have been developed in this literature, including rerandomization \citep{morgan2012,li2018rerandomization}, pure optimization \citep{kasy2016,kallus2017balance}, as well as approaches based on discrepancy minimization \citep{harshaw2024}.
However, these approaches have been developed for discrete, typically binary, treatments and do not directly address complex treatment spaces.

Complex treatments have been considered in superpopulation settings when the treatment assignment mechanism is fixed and externally given.
\citet{hirano2004continuous} and \citet{imai2004general} generalize the propensity score to continuous treatment regimes, and \citet{kennedy2017continuous} and \citet{colangelo2025dml} develop doubly-robust estimators of dose-response curves.
%\citet{egami2022text} provide a framework for causal inference with text-based treatments through low-dimensional representations.
In recent work, \citet{kramer2026highdim} study causal inference with high-dimensional treatments, deriving estimators and minimax rates for mean potential outcomes.
%Related analyses of high-dimensional and multivariate treatments include \citet{goplerud2025heterogeneous} and \citet{hsiao2026mixtures}.
However, none of these papers consider experimental design in such settings.

There is a large literature on improving efficiency by stratified randomization for discrete treatments, which is closely related to our work.
For example, matched pairs designs assign opposite treatments to pairs of similar units \citep{fisher1935,Greevy2004Optimal,bruhn2009}. 
The efficiency properties of matched pairs have been studied by \citet{imai2008}, \citet{fogarty2018}, \citet{bai2020pairs}, \citet{bai2021inference}, \citet{pashley2021}, among others.
Generalizations to matched $k$-tuples and other stratified designs have been considered by several authors \citep{cochran1957experimental,higgins2016blocking,bugni2018inference,bugni2019,cytrynbaum2023,bai2023tuples,bai2023efficiency}.
\citet{koo2026incomplete} provide a modern, potential outcomes based analysis of incomplete block designs \citep{yates1936incomplete}, which is stratified randomization when the number of distinct (discrete) treatments is larger than the stratum size.
%Similarly, in our setting it will either be highly inefficient or outright impossible to implement every distinct treatment within each matched group.  

The coupling designs we introduce in this paper extend the basic mechanism of stratified randomization to complex treatment spaces by replacing treatment permutations within strata with general negatively dependent couplings that disperse treatment assignments across the treatment space $\suppf \sub \mr^m$.
To generate highly dispersed treatments, we combine coupling techniques from the Monte Carlo integration literature \citep{robert2004monte,Owen2013} with maps from optimal transport theory \citep{brenier1991polar,merigot2011multiscale,carlier2010knothe} that approximately preserve geometry, in the sense that nearby random inputs map to nearby treatments in the target space $\suppf$.
In addition to enabling covariate-balancing randomization in a much broader class of experiments, coupling designs also provide new insights into conventional stratified designs, and can even improve on such designs in classical settings by more efficiently trading off between the key forces of dispersion and match quality.

\section{Overview and Illustrative Applications}\label{section:overview}

\subsection{Generalizing Matched Pairs}\label{paragraph:antithetic-variates}

Coupling designs extend the basic mechanism of stratification to allow efficient randomization from any distribution $F$ within tightly matched groups of units.
To illustrate this idea, consider extending matched pairs designs to allow randomization of continuous, univariate treatments $\Di \in \mr$.
The conventional matched pairs design can be understood as drawing $(\Di, \Dj) \sim G$ from a coupling $G$ with fixed marginals $G_i = G_j = \bern(1/2)$ and $\Di = 1-\Dj$, which achieves maximal negative correlation: $\corrg(\Di, \Dj) = -1$.
For more general distributions $F$, this perspective suggests first matching pairs of similar units $i$ and $j$, then drawing $(\Di, \Dj) \sim G$ from a coupling with fixed marginals $G_i = G_j = F$ and strong pairwise negative correlation: $\corrg(\Di, \Dj) \ll 0$.
One way to construct such a coupling for any $F$ is via a classic idea from Monte Carlo integration theory known as \emph{antithetic variates} sampling \citep{hammersley1956antithetic}.

\paragraph{Antithetic Matched Pairs.}

Antithetic variates sampling generates assignments $\Di^* = F\inv(U)$ and $\Dj^* = F\inv(1-U)$ using a common uniform variate $U \sim \unif[0, 1]$.
Since both $U$ and $1-U$ are uniform, the quantile transform produces $\Di^*, \Dj^* \sim F$ marginally.
By drawing opposing quantiles $U$ and $1-U$, the coupling $(\Di^*, \Dj^*) \sim G$ induces strong negative correlation between the treatments.
For example, the results of \citet{hoeffding1940masstab} show that for any monotone function $y(\cdot)$ of the treatment, antithetic variates achieves minimal correlation:
\begin{equation} \label{eqn:antithetic-correlation}
\corr(y(\Di^*), y(\Dj^*)) = \min_{G_i=G_j=F} \corrg(y(\Di), y(\Dj)).
\end{equation}

We can construct an \emph{antithetic matched pairs} design by assigning treatments $(\Di^*, \Dj^*) \sim G$ within tightly matched pairs, inducing strong negative correlation while also implementing the chosen marginal $\Di^* \sim F$ for each unit $i \in [n]$. 
When $F = \bern(1/2)$, the design yields $\Di^* = 1 - \Dj^*$, exactly recovering the conventional matched pairs design. 
However, this construction can be used more generally for any univariate distribution $F$, for example $F = \unif[0, u]$.

Antithetic variates were developed to improve efficiency in Monte Carlo integration problems, such as estimating $\thetatrue = \ef[y(D)]$ using $\est = n\inv \sum_i y(\Di)$.
In contrast to classic Monte Carlo integration, in causal inference units generally have heterogeneous responses to the treatment with $\yi(\cdot) \neq \yj(\cdot)$.
We can use matching to enforce approximate homogeneity at the group level.
After successful matching, $\yi(\cdot) \approx \yj(\cdot)$ within matched pairs, allowing the design to leverage the efficiency improvements from antithetic variates as if the responses were homogeneous.

Simplifying for illustration, consider an estimator $\est = (1/2)(\yi(\Di) + \yj(\Dj))$ of the average outcome $\thetatrue = (1/2)(\ef[\yi(D)] + \ef[\yj(D)])$.
If units are perfectly matched $\yi(\cdot) = \yj(\cdot) = y(\cdot)$, the variance relative to independent assignment is
\begin{equation}\label{eqn:relative-efficiency-toy}
\frac{\var_G(\est)}{\var_{\giid}(\est)} - 1 = \corr_G(\yi(\Di), \yj(\Dj)) = \corr_G(y(\Di), y(\Dj)).
\end{equation}
If $y(\cdot)$ is monotone, then by Equation \eqref{eqn:antithetic-correlation} this antithetic pairs design can significantly improve efficiency.
Indeed, under perfect matching and monotone $y(\cdot)$, it minimizes variance of the estimator $\est$ among all pair-wise couplings $G$.
Note that this example is only for illustration, and in what follows we do not impose monotonicity of the outcome functions $\yi(\cdot)$.

\subsection{Coupling Designs}

The simple example above illustrates how matched pairs randomization can be extended to randomize efficiently from any univariate distribution $F$ when the response function is monotone.
In what follows, we generalize this construction to provide coupling-based randomization methods for general treatment spaces $\suppf \sub \mr^m$.
We show that these methods can improve efficiency under weak smoothness conditions on the potential outcomes $\yi(\cdot)$, given successful matching.

We construct general coupling designs by first matching the experimental units into homogeneous groups of size $k \ge 2$ using covariates.
Next, treatments are drawn within each group of $k$ units from a coupling $(\Di)_{i=1}^k \sim G$ with marginals $G_i = F$ for all $i \in [k]$, for a fixed distribution $F$ over the treatment space $\suppf = \supp(F)$.
We can view this as a matched $k$-tuples design with coupling-based randomization within each group.
The family of coupling designs can accommodate very general marginal distributions $F$ and spaces $\suppf$.
When $F$ is discrete, we also recover conventional stratified randomization for appropriate choices of $k$ and $G$.
There are many possible couplings that can be used to randomize within groups.
We discuss several examples and provide a general coupling construction in Section~\ref{section:coupling-designs} below.

\textbf{Dispersed Treatments.}
As suggested by Equation \eqref{eqn:relative-efficiency-toy}, we would like to produce negatively correlated treatments $(\Di)_{i=1}^k$ within matched groups of $k$.
For multivariate $D \in \mr^m$, this can be achieved by sampling $(\Di)_{i=1}^k$ to be highly dispersed or ``spread out'' over the treatment space $\suppf$.
We show that making treatments dispersed in this way implies that $\corrg(\phi(\Di), \phi(\Dj)) < 0$ for $i \neq j$ and sufficiently smooth functions $\phi: \suppf \to \mr$.

When tuple size $k$ is large, it is possible to achieve higher dispersion, since this allows us to coordinate randomization of $(\Di)_{i=1}^k$ to cover more of the treatment space.
However, as tuple size $k$ increases, match quality becomes worse and the response functions $\yi(\cdot)$ and $\yj(\cdot)$ of matched units $i$ and $j$ are less similar on average.
Using a formalization of these concepts introduced in Section~\ref{section:dispersion-match-quality}, we show that the efficiency gain from a coupling design relative to independent assignment is
\begin{equation*}
\text{Efficiency Gain} = \text{Dispersion} \times \text{Match Quality}.
\end{equation*}

Intuitively, by assigning matched $k$-tuples of similar units to highly dissimilar treatments $(\Di)_{i=1}^k$, we prevent spurious in-sample correlations from arising between the treatment assignments and unit-specific heterogeneity.
We demonstrate that this generalizes classical notions of covariate balance for binary treatments to much more complex treatment spaces and distributions.

To show our main result in full generality, Section~\ref{section:efficiency-analysis} defines a \emph{coupling operator} $U_G$ acting on the space of influence functions $\si(\cdot)$ in $\elltwo(F)$, whose eigenspaces can be viewed as the principal directions of the coupling $G$ with respect to random sampling.
The dispersion $\dispg(\phi)$ of any $\phi \in \elltwo(F)$ decomposes orthogonally over these eigenspaces, so the overall efficiency gain depends on how well the estimator's influence functions $\si(\cdot)$, defined in Section~\ref{subsection:causal-estimands}, align with the high-dispersion directions of $G$.
This yields a general decomposition of the variance reduction from coupling designs into a weighted sum of eigenspace-specific dispersion $\times$ match quality terms, where the weights reflect the approximation quality of $\si(\cdot)$ on each eigenspace.

%For an experiment of size $n$, if $R > n$ then conventional stratified randomization is infeasible. 
%Even if $R \le n$, stratified randomization may not meaningfully improve precision due to poor match quality, since it is difficult to find $R$ almost identical users to match together.  

\subsection{Illustrative Applications}\label{ex:irregular-uber-eats}

To make the ideas of the paper concrete, we describe several applications of coupling designs to experiments with complex treatments, where conventional stratification would either be impossible or have very poor performance.

In the leading application, treatments are drawn from a discrete treatment catalog $\suppf = \{d_1, \dots, d_R\}$, where each $d_r \in \mr^m$ is a vector of features describing the treatment itself.
The features are aspects of the treatments related to their hypothesized effects, while baseline covariates $X_i$ describing the units are used separately for matching.
This structure arises in several settings:
\begin{itemize}[itemsep=2pt, topsep=2pt]
\item \textbf{Product promotions \citep{sahni2020}.}
A restaurant-search platform shows restaurant promotions to its users and wishes to learn which type of restaurants its users would respond to.
Featurizing each restaurant by attributes such as cuisine type, average price, and rating gives a treatment catalog $\suppf \subset \mr^m$, an irregular point cloud.
The treatment $\Di \in \suppf$ is the restaurant promoted to user $i$, and the outcome $\yi(d) \in \{0, 1\}$ records whether the user contacts the restaurant when shown $d$.
The platform approximates the average choice surface $\ynbar(d) = n\inv \sum_{i=1}^n \yi(d)$ with a Logit model $\ynbar(d) \approx L(\beta'd)$, where $L(x) = 1/(1 + e^{-x})$.

\item \textbf{Targeted cancer therapies.}
Cancer patients are treated with one of many targeted or biologic agents.
Featurizing each agent by molecular features such as its target-binding profile gives a catalog of therapies, with the outcome $\yi(d)$ the tumor's response to the assigned agent $d$.
For example, the NCI-MATCH trial assigned patients across dozens of agents based on their tumors' genomic alterations \citep{flaherty2020}, and the I-SPY 2 platform trial randomized patients among more than 20 candidate agents within tumor-subtype groups \citep{barker2009}.
\end{itemize}

In each case, the experimenter wants to study features of the average response surface $d \mapsto \ynbar(d)$ over the catalog, possibly approximating it over a parametric class such as linear or logit.
Coupling designs improve precision for such parameters by drawing the treatments $(\Di)_{i=1}^k$ within each matched $k$-tuple from a coupling $G$ that is highly dispersed over the catalog (Figure~\ref{fig:intro-preview}).
Section~\ref{subsection:design-construction} details our construction of such couplings, which combines tools from Monte Carlo integration theory and optimal transport (OT).

\begin{figure}[ht]
    \centering
    \includegraphics[width=\figwstrip]{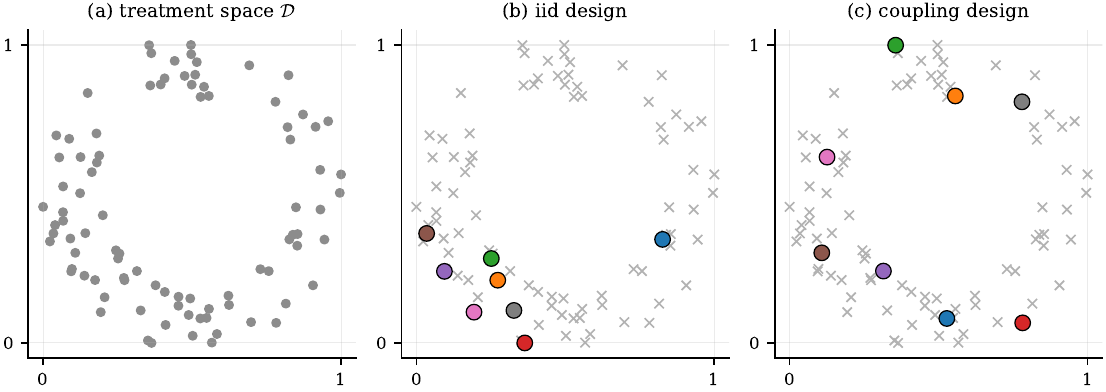}
    \caption{A coupling design over a discrete treatment space.
    (a) A catalog of possible treatments as a point cloud in feature space.
    (b) Under iid randomization, the treatments assigned to one matched $k$-tuple (colored circles, $k = 8$) can cluster on near-duplicate treatments, while under a coupling design (c) the same units' treatments are highly dispersed.}
    \label{fig:intro-preview}
\end{figure}

By contrast, classical stratified randomization would assign each of the $R$ distinct treatments without replacement within groups of $k = R$ units.
This is impossible if the experiment size is $n < R$. More generally, if $R$ is large, stratification will behave similarly to iid randomization, due to poor match quality within groups of $k \ge R$ units.
Coupling designs instead draw treatments within tightly matched $k$-tuples for any $k \ge 2$, producing dispersed treatments subject to a fixed constraint on match quality.
Similar challenges, where the treatment space is too complex for stratification, arise in many other experiments:

\begin{ex}[Continuous Treatments] \label{ex:continuous-dose}
Consider a univariate continuous dose $\Di \in [0, u]$ with marginal $F = \unif[0, u]$, where the researcher wants to estimate the average dose-response $\ynbar(d)$, or its best linear approximation, over the interval.
For example, \citet{haushofer2016} study the effect of unconditional cash transfers on the consumption and well-being of poor households in Kenya, and \citet{kreindler2024} studies how commuters in Bangalore respond to per-kilometer road congestion charges.
In both experiments, the inherently continuous treatment was implemented at only a handful of levels, so the dose-response is identified at just those points.
A coupling design can instead efficiently randomize the dose continuously from $F$ over the whole interval, identifying the full dose-response curve.
\end{ex}

\begin{ex}[Factorial and Constrained Treatments] \label{ex:labor-factorial}\label{ex:development-constrained}
Many experiments cross-randomize several treatment factors, which may be either continuous or discrete, producing a multivariate factorial treatment $\Di = (D_{i1}, \dots, D_{im})$ over a product space $\suppf = \times_{j=1}^m \suppf_j$, with a product distribution such as $F = \unif(\suppf)$.
For example, \citet{banerjee2025} study immunization-nudges by cross-randomizing monetary incentives, reminder messages, and network ``ambassadors'' into $75$ distinct treatment bundles, and \citet{bertrand2010} cross-randomize interest rates and creative-content features of consumer-credit offers.
In other experiments, logistical, fairness, budgetary, or other considerations restrict which treatments can be implemented, leading to a constrained treatment space $\suppf = \{d \in \suppfpre : C(d) \leq B\}$, where $\suppfpre$ may be a product set.
For example, \citet{ashraf2010} randomize a household's offer price $p$ and discounted transaction price $\tilde p \le p$ for a health product, so the treatments lie in a triangle, \citet{reesjones2019} randomize a schedule's average and marginal tax rates over a product grid with the strongly progressive schedules removed, and \citet{masatlioglu2023} assign information structures $\Di = (p, q)$ lying on a one-dimensional contour of constant informativeness in $[0, 1]^2$.
In each case the number of distinct treatments is large or infinite, and the space $\suppf$ itself may be irregular.
\end{ex}

\begin{ex}[Text and Image Treatments] \label{ex:text-images}
\citet{evsyukova2025} run a field experiment on LinkedIn in which fictitious profiles send connection requests, signaling race through AI-generated profile photographs while holding the rest of the profile fixed.
We can view the treatment as $\Di = (C_i, Z_i)$, where $C_i \in \{0, 1\}$ is the race cue and $Z_i$ collects high-dimensional, featurized content such as text and listed skills.
Letting $\Di \sim F$ with $C_i \indep Z_i$, regressing $\yi \sim 1 + C_i$ estimates the average effect of the race cue, generalizing the classic resume correspondence studies of \citet{bertrand2004emily} and \citet{kline2022}.
Here a coupling design disperses both the race cue and the content features $Z_i$ within matched $k$-tuples of similar recipients.
\end{ex}

\textbf{Exploiting Geometry and Smoothness.}
In each application, the treatment space is embedded in $\mr^m$, directly or through the featurization, and carries a geometry, with treatments $d_r$ and $d_l$ similar when $|d_r - d_l|_2$ is small.
It is often reasonable to assume that the outcome functions have some form of smoothness in this geometry, with units responding similarly to nearby treatments: $\yi(d_r) \approx \yi(d_l)$ when $|d_r - d_l|_2 \approx 0$.
For instance, a user decides similarly whether to contact two promoted restaurants with nearly identical cuisine, price, and rating.
Because of this, assigning units with similar response functions $\yi(\cdot) \approx \yj(\cdot)$ to nearby treatments effectively wastes a sample: we learn the same thing from observing $\yi(d_r)$ and $\yj(d_l)$.
This geometry is what a coupling design exploits.
By assigning matched groups of similar units to highly dissimilar treatments, coupling designs make each observation informative about a distinct part of the response surface.
Sections~\ref{section:efficiency-analysis} and \ref{section:coupling-analysis} connect the efficiency gain from coupling designs to the smoothness of the responses $\yi(\cdot)$, using a general notion of smoothness based on total variation that accommodates discrete treatments.

\section{Coupling Designs}\label{section:coupling-designs}

\subsection{Causal Estimands and Estimators} \label{subsection:causal-estimands}

We consider estimators $\est = n\inv \sum_i \si(\Di)$, where the functions $\si(\cdot)$ are non-random in a design-based framework.
For example, we could have $\si(d) = s(\yi(d), d, X_i)$ for a fixed function $s(\cdot)$ of the data.
The corresponding finite-population estimand is
\begin{equation} \label{eqn:finite-estimand}
\thetan \equiv n \inv \sum_i \ef[\si(D)].
\end{equation}
We can view $\thetan$ as a ``fully heterogeneous'' version of the classic Monte Carlo estimand $\thetatrue = \ef[s(D)]$. 
In fact, these types of estimators and estimands are ubiquitous in causal inference problems \citep{harshaw2025general}.
Thus, experimental design in causal inference can be viewed as a Monte Carlo integration problem with fully heterogeneous test functions $\si(\cdot)$ for $i = 1, \dots, n$.
For brevity, in what follows we will denote $\en[\ai] \equiv n\inv \sum_i \ai$ for any array $(\ai)_{i=1}^n$.

\begin{ex}[Dose-Response BLP] \label{ex:blp}
Define the average dose-response function $\ynbar(d) = \en[\yi(d)]$ for $d \in \mr^m$.
The best linear approximation (BLP) coefficient is
\begin{equation} \label{eqn:average-dose-response-blp}
\thetablp = \argmin_{\theta} \min_{\alpha} \ef[(\ynbar(D) - \alpha - \theta'D)^2].
\end{equation}
For binary treatments, this is just the sample average treatment effect (SATE): $\thetablp = \en[\yi(1) - \yi(0)]$.
More generally, see \citet{yitzhaki1996} for an interpretation of $\thetablp$ as a weighted average of the marginal effects $\partial \ynbar(d) / \partial d$ for $d \in \mr$.
For estimation, define the weights $H(d) = \varf(D)\inv(d - \ef[D])$ and form an estimator $\est = \en[\yi(\Di)H(\Di)]$, which is of the form above for $\si(d) = \yi(d)H(d)$.
This can be understood as a generalization of the classic Horvitz-Thompson estimator. 
See \cite{harshaw2023bipartite} for a related construction in the context of bipartite experiments.
A simple calculation shows that this estimator recovers the BLP coefficient:
\begin{equation*}
n\inv \sum_i \ef[\si(D)] = \varf(D)\inv \, n\inv \sum_i \covf(\yi(D), D) = \thetablp.
\end{equation*}
\end{ex}

\paragraph{Influence Functions.}

A variety of estimators admit the design-based asymptotic linearization $\wh \beta - \betan = \en[\si(\Di)] + \op(n^{-1/2})$.
Because of this, our efficiency analysis of the quantity $\en[\si(\Di)]$ under coupling designs also characterizes the first-order efficiency of significantly more general parametric estimators.
For example, let $\wh \beta$ be the coefficient from the OLS regression $\yi \sim 1 + \Di$.
Then, under weak conditions on the design,
\begin{equation} \label{eqn:ols-influence-function}
\wh \beta - \thetablp = \en[\si(\Di)] + \Op(n\inv).
\end{equation}
We have $\si(d) = e_i(d)H(d)$ for residual $e_i(d) = \yi(d) - \ef[\ynbar(D)] - \thetablp'(d - \ef[D])$.
See Appendix~\ref{subsection:estimator-asymptotics} for more details.
Slightly abusing terminology, in what follows we refer to $\si$ as unit $i$'s \emph{influence function}.

\begin{ex}[Discrete Choice] \label{ex:logit}
Recall the binary choice setting in Section~\ref{ex:irregular-uber-eats}.
The dose-response function $\ynbar(d) \equiv \en[\yi(d)]$ for potential outcomes $\yi(d) \in \{0, 1\}$ is the proportion of units in the finite population that choose $\yi(d) = 1$ when given treatment $d \in \mr^m$.
Suppose that we approximate $\ynbar(d)$ using the logit model $L(\beta'd)$ by maximum likelihood, with $L(x) = 1/(1 + e^{-x})$.
Let $\wh \beta$ be the MLE estimator of the coefficients of the logit model.
We show $\wh \beta \convp \betan$, where $\betan$ is the best logistic approximation to the dose-response $\ynbar(\cdot)$.
This approximation minimizes the expected KL divergence between a Bernoulli outcome with mean $\ynbar(D)$ and one with mean $L(\beta'D)$.
In particular, letting $\kl(p \,\|\, q)$ denote the KL divergence between two Bernoulli distributions with means $p$ and $q$, we show
\begin{equation} \label{eqn:logit-kl-projection}
\betan = \argmin_{\beta} \; \ef\big[\kl\big(\ynbar(D) \; \|\; L(\beta'D)\big)\big].
\end{equation}

Define Jacobian $J_n = \ef[L(\betan'D)(1-L(\betan'D))DD']$ and prediction residual $e_i(d) = \yi(d) - L(\betan'd)$.
We show in Appendix~\ref{subsection:estimator-asymptotics} that $\wh \beta - \betan = \en[\si(\Di)] + \Op(n\inv)$ for influence functions $\si(d) =  e_i(d) \cdot J_n\inv d$.
Thus, our efficiency theory for simple estimators $\est = \en[\si(\Di)]$ also characterizes the first-order efficiency of the logit MLE $\wh \beta$ and other parametric M-estimators, for example.
\end{ex}

Examples \ref{ex:blp} and \ref{ex:logit} can be extended to accommodate regressions $\yi \sim 1 + t(\Di)$ or logit with a set of basis functions $t(d)$.
For example, if $D \in \mr^2$, we could choose $t(d) = (d_1, d_2, d_1 d_2)$ for the regression $\yi \sim 1 + D_{i1} + D_{i2} + D_{i1}D_{i2}$, as in Example~\ref{ex:labor-factorial} above.
Similar to the previous results, this provides the best approximation of $\ynbar(d)$ among all linear functions with both main effects and two-way interactions.

%\begin{remark}[Marginal Distribution] \label{rem:marginal-distribution}
%We treat the marginal distribution $F$ as given.
%For motivation, observe in Example \ref{ex:blp} that changing $F$ changes the estimand $\thetan = \thetan(F)$ naturally recovered by OLS.
%Alternatively, one could view $F$ as a free parameter to be optimized.
%To do so, one would fix a target estimand $\thetan(Q)$, then use importance sampling $\est = \en[\si(\Di)w(\Di)]$ with $w(d) = (dQ/dF)(d)$ to recover $\thetan(Q)$ for any suitable choice of $F$.
%This generalizes the classic Neyman allocation, which optimizes treatment probability $p = P(D=1)$ to improve estimation efficiency.
%In this paper, we view any such optimization of the marginal $F$, if desired, as already fixed and instead focus on providing efficient randomization procedures to implement this choice of $F$.
%See Appendix (cite) for further discussion.
%\end{remark}

\subsection{Design Construction}\label{subsection:design-construction}

Consider a target distribution $D \sim F$ over the treatment space $\suppf \sub \mr^m$. A coupling design can be implemented using the following three steps:

\begin{enumerate}[label=(\arabic*), itemsep=0.1pt]
	\item \textbf{Match.}
Match units into homogeneous groups $g$ of size $|g| = k \ge 2$ using covariates.
	We refer to the groups as matched $k$-tuples, assuming $k$ divides $n$ for simplicity.
	The matching is a bijection $\match: [n] \to [k] \times [n/k]$ that assigns each unit $i$ to $\match(i) = (j,g)$, a unique position $j \in [k]$ in group $g \in [n/k]$.
   \item \textbf{Disperse.}
Independently for each group $g$, draw samples $(U_{ig})_{i=1}^k \sim G_U$ from an exchangeable coupling $G_U$ with each marginal $U_{ig} \sim \unif[0, 1]^m$, such that the collection $(U_{ig})_{i=1}^k$ is highly dispersed over the unit cube $[0, 1]^m$.
   \item \textbf{Transport.}
Set treatments $D_{ig} = T(U_{ig})$ for all units $i \in [k]$ in group $g$, where $T: [0, 1]^m \to \suppf$ is a measurable, geometry-preserving transport map such that $T(U) \sim F$ for $U \sim \unif[0, 1]^m$.
\end{enumerate}

The first matching step is common to all stratified designs, while the second and third steps are unique to the coupling designs we introduce in this paper.
One strength of our design-based theory is to remain largely agnostic about the specific details of the matching algorithm.
The efficiency gain from coupling designs will be larger to the extent that we are able to match groups of units with very similar response functions $\yi(\cdot)$.
One possible proxy for this objective is to minimize a covariate discrepancy:
\begin{equation} \label{eqn:covariate-matching-discrepancy}
\sum_{g} \sum_{i, j \in [k]} |X_{ig} - X_{jg}|^2.
\end{equation}
There are several algorithms for this problem \citep{Greevy2004Optimal,bai2021inference,cytrynbaum2023}.
Due to the curse of dimensionality in matching, this should be done using a small set of covariates expected to be highly predictive for endline outcomes $\yi$.
%For a more detailed technical discussion about the role of covariates, see Section \ref{subsection:match-quality} below.

A coupling is a joint distribution $G$ over $\suppf^k$ with fixed marginals.
For tractability and expositional clarity, we consider couplings $G$ that are exchangeable and have identical marginals $G_i = F$ for $i \in [k]$.
A coupling is exchangeable if the joint distribution of $(D_i)_{i=1}^k$ is invariant to permutations of the indices $i \in [k]$.
We define the set of feasible couplings $\transportsym(F)$ to be the exchangeable joint distributions $G$ over $\suppf^k$ with fixed marginals $G_i = F$.

%By drawing highly dispersed $(\Di)_{i=1}^k$ within matched $k$-tuples of similar units, coupling designs prevent spurious in-sample correlations from arising between the treatments $\Di$ and unit-specific heterogeneity in $\si(\cdot)$.

\subsection{Uniform Couplings}\label{subsection:matching-couplings}

There are many strategies for sampling highly dispersed uniform random variables $(U_i)_{i=1}^k \sim G_U$.
Inspired by the Monte Carlo integration literature, we consider three canonical examples based on the Gaussian copula, binning, and randomly shifted lattices.

% Univariate-couplings figure removed in Pass 07 to meet the page limit.
% Restore by uncommenting the following line.
%\input{inputs/paper_univariate_couplings}

\begin{figure}[ht]
    \centering
    \includegraphics[width=\figwfull]{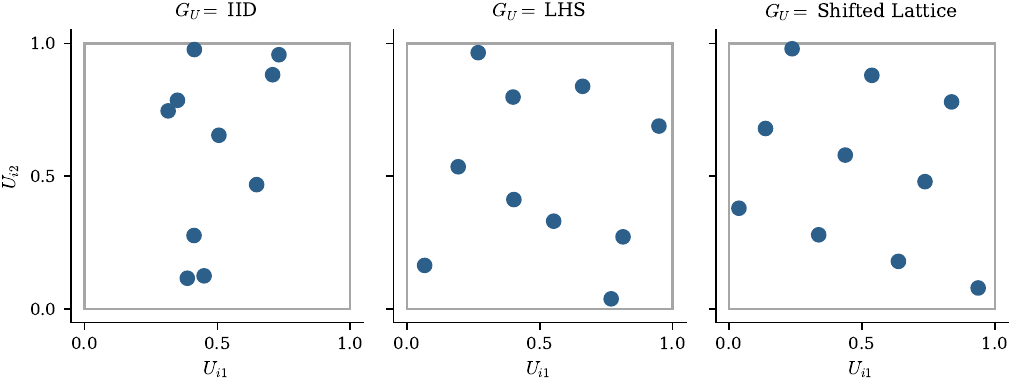}
    \caption{Uniform draws $(U_i)_{i=1}^k$ in $[0,1]^2$ with tuple size $k=10$ under the iid, Latin hypercube, and shifted lattice couplings, in increasing order of joint dispersion.}
    \label{fig:multivariate-uniform-draws}
\end{figure}

\begin{ex}[Gaussian Copula]\label{ex:gaussian}
For univariate treatments $m=1$, draw $Z \sim \normal(0, \Sigma)$ with correlation matrix $\Sigma_{ii} = 1$ and maximal negative correlation $\Sigma_{ij} = -(k-1)\inv$ for $i \neq j$.
Let $\Phi$ be the standard normal CDF and rank transform $U_i = \Phi(Z_i)$ for $i \in [k]$ so that each $U_i \sim \unif[0, 1]$.
For $\suppf \sub \mr^m$ with general $m \ge 1$, independently draw $(U_{ij})_{i=1}^k$ for each dimension $j \in [m]$ using this procedure.
\end{ex}

\begin{ex}[Latin Hypercube]\label{ex:latin-hypercube}
For $m=1$, partition $[0, 1]$ into $k$ bins of width $1/k$, with bin $J_r = [(r-1)/k, r/k]$ for $r \in [k]$.
Draw a random permutation $\pi$ of $[k]$ uniformly and assign unit $i$ to bin $J_{\pi(i)}$.
Conditional on $\pi$, draw $U_i \sim \unif(J_{\pi(i)})$, independently for $i \in [k]$.
For general $m \ge 1$, sample $(U_{ij})_{i=1}^k$ as above independently for each dimension $j \in [m]$ \citep{mckay1979comparison}.
This is sometimes referred to as a k-rooks design.
To see why, note if $m = 2$, each point $U_i$ occupies a unique row and column of the $k \times k$ grid of bins on $[0,1]^2$, so the $k$ points can be viewed as non-attacking rooks on a chessboard \citep{shirley1991discrepancy}.
\end{ex}

Latin hypercube produces samples $(U_i)_{i=1}^k \sim G_U$ such that the coordinate projections $(U_{ij})_{i=1}^k$ are spread out over the interval $[0, 1]$ for each dimension $j \in [m]$.
However, even if the one-dimensional projections $(U_{ij})_{i=1}^k$ are highly dispersed, the samples $(U_i)_{i=1}^k$ may not be jointly well-dispersed through the hypercube $[0, 1]^m$ if $m > 1$.
Because of this, Latin hypercube produces strong negative correlations $\corrg(\phi(\Di), \phi(\Dj)) < 0$ only for univariate functions like $\phi(d_1, \dots, d_m) = d_1$, but tends to have weaker effects for jointly varying functions like $\phi(d) = d_1 d_2$.
Figure~\ref{fig:multivariate-uniform-draws} illustrates this by comparing draws in $[0, 1]^2$ under iid sampling, the Latin hypercube, and the more dispersed shifted lattice coupling introduced next.

This problem can be solved with more advanced multivariate coupling constructions.
Some examples include orthogonal array Latin hypercube sampling \citep{owen1992orthogonal, tang1993orthogonal}, scrambled digital nets \citep{owen1995scrambled}, and shifted rank-1 lattice rules \citep{cranley1976randomization, sloan1994lattice}.
For brevity, we only formally describe shifted lattice rules.
In the univariate case, such couplings draw $(U_i)_{i=1}^k \sim G_U$ by adding a uniform random shift to a regular grid $(l/k)_{l=0}^{k-1} \sub [0, 1]$.
For more general $m \ge 1$, they apply a random shift $S \sim \unif[0, 1]^m$ to a dispersed lattice of points determined by a number-theoretic construction.

\begin{ex}[Exchangeable Shifted Lattice] \label{ex:shifted-lattice}
Pick integers $z_j$ with $\gcd(z_j, k) = 1$ for each $j \in [m]$ and let $z = (z_1, \dots, z_m)$.
Draw a random permutation $\pi$ and a shift $S \sim \unif[0, 1]^m$ with $\pi \indep S$.
Then, for each $i \in [k]$ set
\begin{equation} \label{eqn:shifted-lattice}
    U_i = \left( \frac{\pi(i)}{k} z + S \right) \pmod 1.
\end{equation}
\end{ex}
For $m=1$ and $z=1$, we have $(U_i)_{i=1}^k = (\pi(i)/k + S)_{i=1}^k \pmod 1$, a random shift and permutation of the regular grid $(l/k)_{l=0}^{k-1}$. This is known as rotation sampling \citep{fishman1983antithetic}.
For $m \ge 2$, the restriction $\gcd(z_j, k) = 1$ guarantees that the projections $(U_{ij})_{i=1}^k$ onto each coordinate $j \in [m]$ are themselves rotation samples.
In addition to being marginally dispersed on each dimension, a well-chosen generating vector $z$ ensures that the $(U_i)_{i=1}^k$ are also jointly well-dispersed through $[0, 1]^m$.
For a theoretical analysis of the choice of generating vector $z$, see \cite{kuo2003component}.
We slightly modify the standard construction of \citet{sloan1994lattice}, adding a random permutation $\pi$ for exchangeability due to heterogeneity of the units in causal inference problems.

\begin{figure}[ht]
    \centering
    \includegraphics[width=\figwstrip]{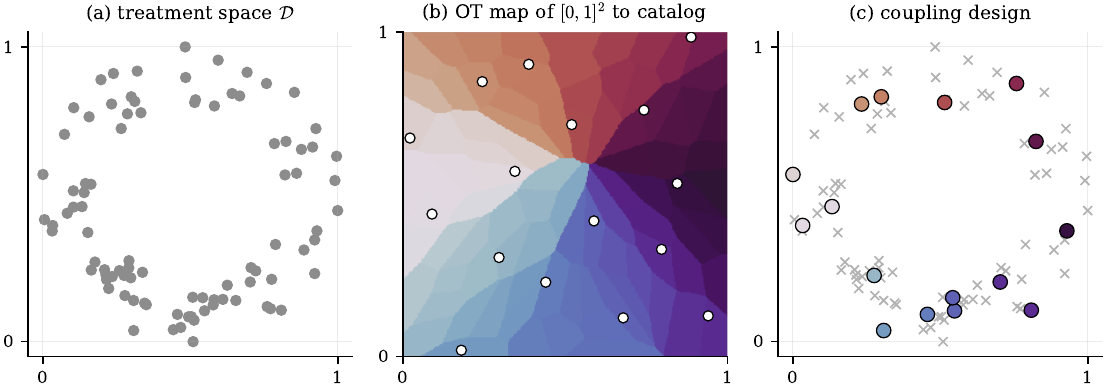}
    \caption{A coupling design over a discrete catalog. (a) The catalog $\suppf$ in feature space. (b) The Brenier map partitions $[0,1]^2$ into one cell per treatment, with draws from a scrambled digital net in white. (c) The assignments $\Di = \tstar(U_i)$ are dispersed over the catalog.}
    \label{fig:multivariate-net-lattice}
\end{figure}

\subsection{Transport Maps}\label{subsection:transport}

In general, we must map the uniform samples $(U_i)_{i=1}^k$ to treatments $\Di = T(U_i)$ that are both dispersed over $\suppf$ and have the correct marginal distribution $F$.
For univariate treatments $m=1$, we can use the transport map $T(u) = F\inv(u)$, setting $\Di = F\inv(U_i)$, so that $\Di \sim F$ by the properties of the quantile function.
In the multivariate case, if $F = \otimes_{j=1}^m F_j$ has independent components $D_{ij} \indep D_{il}$ for $j \neq l \in [m]$, then we can similarly enforce $T(U_i) \sim F$ by applying the quantile transform componentwise:
\begin{equation}\label{eqn:product-quantile-map}
T(U_i) = \big(F_1\inv(U_{i1}), \, \dots, \, F_m\inv(U_{im})\big).
\end{equation}
The componentwise quantile transform can only be used when the treatment space is a product set $\suppf = \times_{j=1}^m \suppf_j$ and the desired marginal distribution $F$ has independent components.
This will not work for general multivariate distributions $F$ with dependent components, since $\Di = T(U_i)$ would not have the correct joint distribution.
Note that several of the applications described in Section~\ref{section:overview} do not have such an independent structure and require more general transport maps.

\textbf{Geometry Preservation.} The quantile transform preserves the ordering of the samples $(U_i)_{i=1}^k$ and maps well-separated quantile levels to treatments that are far apart in $F$-probability.
In particular, if $F$ is non-atomic and $U_i$ and $U_j$ are far apart in $[0, 1]$, then $F\inv(U_i)$ and $F\inv(U_j)$ will also be far apart in $\suppf$.
To extend this construction to complex treatment spaces with dependent components, we require a geometry-preserving map $T: [0, 1]^m \to \suppf$ with $T(U) \sim F$.
This requirement naturally suggests the \emph{Brenier maps} from optimal transport \citep{brenier1991polar}:
\begin{equation}\label{eqn:monge-ot-intro}
\tstar = \argmin_{T:\, T(U) \sim F} \int_{[0,1]^m} |u - T(u)|_2^2 \, du.
\end{equation}
For $F$ with a finite second moment, a minimizer exists and is unique up to almost-everywhere equivalence \citep{brenier1991polar, mccann1995existence}.
In the univariate case and when $F = \otimes_{j=1}^m F_j$, the Brenier map recovers the componentwise quantile transform in Equation \eqref{eqn:product-quantile-map} above.
However, optimal transport can be used to construct geometry-preserving maps $T^*(U) \sim F$ for much more general spaces $\suppf$.
For discrete spaces $\suppf \sub \mr^m$, such maps can also be computed efficiently by semi-discrete optimal transport \citep{merigot2011multiscale}.
Figure~\ref{fig:multivariate-net-lattice} illustrates this construction for a discrete treatment catalog.
We discuss the geometry-preserving property motivating this definition and further computational details in Appendix~\ref{section:ot-maps}.

\section{Dispersion and Match Quality} \label{section:dispersion-match-quality}

This section describes the key mathematical objects for quantifying the efficiency of coupling design randomization, defining appropriate measures of \emph{match quality} and \emph{sample dispersion}.
This allows us to show a simple fundamental relation between these objects and the efficiency gain from a coupling design:
\begin{equation*}
\text{Efficiency Gain} = \text{Dispersion} \times \text{Match Quality}.
\end{equation*}

\subsection{Smoothness and Dispersion} \label{subsection:dispersion}

Above, we constructed couplings that ``spread out'' treatments over the treatment space $\suppf$.
For intuition about how this improves precision, recall the restaurant promotion application in Section~\ref{ex:irregular-uber-eats}, where we argued that coupling designs can exploit smoothness of the decision $y(d) \in \{0,1\}$ in restaurant features $d \in \mr^m$.
Let $n = k$ and suppose preferences are homogeneous with $\yi(d) = y(d)$.
We can use the Horvitz-Thompson estimator to estimate the best linear approximation of $y(\cdot)$ as in Example~\ref{ex:blp}.
Then $\est = k\inv \sum_{i=1}^k \phi(\Di)$ with influence function $\phi(d) = y(d)H(d)$ for weights $H(d) = \varf(D)\inv(d - \ef[D])$, a homogeneous version of the more general heterogeneous estimation problems in Section~\ref{subsection:causal-estimands}.
When $y(\cdot)$ is smooth, so is $\phi(\cdot)$.
If by random chance we sample treatments $\Di \approx \Dj$ close together in the space $\suppf$, the samples $\phi(\Di) \approx \phi(\Dj)$ will be quite similar, which effectively wastes an experimental sample.
By contrast, if $(\Di)_{i=1}^k$ are dispersed over $\suppf$, then we learn more about the function $\phi(\cdot)$ from a given fixed sample size $k$.

To formalize this intuition, let sample variance $\var_k(a_i) \equiv (k-1)^{-1} \sum_{i=1}^k (a_i - \bar a)^2$ for any $(\ai)_{i=1}^k$.
We define dispersion as a normalized measure of how spread out the samples $(\phi(\Di))_{i=1}^k$ are in expectation over $G$.

\begin{defn}[Dispersion] \label{defn:dispersion}
For coupling $G$ with margins $F$, if $\varf(\phi) > 0$, define
\begin{equation} \label{eqn:dispersion-def}
\dispg(\phi) \equiv (k-1) \left( \frac{\eg \vark(\phi(\Di))}{\varf(\phi)} - 1 \right).
\end{equation}
If $\varf(\phi) = 0$, define $\dispg(\phi) \equiv 0$.
\end{defn}
For the iid design $\giid = \otimes_{i=1}^k F$, we have $E_{\giid} \vark(\phi(\Di)) = \varf(\phi)$ by unbiasedness of the sample variance, so $\disp_{\giid}(\phi) = 0$.
If $\dispg(\phi) > 0$, then the samples $(\phi(\Di))_{i=1}^k$ are more spread out in expectation under $G$ than under iid randomization.
For homogeneous estimation problems and smooth $\phi(\cdot)$, this improves efficiency by the mechanism described above.
Indeed, the relative efficiency for the homogeneous problem above is
\begin{equation}\label{eqn:pure-dispersion}
1 - \frac{\var_G(\est)}{\var_{\giid}(\est)} = \dispg(\phi).
\end{equation}

Because $\var_G(\est) \geq 0$, we have $\dispg(\phi) \leq 1$ for all $\phi(\cdot)$ and $G \in \transportsym(F)$.
We also have the lower bound $\dispg(\phi) \ge -(k-1)$, which is attained for any $\phi(\cdot)$ under a clustered coupling with $D_i = D_j$ for all $i, j \in [k]$.
For exchangeable couplings, the dispersion can be interpreted as a normalized measure of negative correlation between the samples $\phi(\Di)$ and $\phi(\Dj)$ for $i \neq j$, as shown in the following proposition.

\begin{prop} \label{prop:dispersion-interpretation}
Let $0 < \varf(\phi) < \infty$ and $G \in \transportsym(F)$.
Then we have
\begin{equation} \label{eqn:dispersion-prop}
\dispg(\phi) = -(k-1)\corrg(\phi(\Di), \phi(\Dj)), \quad \quad i \not = j.
\end{equation}
\end{prop}

Negatively correlated samples tend to ``repel'' each other, making them more spread out.
We work primarily with this formulation in what follows.
In Section~\ref{section:efficiency-analysis}, we develop the technical machinery to describe exactly how $\dispg(\phi)$ is determined by the smoothness and shape of $\phi(\cdot)$.

\subsection{Match Quality} \label{subsection:match-quality}

The direct connection between dispersion and efficiency in Equation \eqref{eqn:pure-dispersion} above only holds for homogeneous problems with $\est = k\inv \sum_{i=1}^k \phi(\Di)$.
For realistic causal estimation problems, we also need to account for heterogeneity of the functions $\si(\cdot)$ within matched $k$-tuples of units, due to imperfect matching on covariates that are only partially predictive of influence-function heterogeneity.
To do so, next we define an appropriate measure of within-group match quality.

\begin{defn}[Match Quality]\label{defn:match-quality-coefficient}
Let $\viid(s) \equiv \en \varf(\si(D))$ denote the average marginal variance of the influence functions $\si(\cdot)$.
The match quality coefficient is defined as $\matchcoeffk(s) \equiv 1 - \vmatch(s)/\viid(s)$ for discrepancy
\begin{equation} \label{eqn:influence-discrepancy}
\vmatch(s) \equiv \frac{1}{2n(k-1)} \sum_g \sum_{i\neq j \in [k]} \varf(\sig(D) - \sjg(D)).
\end{equation}
We set $\matchcoeffk(s) \equiv 1$ by convention when $\viid(s) = 0$.
\end{defn}

The term $\vmatch(s)$ is a design-based matching discrepancy, with $\vmatch(s) = 0$ under perfect matching, $\sig = \sjg$ for all $i, j$ and groups $g$.
The match coefficient $\matchcoeffk(s)$ measures how homogeneous the functions $\sig(\cdot)$ are within each matched $k$-tuple, with $\matchcoeffk(s) = 1$ under perfect matching.
More generally, we have lower and upper bounds $-(k-1)\inv \le \matchcoeffk(s) \le 1$ for all populations $s = (\si)_{i=1}^n$ and matching procedures.
If units are matched at random, a short calculation shows that the expected match quality is $E_{\match}[\matchcoeffk(s)] \ge -(n-1)\inv$ in the worst case.
In theory, it is possible to approach the lower bound by purposefully matching units into $k$-tuples to be as dissimilar as possible, but we do not expect this to arise in practice.

\paragraph{Covariate Power.}

Recall units are matched into $k$-tuples using observed baseline covariates $(X_i)_{i=1}^n$.
The match quality coefficient $\matchcoeffk(s)$ will be large if both:
\begin{enumerate}[label=(\alph*), itemsep=0.1pt]
    \item Covariates $X_i$ are highly predictive of heterogeneity in $\si(\cdot)$.
    \item Matching discrepancy on $X_i$ within $k$-tuples (Equation \eqref{eqn:covariate-matching-discrepancy}) is small.
\end{enumerate}
Under standard iid sampling assumptions, one can show that $\matchcoeffk(s) = R^2_{s|X} \cdot \matchcoeffk(\mu) + \op(1)$ as $n \to \infty$, where $\matchcoeffk(\mu)$ is the match quality on features $\mu$ of the covariates alone and $R^2_{s|X}$ measures the predictive power of the covariates for the heterogeneity in $\si$.

%In principle, one could try to estimate such quantities using pilot data to decide exactly which covariates to collect and match on, similar to the approach in \cite{bai2020pairs}.
%This could work well if the experimental design already calls for a large pilot or several experimental waves.
%If not, it may not be worth running a whole extra experiment to estimate such quantities.
%We leave such optimized matching and multi-wave considerations to future work.

%\textbf{Practical Recommendations.}
%In this paper, we take a simple, practical approach to matching, focusing on the case of single-wave experiments.
%We recommend experimenters just match on a small set of covariates $X_i$ they believe will be highly predictive of important endline outcomes $Y_i$.
%One natural choice is to include baseline outcomes in $X_i$, as well as a few extra covariates expected to be highly predictive from researcher intuition, if available.
%The match quality on influence functions $\matchcoeffk(s)$ and corresponding efficiency gains will reflect whatever predictive power these covariates end up providing, without trying to optimize such quantities at design time when data is not yet available.

\subsection{Efficiency from Dispersion and Match Quality} \label{subsection:efficiency}

Our main theoretical result is that the efficiency gain from coupling designs is proportional to the product of dispersion $\dispg(\phi)$ and match quality $\matchcoeffk(s)$.
To state this result in full generality requires additional technical machinery, which we develop in Section \ref{section:efficiency-analysis} below.
To build intuition, we first state the result in a simple univariate parametric model.

\begin{thm}[Relative Efficiency] \label{thm:variance-parametric}
For $G \in \transportsym(F)$ and $0 < \varf(\phi) < \infty$, let $\si(d) = \ci + \ai \phi(d)$ for $i \in [n]$.
Then the variance of estimator $\est = \en[\si(\Di)]$ relative to the iid design is given by
\begin{equation} \label{eqn:parametric-variance}
1 - \frac{\var_G(\est)}{\var_{\giid}(\est)} = \dispg(\phi) \times \matchcoeffk(s).
\end{equation}
\end{thm}

For intuition, note that in the homogeneous special case $\si(d) = c + a \phi(d)$ for all units $i \in [n]$, relative efficiency is exactly given by $\dispg(\phi)$.
In general heterogeneous problems, the relative efficiency is dampened by imperfect matching, captured by the match quality coefficient $\matchcoeffk(s) < 1$.

Another interesting special case occurs when $\dispg(\phi) = 1$, so that relative efficiency is exactly equal to match quality $\matchcoeffk(s) = 1 - \vmatch(s)/\viid(s)$. 
Rearranging, we find $n\var_G(\est) = \vmatch(s)$, the matching discrepancy from Definition \ref{defn:match-quality-coefficient}. 
This shows how the matching discrepancy $\vmatch(s)$ can also be interpreted as the ideal variance under perfect dispersion, where efficiency is only limited by imperfect matching.

\textbf{Tuple Size Tradeoff.}
Dispersion $\dispg(\phi)$ is generally increasing in tuple size $k$, since for larger $k$ it becomes easier to jointly correlate the treatments $(\Di)_{i=1}^k$ to be spread out over $\suppf$, while match quality $\matchcoeffk(s)$ is generally decreasing in $k$, since it becomes harder to find many similar units to match together.
For general $F$ and $\phi(\cdot)$, perfect dispersion $\dispg(\phi) \to 1$ is only possible in the limit as $k \to \infty$, destroying match quality.
We study this tradeoff in detail in Sections \ref{section:efficiency-analysis} and \ref{section:coupling-analysis} below.

%Before moving on to our general efficiency analysis in Section \ref{section:efficiency-analysis}, next we briefly illustrate how classical notions of covariate balance can be extended to experiments with complex treatments $D \in \suppf$, showing how coupling designs prevent imbalances.

\paragraph{Illustration of Tradeoff.}
Consider a researcher estimating the dose-response of various welfare measures to the amount $D \ge 0$ of an unconditional cash grant, randomizing $D \sim F$ for $F = \text{Exp}(1)$ to provide many small grants while also trying some larger amounts.
Suppose the treatment has no effect on potential outcomes, so $\yi(d) = \one'X_i$ and $\si(d) = (\one'X_i) H(d)$. Then efficiency under a coupling $G$ is the product of $\dispg(H)$ and $\matchcoeffk(s)$ by Theorem~\ref{thm:variance-parametric}.
Figure~\ref{fig:frontier_efficiency} plots both quantities, the feasible frontier, and the resulting efficiency as functions of $k$ for the Latin hypercube (LHS), rotation sampling, and Gaussian copula couplings of Section~\ref{subsection:matching-couplings}.
Note that match quality is the same for all couplings at any given $k$.
The Gaussian copula, which is equivalent to antithetic variates at $k = 2$ (Remark~\ref{rem:gaussian-equivalence}), is the most efficient for small tuple sizes $k$; as $k$ increases, the LHS coupling overtakes it, producing higher dispersion and peaking at $k = 10$ with the highest efficiency; this reflects the optimal tradeoff in $k$ of reduced match quality for higher dispersion.

\begin{figure}[ht]
    \centering
    \includegraphics[width=\figwquad]{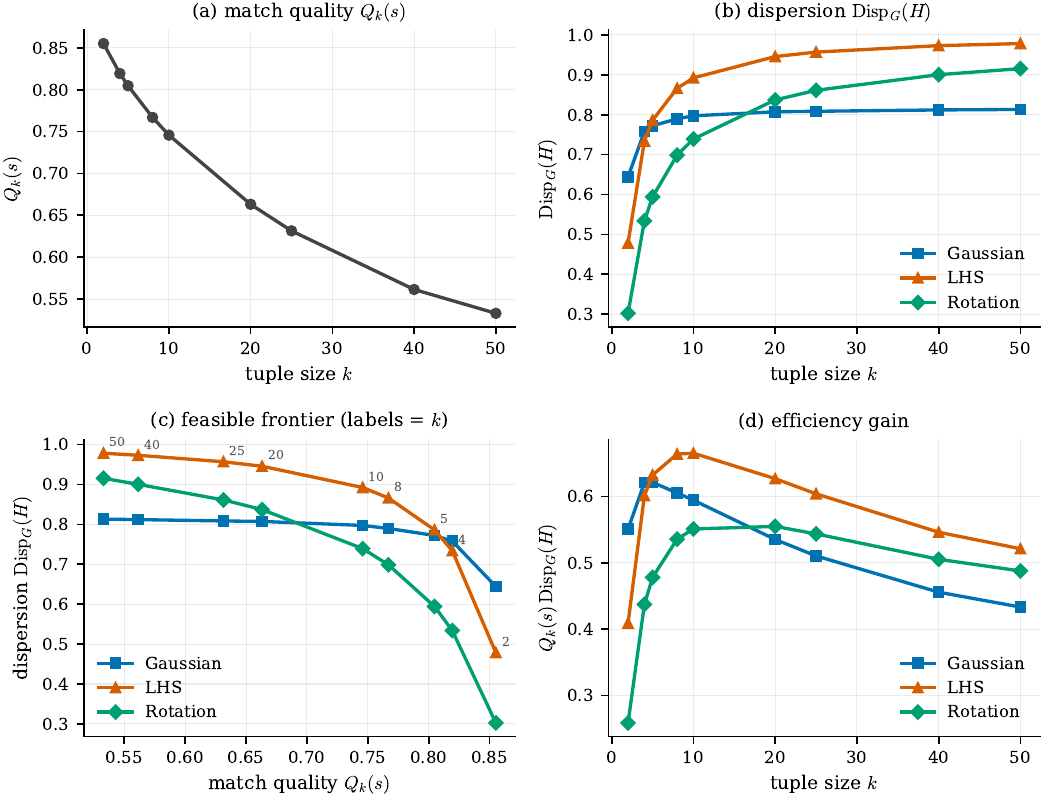}
    \caption{Tuple size tradeoff with $F = \text{Exp}(1)$, $n=1200$, and $\dim(X) = 4$.
    Panels show match quality (a), dispersion (b), the feasible (match quality, dispersion) frontier with labels indicating tuple size $k$ (c), and efficiency (d), for the Gaussian copula, Latin hypercube, and rotation sampling couplings.}
    \label{fig:frontier_efficiency}
\end{figure}

\begin{ex}[Stratified Randomization] \label{ex:stratified-randomization}
For discrete treatments, conventional stratified randomization is an example of a design that lexicographically emphasizes dispersion over match quality.
By completely randomizing within each stratum, stratified randomization ensures treatments are perfectly dispersed over the discrete space $\suppf$, achieving $\dispg(\phi) = 1$ for all $\phi(\cdot)$.
However, this is only possible for large enough stratum size $k$.
In particular, if $P(\Di = j) = p_j$, conventional stratified randomization is possible if and only if $p_j \cdot k \in \mathbb{N}$ for all $j$.  
In cases with many treatments or irregular assignment probabilities, this can impose a severe cost in match quality. 
Because efficiency is proportional to the product of both dispersion and match quality, such a lexicographic preference for dispersion will often be inefficient. 
See Appendix~\ref{subsection:dispersion-stratified} for a detailed discussion.
\end{ex}

\subsection{Covariate Balance for Complex Treatments} \label{subsection:covariate-balance-section}

An important motivation for conventional stratified randomization is that it reduces covariate imbalances between the different treatment groups.
These imbalances are equivalent to spurious in-sample correlations between the treatment $\Di$ and covariates $X_i$.
Here, we show that by assigning similar groups of units to highly dissimilar treatments, coupling designs prevent such spurious correlations from arising, enabling covariate-balancing randomization over complex treatment spaces.

For binary $\Di \in \{0, 1\}$, covariate balance is commonly assessed using the t-statistics from a regression $X_i \sim 1 + D_i$.
Up to a normalization, this is equivalent to checking the magnitude of the sample covariance $\cov_n(\Di, X_i)$.
If this covariance is small, then treatments are approximately independent of unit-specific heterogeneity in-sample.
Thus, ensuring covariate balance is equivalent to randomizing in a way that enforces $\eg[\cov_n(\Di, X_i)^2] \approx 0$.

To extend this balance measure to complex treatment spaces $\suppf$, let $\phi: \suppf \to \mr$ and $b: \mr^p \to \mr$ be basis functions.
Write $\var_n(b_i) = n\inv \sum_{i=1}^n (b_i - \bar b)^2$ for $b_i = b(X_i)$ and define the imbalance under a coupling $G$ as the mean-squared sample covariance 
\begin{equation} \label{eqn:covariate-balance}
\imb_G(\phi, b) \equiv \eg[ \, \cov_n(\phi(\Di), b(X_i))^2 \,].
\end{equation}

We can view $\Di$ as being approximately independent of covariates $X_i$ if $\imb_G(\phi, b) \approx 0$ for a rich set of basis functions $\phi(\cdot)$ and $b(\cdot)$.
In particular, the next corollary shows that the imbalance measure $\imb_G(\phi, b)$ is decreasing in the product $\dispg(\phi) \times \matchcoeffk(b)$ of dispersion and match quality.

\begin{cor}[Covariate Imbalance] \label{cor:covariate-balance}
If $0 < \varf(\phi) < \infty$ and $\var_n(b_i) > 0$, then the covariate imbalance under a coupling design $G$ relative to the iid design is 
\begin{equation}\label{eqn:covariate-balance-result}
\frac{\imb_G(\phi, b)}{\imb_{\giid}(\phi, b)} = 1- \dispg(\phi) \times \matchcoeffk(b).
\end{equation}
Here, the term $\matchcoeffk(b) \equiv 1 - (n/k)\inv \sum_g \vark(b_{ig}) / \var_n(b_i)$ denotes within-group match quality on $b_i = b(X_i)$.
\end{cor}

\section{Efficiency Theory} \label{section:efficiency-analysis}

%This technical analysis allows us to characterize how dispersion, and thus overall estimation efficiency, is driven by a tight alignment between the shape of the influence functions $\si(\cdot)$ and the high dispersion directions of the chosen coupling $G$.
%For example, in the case of $G = $ Latin hypercube, the high dispersion subspace consists of piecewise constant histogram functions on a certain partition of $\suppf$, with partition fineness increasing in $k$.
%If the influence functions $\si(\cdot)$ are sufficiently smooth, they will be well approximated on this subspace of histogram functions, leading to high precision of $\est = \en[\si(\Di)]$ under the coupling $G$.

\subsection{Dispersion Basis} \label{subsection:dispersion-basis}

Our main efficiency result holds in full generality, without the simplifying parametric assumption we imposed when previewing the results in Section \ref{section:dispersion-match-quality}.
To show this, we first develop the technical machinery to decompose the space of influence functions into a basis of orthogonal subspaces on which $\dispg(\cdot)$ is constant. We view these as the principal directions of the coupling $G$, and they are central to our efficiency analysis.

Define the square-integrable functions $\elltwo(F) = \{\phi: \ef[\phi(D)^2] < \infty\}$.
We assume $\si(\cdot) \in \elltwo(F)$ for $i \in [n]$ throughout.
Also denote the mean zero subspace $\ltwonot \equiv \{ \phi \in L^2(F) : \ef[\phi(D)] = 0 \}$.
We define the principal directions of a coupling $G$ to be the eigenspaces of the following linear operator, which captures the pairwise dependence structure of the design.

\begin{defn}[Coupling Operator] \label{defn:coupling-operator}
For $G \in \transportsym(F)$, let $U_G: L^2(F) \to L^2(F)$
\begin{equation} \label{eqn:coupling-operator}
(U_G \phi)(d) = E_G[\phi(D_i) \mid D_j = d \,], \quad \; i \neq j.
\end{equation}
\end{defn}

This operator is well-defined for any choice of $i \ne j$ due to exchangeability of $G$.
Since dispersion $\dispg(\phi)$ is invariant to constant shifts, it is convenient to work with the mean zero subspace $\ltwonot$ in what follows. 
Note that if $\phi \in \ltwonot$, then by tower law $\ef[(U_G \phi)(D)] = 0$, so $U_G$ also maps $\ltwonot$ to itself.
The coupling operator is self-adjoint and linear on the Hilbert space $\ltwonot$, so it has a real spectrum.  
We additionally require the following condition:  

\begin{assumption}[Discrete Spectrum] \label{assump:direct-sum}
There exist eigenspaces $(E_m)_{m \ge 1}$ of $U_G$ such that $\ltwonot = \oplus_{m \ge 1} E_m$.
\end{assumption}

Assumption \ref{assump:direct-sum} holds for any coupling $G \in \transportsym(F)$ if the marginal $F$ has finite support.
For continuous $F$, Lemma \ref{lem:direct-sum-satisfied} in the appendix shows that it holds for all of the univariate couplings described in Section \ref{subsection:matching-couplings}.
More generally, by the spectral theorem Assumption \ref{assump:direct-sum} holds whenever the operator $U_G$ is compact.

A key insight for our analysis is that $\dispg(\phi)$ decomposes orthogonally over the eigenspaces of $U_G$ for any $\phi \in L^2(F)$.

\begin{thm}[Dispersion Basis] \label{thm:dispersion-decomposition}
Let $G \in \transportsym(F)$.
\begin{enumerate}[label=(\alph*), itemsep=0.1pt]
\item If $E \subseteq \ltwonot$ is an eigenspace with $U_G \phi = \eval \phi$ for $\phi \in E$, then the dispersion $\dispg(\phi) = -(k-1)\lambda$ for all nonzero $\phi \in E$.
In particular, $\dispg(\phi) = \dispg(\psi)$ for all nonzero $\phi, \psi \in E$.
\item Impose Assumption \ref{assump:direct-sum}.
For $\phi \in L^2(F)$, let $P_m \phi = \argmin_{f \in E_m} \varf(\phi - f)$ be the orthogonal projection onto eigenspace $E_m$.
Then
\begin{equation}
\dispg(\phi) = \sum_{m \ge 1} \frac{\varf(P_m \phi)}{\varf(\phi)} \dispg(E_m).
\end{equation}
\end{enumerate}
\end{thm}

We view the eigenspaces of $U_G$ as the principal directions of the coupling $G$ in $\elltwo(F)$ with respect to sampling, since they control how much dispersion is produced when sampling from $G$.
In particular, $\dispg(\phi)$ will be large if $\phi(\cdot)$ is well approximated on the high dispersion eigenspaces, e.g.\ with $\dispg(E_m) \approx 1$.
The following corollary is immediate. 

\begin{cor}[Principal Directions] \label{cor:extremal-directions}
Suppose $\ltwonot = E \oplus E^\perp$ are eigenspaces of the operator $U_G$ with $\dispg(E) > \dispg(E^\perp)$.
Then
\begin{equation}
E = \argmax_{\substack{\phi \in \ltwonot \\ \phi \neq 0}} \, \dispg(\phi), \quad \quad E^\perp = \argmin_{\substack{\phi \in \ltwonot \\ \phi \neq 0}} \, \dispg(\phi).
\end{equation}
\end{cor}

More generally, if $\ltwonot = \oplus_{m=1}^M E_m$, then each eigenspace can be obtained by maximizing $\dispg(\phi)$ subject to orthogonality to previously found eigenspaces.
This is analogous to principal components analysis (PCA), where each principal component can be found by maximizing data variance subject to orthogonality to previous components.

\textbf{Coupling Analysis.} 
Theorem~\ref{thm:dispersion-decomposition} also provides a simple recipe for computing $\dispg(\phi)$ by analyzing the eigenspaces and eigenvalues of the operator $U_G$. 
We illustrate this by computing the exact dispersion for the univariate Latin hypercube coupling.
Our analysis highlights the role played by smoothness of the function $\phi(\cdot)$ in guaranteeing high dispersion $\dispg(\phi)$ under this coupling.
We provide a detailed comparison with other couplings in Section~\ref{section:coupling-analysis} below.

\begin{ex}[Latin Hypercube, Dispersion] \label{ex:lhs-dispersion}
Let $F = \unif[0, 1]$ and define the histogram space $\ehist \equiv \{ \phi \in \ltwonot : \phi(d) = \sum_l \alpha_l \cdot \one(d \in J(l)) \}$ on bins $J(l) = [(l-1)/k, l/k)$ for $l \in [k]$.
We show that $\ltwonot = \ehist \oplus \ehist^\perp$, eigenspaces of $U_G$ with dispersions $1$ and $0$ respectively (Lemma \ref{lem:lhs-operator}).
See Figure~\ref{fig:lhs_projections} for a graphical illustration of these spaces.
Theorem \ref{thm:dispersion-decomposition} implies that for $G = $ Latin hypercube sampling (LHS), we have:
\begin{equation} \label{eqn:lhs-dispersion-val}
\dispg(\phi) = \frac{\varf(\phist \phi)}{\varf(\phi)}.
\end{equation}
Here, $\phist: L^2(F) \to \ehist$ is the orthogonal projection onto $\ehist$, which is the de-meaned best histogram approximation of $\phi(\cdot)$:
\begin{equation} \label{eqn:histogram-projection-disp}
(\phist \phi)(d) = \sum_l \ef[\phi(D) | D \in J(l)] \cdot \one(d \in J(l)) - \ef[\phi(D)].
\end{equation}

Dispersion is large under $G = \text{LHS}$ to the extent that $\phi(\cdot)$ is well approximated by histograms on the partition $(J(l))_{l=1}^k$.
For fixed tuple size $k$, well-approximation on the space $\ehist = \ehist(k)$ requires the function $\phi(\cdot)$ to be relatively smooth.
As $k$ increases, the partition becomes finer and the histogram space is more expressive, so we obtain high dispersion even for ``rougher'' functions $\phi$.
In particular, consider tuple sizes $k \le r$ and suppose $k$ divides $r$, i.e.\ $k \, | \, r$.
Then $\ehist(k) \sub \ehist(r)$, so the projection weight on $\ehist(r)$ is larger.
It follows that dispersion is increasing in tuple size:
\begin{equation}
\disp_{G_k}(\phi) \le \disp_{G_r}(\phi) \quad \; \quad \text{$k \leq r$ with $k \, | \, r$.}
\end{equation}
\end{ex}

\begin{figure}[t]
    \centering
    \includegraphics[width=\figwproj]{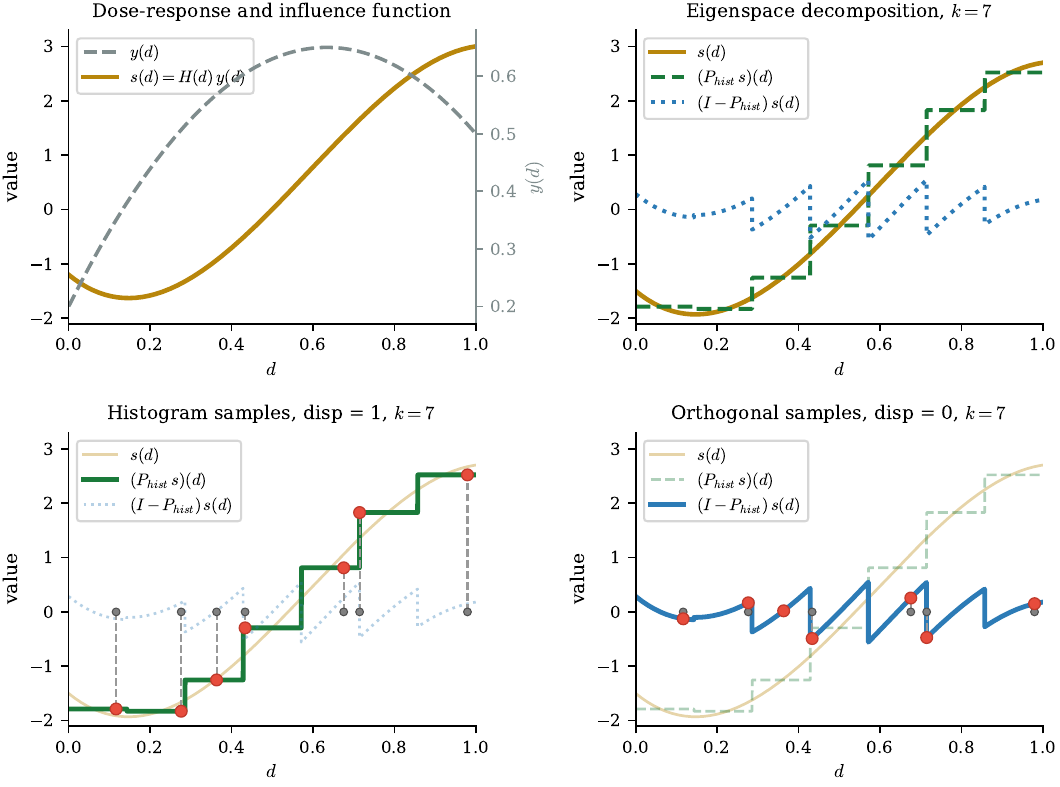}
    \caption{Eigenspace illustration for the LHS coupling with $k=7$, for the influence function $\si(d) = H(d)\, \yi(d)$ of $\thetablp$ (Example~\ref{ex:blp}) with $F = \unif[0,1]$.}
    \label{fig:lhs_projections}
\end{figure}

%\paragraph{Principal Directions.}

%Next, we show that the eigenspaces of $U_G$ are in fact the directions of maximal and minimal dispersion in $\ltwonot$.
%This motivates us to view these eigenspaces as the ``principal directions'' of $G$ in $\ltwonot$ for the purpose of random sampling.
%For $\ltwonot = E \oplus E^\perp$, without loss suppose the ordering $\dispg(E) \ge \dispg(E^\perp)$.
%In the corollary statement, we also implicitly let $\phi \ne 0$.

%\begin{cor} \label{cor:extremal-directions}
%If $\ltwonot = E \oplus E^\perp$ eigenspaces of $U_G$ with distinct eigenvalues,
%\begin{equation}
%E = \argmax_{\phi \in \ltwonot} \, \dispg(\phi), \quad \quad E^\perp = \argmin_{\phi \in \ltwonot} \, \dispg(\phi).
%\end{equation}
%\end{cor}

%More generally, if $\ltwonot = \oplus_{m=1}^M E_m$, then each eigenspace can be viewed as the set of directions that maximizes $\dispg(\phi)$ and is orthogonal to previously found eigenspaces.
%This is analogous to principal components analysis (PCA), where each principal component can be found by maximizing data variance subject to orthogonality to previous components.

\begin{remark}[Canonical Marginals]\label{rem:canonical-marginals}
It may appear that we need to separately solve for the eigenspace decomposition $\ltwonot = \oplus_{j \ge 1} E_j$ for each choice of marginal $F$ and coupling operator $U_G$.
In fact, it generally suffices to solve for the eigenspaces once for a canonical choice of marginal, e.g.\ $F = \unif[0, 1]^m$.
To see this, recall we constructed general $G \in \transportsym(F)$ by sampling high dispersion $(U_i)_{i=1}^k \sim G_U$ with each $U_i \sim \unif[0, 1]^m$, then setting $\Di = T(U_i)$ for a transport map $T: [0,1]^m \to \suppf$.
Then $\si(\Di) = \si(T(U_i)) \equiv \tilde \sfn_i(U_i)$ for the effective influence function $\tilde \sfn_i = \si \circ T$, so we can apply our results under the canonical marginal $\unif[0,1]^m$ with the influence functions $\si(\cdot)$ replaced by $\tilde \sfn_i(\cdot)$.
In particular, our analysis of LHS in Example \ref{ex:lhs-dispersion} with $F = \unif[0, 1]$ is without loss of generality.
\end{remark}

\subsection{Efficiency} \label{subsection:efficiency-theorem}

The dispersion basis in Theorem \ref{thm:dispersion-decomposition} allows us to generalize the simple relationship $\text{Efficiency} = \text{Dispersion} \times \text{Match Quality}$ from Theorem \ref{thm:variance-parametric} to general influence functions $\si(\cdot) \in L^2(F)$.
Impose Assumption \ref{assump:direct-sum} and let $\si^m(d) = (P_m \si)(d)$ be the projection of $\si(\cdot)$ onto $E_m$.
Define approximation weights
\begin{equation} \label{eqn:weight-def}
\weightm(s) \equiv \frac{n\inv \sum_i \varf(P_m \si)}{n\inv \sum_i \varf(\si)}.
\end{equation}
The weights $w_m(s)$ quantify how well the influence functions can be approximated using functions in eigenspace $E_m$, on average over the experimental units.
In what follows, denote $s = (\si)_{i=1}^n$ and $s^m = (\si^m)_{i=1}^n$.
We also write $\dispg(m)$ as the common dispersion on $E_m$.

\begin{thm}[Efficiency] \label{thm:eigenspace-decomposition}
Impose Assumption \ref{assump:direct-sum}.
Then for $G \in \transportsym(F)$
\begin{align} \label{eqn:efficiency-main}
1- \frac{\var_G(\est)}{\var_{\giid}(\est)} &= \sum_{m \ge 1} w_m(s) \cdot \disp_G(m) \matchcoeffk(s^m).
\end{align}
\end{thm}

The weights $w_m(s)$ are non-negative and sum to one.
Thus, the theorem shows that the efficiency gain from coupling design randomization is a convex combination of products of dispersion $\dispg(m)$ and match quality $\matchcoeffk(s^m)$ across eigenspaces $E_m$.
The efficiency gain is therefore large when the influence functions $\si(\cdot)$ align well with the high dispersion eigenspaces $\dispg(m) \approx 1$ and good match quality $\matchcoeffk(s^m) \approx 1$.
We show that many of the couplings introduced above achieve high dispersion generically for smooth functions $\si(\cdot)$.

\begin{ex}[Latin Hypercube, Efficiency] \label{ex:lhs-analysis}
As in Example \ref{ex:lhs-dispersion}, for $G = $ LHS, we have $\ltwonot = \ehist \oplus \ehist^\perp$ with dispersions $1$ and $0$.
Then by Theorem \ref{thm:eigenspace-decomposition},
\begin{equation} \label{eqn:lhs-variance}
1- \frac{\var_G(\est)}{\var_{\giid}(\est)} = \whist(s) \cdot \matchcoeffk(s^{hist}).
\end{equation}
In particular, LHS is efficient when the influence functions $\si(\cdot)$ are smooth enough to be well-approximated on the histogram space $\ehist$, and when units are well-matched on the projection $s^{hist}$ of the influence functions onto this space.
Equation \eqref{eqn:lhs-variance} also shows the dispersion versus match quality tradeoff discussed in Section \ref{subsection:efficiency}.
As tuple size $k$ grows, the space $\ehist(k)$ becomes more expressive and $\whist(s) = \whist^k(s)$ generally increases, with $\whist^k(s) \le \whist^r(s)$ whenever $k \mid r$ as in Example~\ref{ex:lhs-dispersion}.
By contrast, match quality $\matchcoeffk(s^{hist})$ is generally decreasing in $k$.
\end{ex}

\paragraph{Approximate Stratification.}

Recall the definition of the match quality coefficient $\matchcoeffk(s) = 1 - \vmatch(s) / \viid(s)$ from Definition \ref{defn:match-quality-coefficient}, where $\viid(s)$ is the iid variance and $\vmatch(s)$ is the average within-group variance of the influence functions.
When units are well-matched, we have $\matchcoeffk(s) \approx 1$ so that $\vmatch(s) \ll \viid(s)$.

\begin{cor}[Nominal Variance] \label{cor:efficiency-convex}
Under the conditions of Theorem \ref{thm:eigenspace-decomposition},
\begin{equation} \label{eqn:efficiency-convex}
n \var_G(\est) = \sum_{m \ge 1}  \dispg(m) \cdot \vmatch(s^m) + [1- \dispg(m)] \cdot \viid(s^m).
\end{equation}
\end{cor}

If $\dispg(m) = 1$, we obtain the variance $\vmatch(s^m)$.
For simple discrete treatments, $\dispg(\cdot) = 1$ for classic stratified randomization for large enough tuple size $k$ (Example \ref{ex:stratified-randomization}), so we can regard $\vmatch(s)$ as the perfectly stratified variance.
For general complex treatment spaces $\suppf \sub \mr^m$, perfect stratification is impossible, but the corollary shows a sense in which we can approximate it by using high dispersion couplings.

\begin{ex}[Latin Hypercube, Approximate Stratification] \label{ex:lhs-analysis-2}
Applying Corollary \ref{cor:efficiency-convex} to $G = $ LHS, we obtain $n \var_G(\est) = \vmatch(s^{hist}) + \viid(s - s^{hist})$.
%\begin{equation} %\label{eqn:lhs-variance-2}
%n \var_G(\est) = \vmatch(s^{hist}) + \viid(s - s^{hist}).
%\end{equation}
The LHS coupling behaves like perfect stratification $\vmatch(s^{hist})$ on the histogram projection $\si^{hist}(\cdot)$ of each influence function, but behaves like iid randomization on the residual $(\si - \si^{hist})(\cdot)$.
See Figure \ref{fig:lhs_projections} for a visualization of this effect. 
\end{ex}

\section{Coupling Analysis and Comparisons} \label{section:coupling-analysis}

%The analysis shows that RS also generically produces high dispersion for large enough $k$, but is less robust to adversarial influence function shapes $\si(\cdot)$ than LHS.
%By contrast, for moderate $k$ the Gaussian copula only produces high dispersion for approximately linear functions, a strong parametric restriction.

\subsection{Rotation Sampling} \label{subsection:rs}

We illustrate how the theory developed in the previous section can be applied to compare the efficiency and robustness of designs based on LHS with designs based on rotation sampling (RS) and the Gaussian copula, starting with rotation sampling.

Recall from Example \ref{ex:shifted-lattice} that a rotation sample $(U_i)_{i=1}^k \sim G_U$ lies on a randomly shifted equispaced grid, so that $U_i = U_j \oplus l/k$ for some $l \in [k]$ and $i, j \in [k]$, where $a \oplus b \equiv a + b \pmod 1$ for $a, b \in \mr$.
The canonical marginal for univariate rotation sampling is thus $F = \unif[0, 1]$.
We show that this coupling is efficient if influence functions $\si(\cdot)$ are smooth, in a sense defined below.

Suppose a function $\phi(\cdot)$ is perfectly $1/k$-cyclic, with $\phi(x) = \phi(x \oplus 1/k)$ for all $x \in [0, 1]$.
Then under rotation sampling, the samples $\phi(U_i) = \phi(U_j)$ for $i, j \in [k]$ are perfectly correlated, effectively yielding a clustered sample.
The low dispersion eigenspace turns out to be exactly this space of perfectly cyclic functions:
\begin{equation} \label{eqn:cyclic-subspace}
E_{c} \equiv \{\phi \in \ltwonot: \phi(x) = \phi(x \oplus 1/k), \forall x \in [0, 1] \}.
\end{equation}

The acyclic subspace is defined as the orthogonal complement $E_{a} \equiv E_{c}^\perp$ in $\ltwonot$.
Our analysis shows that $\ltwonot = E_c \oplus E_a$, which are eigenspaces of the coupling operator $U_G$ with dispersions $\disp_G(E_a) = 1$ and $\disp_G(E_c) = -(k-1)$.
Projections on $E_c$ and $E_a$ have closed forms, with $P_c \phi(d) = k\inv \sum_{l=1}^k \phi(d \oplus l/k) - \ef[\phi(D)]$ the de-meaned cyclic average and residual $P_a \phi(d) = \phi(d) - k\inv \sum_{l=1}^k \phi(d \oplus l/k)$.
%See Figure~\ref{fig:rs_projections} for an illustration of these projections.
Then by Theorem \ref{thm:dispersion-decomposition}, for any $\phi \in \elltwo(F)$,
\begin{equation} \label{eqn:dispersion-rs}
\disp_G(\phi) = \frac{\varf(P_a \phi)}{\varf(\phi)} - (k-1) \frac{\varf(P_c \phi)}{\varf(\phi)}.
\end{equation}

Recall that for LHS, the low dispersion subspace $\ehist^\perp$ behaves like an iid design with $\dispg(\ehist^\perp) = 0$ (Example \ref{ex:lhs-analysis}).
By contrast, the low dispersion space under RS has $\dispg(E_c) = -(k-1)$, actually harming relative efficiency through the negative term in Equation \eqref{eqn:dispersion-rs}.
Thus, the RS coupling is less robust to adversarial influence function shapes than LHS.
However, for this worst case to arise in practice, the influence functions $\si(\cdot)$ must be strongly cyclic with high frequency, which will be rare in typical social science applications.

\begin{remark}[Smoothness and Shape]
Equation \eqref{eqn:dispersion-rs} shows that $\dispg(\phi)$ is large under the RS coupling if $\phi(\cdot)$ is well-approximated on the acyclic space $E_a$.
This amounts to a combination of smoothness and shape restrictions on $\phi(\cdot)$.
To see this, expand $E_c$ in a Fourier basis, noting $E_c = \linearspan \{ \sin(2\pi l k x), \cos(2\pi l k x) : l \ge 1 \}$ consists of frequencies that are exact multiples of $k$, while $E_a$ contains all remaining frequencies.
Since $E_a$ includes all frequencies lower than $k$, smooth functions are well-approximated on $E_a$ for moderate $k$.
As $k$ increases, the cyclic space $E_c$ shrinks, reducing the degree of smoothness needed for high dispersion.
See also \cite{lecuyer2000variance} for a related Fourier analysis perspective on the variance of (non-exchangeable) shifted lattice designs.
\end{remark}

Let $w_a(s)$ and $w_c(s)$ denote the approximation weights (Equation \eqref{eqn:weight-def}) on $E_a$ and $E_c$ respectively, so that $w_c(s)$ quantifies how cyclic the influence functions $\si(\cdot)$ are, on average over $i \in [n]$.
We apply Theorem \ref{thm:eigenspace-decomposition} to $G =$ RS, using the facts above.

\begin{thm}[Rotation Sampling]\label{thm:rs-variance}
Let $G$ be the RS coupling.
Then
\begin{equation} \label{eqn:efficiency-rs}
1 - \frac{\var_G(\est)}{\var_{\giid}(\est)} = w_a(s) \cdot \matchcoeffk(s^a) - w_c(s) \cdot (k-1) \cdot \matchcoeffk(s^c).
\end{equation}
\end{thm}
The low dispersion cyclic space $E_c = E_c(k)$, where the design performs poorly, typically shrinks as $k$ increases.
In particular, we have $E_c(r) \sub E_c(k)$ for $k \mid r$, so the weights satisfy $w_c^r(s) \le w_c^k(s)$.
This shows a sense in which dispersion is generally increasing in tuple size $k$ for rotation sampling.

\subsection{Gaussian Copula} \label{subsection:gaussian}

We show that when $k$ is moderate, the Gaussian copula produces high dispersion only for approximately linear influence functions $\si(\cdot)$, a restrictive condition relative to the LHS and RS couplings, which produce high dispersion generically for smooth enough functions.
This suggests caution when using the Gaussian copula to generate dispersion in experimental design.

For the Gaussian copula, it is convenient to use canonical measure $F = \normal(0, 1)$.
In this case, the coupling operator $U_G$ coincides with the Mehler kernel operator \citep{mehler1866}.
We use the eigenbasis expansion $\ltwonot = \oplus_{m \ge 1} \linearspan(\hm)$ for $U_G$, where $(\hm)_{m \ge 1}$ are the normalized probabilist's Hermite polynomials \citep{thangavelu1993}.
Each $\hm(x)$ is a polynomial of order $m$, for example, $h_1(x) = x$ and $h_2(x) = (x^2 - 1)/\sqrt{2}$.
As shown in Theorem~\ref{thm:dispersion-decomposition}, the dispersion of a polynomial $h_m$ can be obtained from its eigenvalue $\lambda_m$ under $U_G$. Thus one calculates that $\dispg(\hm) = [-1/(k-1)]^{m-1}$.
The projections onto $E_m = \linearspan(\hm)$ are $\si^m(\cdot) = \covf(\si, \hm) \cdot \hm(\cdot)$.
An application of Theorem \ref{thm:eigenspace-decomposition} yields the following result.

\begin{thm}[Gaussian]\label{thm:gaussian-efficiency}
Let $G$ be the Gaussian coupling.
For $k \ge 3$
\begin{equation}
1 - \frac{\var_G(\est)}{\var_{\giid}(\est)} = \sum_{m \ge 1} w_m(s) \cdot [-1/(k-1)]^{m-1} \cdot \matchcoeffk(s^m).
\end{equation}
\end{thm}

Define the space of linear functions $E_L = \{\phi : \phi(d) = a + bd\}$.
For any non-constant $\phi \in E_L$, we have $\dispg(\phi) = \dispg(h_1) = 1$, so this is a high dispersion subspace.
By contrast, $|\dispg(\hm)| \le (k-1)^{-(m-1)}$ for any Hermite polynomial of order $m \ge 2$, which is rapidly decreasing as tuple size $k$ increases.
For larger $k$, the Gaussian copula produces high dispersion only for the linear component of $\si(\cdot)$, performing no better than iid randomization on $E_L^\perp$.

\begin{cor}[Linear Influences] \label{cor:gaussian-limit}
Let $\si^L(\cdot)$ denote the orthogonal projection of $\si(\cdot)$ onto $E_L$ in $\elltwo(F)$.
The variance $n \var_G(\est) = \vmatch(s^L) + \viid(s - s^L) + O(k\inv)$ and
\begin{equation}
1 - \frac{\var_G(\est)}{\var_{\giid}(\est)} = w_L \cdot \matchcoeffk(s^L) + O(k\inv).
\end{equation}
\end{cor}

Corollary \ref{cor:gaussian-limit} shows that, for moderate $k$, the Gaussian copula is efficient only if the influence functions $\si(\cdot)$ are approximately linear.
Note that $\si(\cdot)$ may not be linear even if the potential outcomes $\yi(\cdot)$ are linear in the treatment, since each $\si(\cdot)$ also depends on the estimand and the design.
For example, in the cash transfer experiment of Example \ref{ex:continuous-dose}, if the responses are linear in the grant amount, $\yi(d) = a_i + b_i d$, the influence function for the best linear approximation coefficient (Example \ref{ex:blp}) is $\si(d) = \yi(d)H(d)$ for $H(d)=(d-\ef[D]) / \varf(D)$, which is quadratic in $d$.

\begin{remark}[Antithetic Variates]
The results above characterized the performance of the Gaussian copula for moderate to large $k$.
At the opposite extreme $k=2$, the Gaussian copula is equivalent to antithetic variates.
In Appendix \ref{subsection:av}, we apply our core efficiency results to study antithetic variates in more detail, extending the original discussion in Section \ref{paragraph:antithetic-variates}.
\end{remark}

\subsection{Parametric vs.\ Nonparametric Couplings} \label{subsection:coupling-comparisons}

Our analysis of the various couplings above suggests that they can be sorted into two categories: non-parametric couplings like LHS and RS, which produce high dispersion under weak smoothness conditions on $\si(\cdot)$, and parametric couplings like the Gaussian copula, which produce high dispersion only for a restricted class of influence functions with specific shapes.
Here we formalize this claim.

For any $\phi: [0,1] \to \mr$, define \emph{total variation} $V_{[0,1]}(\phi) = \sup_\Pi \sum_{j=1}^r |\phi(t_j) - \phi(t_{j-1})|$, where the supremum is over all finite partitions $(t_j)_{j=0}^r$ of $[0,1]$, and similarly $V_J(\phi)$ over any subinterval $J \sub [0,1]$.
Define the class of \emph{bounded variation} functions $\mc H(b) \equiv \{\phi: V_{[0,1]}(\phi) \le b\}$.
This provides a weak notion of smoothness for functions on $[0,1]$, allowing for discontinuous functions (e.g.\ histograms) provided they don't oscillate too much over the interval $[0, 1]$.
This condition is meaningful for both discrete and continuous treatments.
For example, if $F$ is discrete, the effective influence function $\sitilde(\cdot)$ on $[0, 1]$ defined by $\si(D) = \si(F\inv(U)) = \sitilde(U)$ will be discontinuous, but may have small total variation.

\begin{thm}[Dispersion Limits] \label{thm:dispersion-limits}
Let $\phi \in L^2(F)$.
As $k \to \infty$, if $G = $ Gaussian, then $\dispg(\phi) \to \varf(P_L \phi) / \varf(\phi)$ for canonical marginal $F = \normal(0, 1)$.
If $G \in \{\text{LHS}, \text{RS}\}$ with canonical marginal $F = \unif[0, 1]$, then for any $\eps > 0$,
\begin{equation} \label{eqn:dispersion-limit-uniform}
\inf_{\substack{\phi \in \mc H(b) \\ \varf(\phi) > \eps}} \dispg(\phi) = 1 + o(1).
\end{equation}
\end{thm}

The theorem shows that the Gaussian copula produces high dispersion only for approximately linear functions $\phi(\cdot)$, while the non-parametric LHS and RS couplings produce high dispersion for any function $\phi(\cdot)$ with bounded total variation.

To relate the smoothness of the influence functions $\si(\cdot)$ to the variance of the estimator $\est$, define a smoothness coefficient $\etatv(s) \equiv  \en[V_{[0,1]}(\si)^2] / \viid(s)$, which is a normalized measure of average total variation of the $\si(\cdot)$.

\begin{thm}[Efficiency from Smoothness]\label{thm:structured-worst-case}
Let $G \in \{\text{LHS}, \text{RS}\}$ with canonical marginal $F = \unif[0, 1]$.
Then
\begin{equation}\label{eqn:smoothness-efficiency-bound}
1 - \frac{\var_G(\est)}{\var_{\giid}(\est)} \;\geq\; \matchcoeffk(s) - \frac{\etatv(s)}{k}.
\end{equation}
\end{thm}

In particular, the theorem guarantees that LHS and RS improve precision relative to iid randomization when $\matchcoeffk(s) \ge \etatv(s) / k$, so that the match quality is sufficiently high relative to the roughness of the influence functions $\si(\cdot)$.
This lower bound on efficiency can be improved by using stronger notions of smoothness.
For example, by imposing uniform Lipschitz continuity on the $\si(\cdot)$, the lower bound for LHS can be improved to $\matchcoeffk(s) - O(1/k^2)$.
However, such conditions typically rule out discrete treatments, motivating the weaker total variation notion of smoothness that we use here.

\section{Asymptotics and Inference}\label{section:asymptotics}

\subsection{Consistency}

We consider an asymptotic regime with a sequence of experimental populations and coupling designs indexed by $n$ with $n \to \infty$.
Thus, in this section all variables are implicitly indexed by $n$, denoting their place in the sequence.
For example, influence functions $\si(\cdot) = \si^n(\cdot)$ for $i \in [n]$, and similarly the design parameters $G = G(n)$ and $k = k(n)$.
We often suppress the indexing for brevity.

%Coupling designs are meant to improve efficiency relative to iid randomization.
%The iid design is $\rootn$-consistent under weak moment conditions, so if a coupling design is to be useful, it must not perform worse than this.
%While $\rootn$-consistency is a low bar, the experimental design literature has noted a trade-off between potential efficiency gains from balancing and robustness \citep{efron1971coin,harshaw2024}, which can manifest as slower rates of consistency.
%The results in Section~\ref{subsection:coupling-comparisons} implies that the nonparametric couplings considered there perform well under weak smoothness conditions.
%We here provide a more general characterization of consistency for coupling designs.

The experimental design literature has noted an important tradeoff between efficiency and robustness, showing how balancing covariates to improve efficiency can entail a loss of robustness in adverse experimental settings \citep{efron1971coin,harshaw2024}.
Motivated by this, we study a robust notion of consistency, requiring that coupling designs perform well uniformly over a range of possible empirical settings \citep{harshaw2025general}.
Given a family of influence functions $\mathcal{S}$, the uniform mean square error $\maxraten(G, \mc S)$ for a coupling $G$ is
\begin{equation} \label{eqn:uniform-consistency-rate}
\maxraten(G, \mathcal{S}) \equiv \sup_{s \in \mathcal{S}} n \eg\big[(\est - \thetan)^2\big].
\end{equation}
If $\maxraten(G, \mathcal{S}) = O(1)$, we say that $\est$ is \emph{uniformly $\rootn$-consistent} under the family $\mathcal{S}$ and the coupling sequence $G(n)$.

We can study the worst-case performance of our designs by requiring uniform consistency under weak restrictions on $\mc S$.
For example, we can set $\smom$ to be the set of all $\si(\cdot)$ with bounded average moments: $\smom = \{ s_1, \dotsc, s_n \in L^2(F) : \en[\varf(\si)] \le 1 \}$.
This allows for settings with highly non-smooth influence functions and negative match quality, requiring good performance even if an adversary matched units into groups that are maximally dissimilar. 
We show that coupling designs are still uniformly $\rootn$-consistent over $\smom$ under weak conditions on the design parameters, providing a strong robustness guarantee.

\begin{thm} \label{thm:uniform-consistency-rates}
Let $\infdispg = \inf_{\phi \neq c} \dispg(\phi)$ and $\supdispg = \sup_{\phi \neq c} \dispg(\phi)$ be the extremal dispersions of $G$ over non-constant $\phi \in \elltwo(F)$.
Impose Assumption \ref{assump:direct-sum}.
For any $G \in \transportsym(F)$,
\begin{equation} \label{eqn:maxrate-sup-inf}
\maxraten(G, \smom) = 1 + \max \left(-\infdispg, \frac{\supdispg}{k-1} \right).
\end{equation}
In particular, $1 \le \maxraten(G, \smom) \le k$.
Thus, if $k = O(1)$, then uniform $\rootn$-consistency is attained for any coupling sequence $G(n)$.
\end{thm}

\begin{remark}[Minimaxity] \label{rem:minimaxity}
We have $\maxraten(\giid, \smom) = 1$, so the iid design is minimax optimal over $\smom$.
Suppose $F$ is continuous.
If $G \in \{\text{LHS}, \text{Gaussian}\}$, then $\maxraten(G, \smom) = k/(k-1) \to 1$ as $k \to \infty$, so these couplings are asymptotically minimax optimal over $\smom$.
Note also that if $G = $ RS then $\maxraten(G, \smom) = k$, attaining the worst case rate.
This reflects the fact that rotation sampling can be exploited by highly non-smooth, perfectly cyclic influence functions $\si(\cdot)$.
\end{remark}

The influence function configurations that attain the worst case rate in the proof of Theorem \ref{thm:uniform-consistency-rates} are generally pathological, with either highly non-smooth influences $\si(\cdot)$ or negative match quality, which might not be relevant for empirical practice.
For example, for $G = $ LHS, $\maxraten(G, \smom)$ is attained by placing perfectly negatively correlated influence functions $\si(\cdot)$ within each group, which assumes that our matching was not only ineffectual, but actually much worse than random.
We can place mild regularity conditions on the set $\smom$ to rule out such pathological settings, obtaining bounds that are more informative about the performance of coupling designs in typical applications.

To illustrate this, we study the nonparametric couplings $G \in \{\text{LHS}, \text{RS}\}$ under a restricted family of influence functions $\sreg$ that have reasonable match quality and bounded average total variation.
The next result shows that under such conditions, these couplings uniformly dominate the iid design. 

\begin{thm} \label{thm:consistency-under-regularity}
Let $\sreg = \{s \in \smom : \matchcoeffk(s) \ge q_0 > 0 \;\text{ and }\;  \etatv(s) \le \bar \eta \}$ be the family of influence functions with match quality at least $q_0$ and smoothness coefficient at most $\bar \eta$, with $\etatv$ measured for the effective influence functions $\si \circ F\inv$ on $[0,1]$ (Remark~\ref{rem:canonical-marginals}).
Impose Assumption \ref{assump:direct-sum}.
Then for $G \in \{\text{LHS}, \text{RS}\}$,
\begin{equation} \label{eqn:restricted-uniform-rate}
\maxraten(G, \sreg) \le 1 - q_0 + \bar \eta / k.
\end{equation}
\end{thm}

For the iid coupling $\giid$, we have $\maxraten(\giid, \sreg) = 1$, so the LHS and RS coupling designs dominate the iid design in a minimax sense over $\sreg$ when $\bar \eta / k < q_0$, which holds for large enough $k$.
Intuitively, this is because iid randomization has no way of exploiting the structure imposed on the influence functions by $\sreg$.
Moreover, after ruling out arbitrarily non-smooth functions $\si(\cdot)$, rotation sampling $G=$ RS is uniformly $\rootn$-consistent over $\sreg$ even as $k \to \infty$.

\subsection{Asymptotic Normality}

We provide conditions under which the estimator $\est$ is asymptotically normal.

\begin{assumption}[CLT]\label{assumption:clt}
The following hold:
\begin{enumerate}[label=(\arabic*)., itemsep=0.1pt]
\item (Bounded Fourth Moments) $M_{4,n} = \en\bigl[\ef[\si(D)^4]\bigr] = O(1)$.
\item (Not Superefficient) $\eg[(\est - \thetan)^2] = \Omega(n^{-1})$.
\item (Group Size) $k = o(n^{1/3})$.
\end{enumerate}
\end{assumption}

Part (2) of Assumption \ref{assumption:clt} is a high level condition ruling out certain degenerate estimation problems.
For example, under perfect homogeneity $\si(d) = s(d)$ with $s(\cdot)$ Lipschitz continuous, one can show that super-efficient estimation is possible using the Latin hypercube coupling with $k=n$.
Such perfectly homogeneous settings are not empirically relevant.

\begin{prop}[CLT]\label{prop:clt}
Impose Assumption~\ref{assumption:clt}.
Then for $\hksd_n^2 = \var_G(\est)$,
\begin{equation} \label{eqn:clt}
(\est - \thetan) / \hksd_n \convd \normal(0, 1),
\end{equation}
\end{prop}

Since each of the $n/k$ groups has independent treatments, and each group's contribution to $\est$ is bounded, the CLT follows from standard Lindeberg-Feller arguments provided the number of groups $n/k$ grows sufficiently fast.
Assumption~\ref{assumption:clt} ensures this by requiring $k = o(n^{1/3})$.

Combining Proposition~\ref{prop:clt} with the linearizations of Section~\ref{subsection:causal-estimands} shows that more general parametric estimators are also asymptotically normal under coupling designs.

\begin{cor}[Parametric Estimators]\label{cor:parametric-clt}
Suppose $\wh\beta - \betan = \en[\si(\Di)] + \Op(n\inv)$ with $\en \ef[\si(D)] = 0$. This holds for the OLS and logit estimators of Section~\ref{subsection:causal-estimands} under the regularity conditions of Propositions~\ref{prop:ols-linearization} and~\ref{prop:logit-linearization}.
Fix $c \in \mr^m$ and impose Assumption~\ref{assumption:clt} for the influence functions $(c'\si)_{i=1}^n$.
Then for $\hksd_n^2 = \var_G(\en[c'\si(\Di)])$,
\begin{equation} \label{eqn:parametric-clt}
c'(\wh\beta - \betan) / \hksd_n \convd \normal(0, 1).
\end{equation}
\end{cor}

\subsection{Variance Estimation and Inference}

We construct a variance estimator for $\var_G(\est)$ using a collapsed strata approach \citep{hansen1953}.
See also \citet{bai2025finely} for a recent application of this approach to classic stratified randomization.
Let $\matchvar: [n / k] \to [n / k]$ be a permutation with no fixed points, $\matchvar(g) \neq g$ for all $g \in [n / k]$, which associates each group $g$ with a paired group $\matchvar(g)$.
Typically, this is a matching $\matchvar(\matchvar(g)) = g$, but we allow for non-matching permutations to accommodate, for example, an odd number of groups.

Let $\est_{g} = k^{-1} \sum_{i = 1}^{k} \sig(\Di)$ and $\theta_{g} = \ef[\est_g] = k^{-1} \sum_{i = 1}^{k} \theta_{ig}$ for $\theta_{ig} = \ef[\sig(D)]$.
The variance estimator is the scaled squared difference between the mean outcomes in the paired groups:
\begin{equation} \label{eqn:variance-estimator}
\widehat{\hksd}_n^2 = \frac{k^2}{2n^2} \sum_{g = 1}^{n / k} \big(\est_{g} - \est_{\matchvar(g)}\big)^2.
\end{equation}
Let $\Delta_n^2 = (n/k)\inv \sum_{g = 1}^{n/k} (\theta_{g} - \theta_{\matchvar(g)})^2$ be the average squared difference in group effects across the paired groups.
Then we have the following result.

\begin{prop}\label{prop:var-est-expectation-variance}
The variance estimator has expectation and variance given by
\begin{equation} \label{eqn:var-est-moments}
\eg\biggl[\frac{\widehat{\hksd}_n^2}{\hksd_n^2}\biggr] = 1 + \frac{k \Delta_n^2}{2 n \hksd_n^2},
\qquad\quad
\var_G\biggl(\frac{\widehat{\hksd}_n^2}{\hksd_n^2}\biggr) \leq \frac{12 k^3}{n^3 \hksd_n^4} M_{4,n}.
\end{equation}
Then
\begin{enumerate}[label=(\alph*), itemsep=0.1pt]
\item The variance estimator is conservative: $\inf_{n} \eg[\widehat{\hksd}_n^2 / \hksd_n^2] \geq 1$.
\item The variance of the normalized variance estimator vanishes given Assumption~\ref{assumption:clt}: $\var_G(\widehat{\hksd}_n^2 / \hksd_n^2) = o(1)$.
\end{enumerate}
\end{prop}

The magnitude of the bias depends on the tuple size $k$ and the heterogeneity in group effects.
Under $\rootn$-consistency, $\hksd_n^2 \asymp n^{-1}$, the normalized bias is $O(k \Delta_n^2)$.
We generally expect $\Delta_n^2$ to be asymptotically bounded but not to vanish, so the variance estimator is typically biased upwards also asymptotically.
Because the variance estimator is conservative, confidence intervals constructed using $\widehat{\hksd}_n^2$ are asymptotically valid.

\begin{prop}\label{prop:ci}
Let $\widehat{\hksd}_n = \sqrt{\widehat{\hksd}_n^2}$ be the standard error estimator, and impose Assumption~\ref{assumption:clt}.
Then, for any $\alpha \in (0, 1)$, the confidence interval
$$
\mathrm{CI}_{1 - \alpha} = \big[\est - z_{1 - \alpha/2}\, \widehat{\hksd}_n,\;\; \est + z_{1 - \alpha/2}\, \widehat{\hksd}_n\big],
$$
where $z_{1 - \alpha/2}$ is the $(1 - \alpha/2)$-quantile of the standard normal distribution, satisfies
$$
\liminf_{n \to \infty} \pg(\thetan \in \mathrm{CI}_{1 - \alpha}) \geq 1 - \alpha.
$$
\end{prop}

\section{Application: Savings Monitors in Village Networks}\label{section:illustrative-applications}\label{subsection:savings-monitors-application}

% Cash-transfer application (formerly Section 8.1) removed in Pass 07 to meet the page limit.
% Restore by uncommenting the following line.
%\input{inputs/paper_application_cash}

We illustrate the theory in an application based on the savings-monitor experiment of \citet{breza2019}, conducted in a set of villages in Karnataka, India.
Each saver in the experiment is paired with a ``monitor'' who is periodically informed of the saver's progress toward a self-set savings goal, creating a reputational commitment device.
\citet{breza2019} show that who the monitor is matters: pairing a saver with a monitor raises savings by over a third on average, and the effect is larger when the monitor is more central in the village social network and more socially proximate to the saver.
The goal in our application is precisely to learn how a saver's behavior responds to their monitor's social position.
We summarize each candidate monitor by a vector of network-position features, so that the treatment is a featurized monitor drawn from a discrete catalog $\suppf = \{d_1, \dots, d_R\} \sub \mr^m$, the setting of Section~\ref{ex:irregular-uber-eats}.
Monitors are assigned within a saver's own village, so each village supplies its own catalog, an irregular point cloud with no product structure.

\paragraph{A catalog of monitors from real networks.}
We construct the catalogs from the social-network data of \citet{banerjee2013} for $75$ Karnataka villages, the same data used by \citet{breza2019}, which record household-to-household ties formed by the union of several types of relationships (borrowing, lending, advice, kinship, and others).
We treat each household with at least one tie as a candidate monitor and featurize it by two standard statistics of its network position, degree centrality (its number of ties, on a logarithmic scale) and clustering coefficient (the fraction of pairs of its neighbors that are themselves tied).
Each feature is normalized to the unit interval over the pooled households of all villages, so that a given point in feature space represents similar network positions in every village.
Villages range from $77$ to $356$ households and differ markedly in the shape and location of the resulting point cloud in $[0,1]^2$.
From each village we draw $R = 64$ candidate monitors uniformly at random (Figure~\ref{fig:net_setup}, left column) and endow the catalog with an inverse-density target marginal $F$ that makes the design roughly equally likely to assign a monitor from any populated region of the feature space.
Because the support is non-convex and the two features are dependent, implementing $F$ requires the semi-discrete optimal-transport map of Section~\ref{subsection:transport} rather than a coordinatewise quantile transform.

\begin{figure}[ht]
\centering
\includegraphics[width=\textwidth]{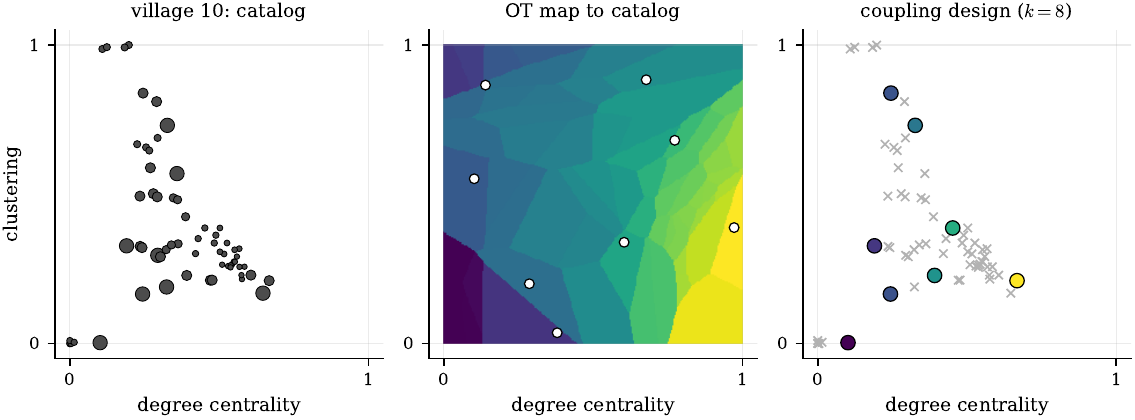}
\caption{Coupling construction for one village.
Left: the catalog of $R = 64$ monitors, marker area proportional to $F$.
Middle: the Laguerre cells of the transport map, one illustrative tuple's uniform draws in white.
Right: the monitors assigned to that tuple ($k = 8$), colored by centrality, gray crosses unassigned.}
\label{fig:net_setup}
\end{figure}

\paragraph{Design and estimand.}
We give each village $64$ savers ($n = 4800$ in total), matching the number of monitors so that stratification is feasible for comparison, though the coupling design itself needs no such restriction.
The binary outcome $\yi(d) \in \{0, 1\}$ records whether the saver reaches their savings goal, with heterogeneous logistic response $P(\yi(d) = 1) = L(\beta_i' t(d))$, where $L$ is the logistic link of Example~\ref{ex:logit} and the basis $t(d)$ collects an intercept, a linear term in the monitor's centrality, and a radial curvature term.
The centrality coefficient is positive on average, as in \citet{breza2019}, and the curvature makes the surface nonlinear, so that dispersing monitors over the feature space is more valuable.
Each saver carries two baseline covariates $X_i$, noisy proxies for the latent coefficients $\beta_i$, and the design matches on the covariates alone.
We investigate two regimes that differ only in covariate predictive power. In a weakly predictive regime, covariates account for $20\%$ of the variance of each coefficient, while in the highly predictive regime they account for $60\%$.
The estimand is the $F$-weighted best logistic approximation of the dose-response surface, estimated by the pooled logit MLE of Example~\ref{ex:logit} over all $75$ villages.
What matters is the profile of this surface, how expected savings vary with a monitor's network position.
All reported variances are design variances over assignment draws for a fixed finite population.

We compare several assignments of monitors to the savers.
Coupling designs first match savers into $k$-tuples on their covariates by the algorithm of \citet{cytrynbaum2023}, for $k \in \{2, 4, 8, 16, 32\}$, under one of two procedures: within-village matching forms tuples among the savers of each village separately, while pooled matching forms tuples from all $4800$ savers at once, so that a tuple may span villages.
Within each matched tuple the design draws jointly dispersed uniforms using various couplings (Latin hypercube, shifted lattice, or scrambled digital net) and pushes each member's uniform onto that member's own village catalog by the village-specific Brenier map, so that similar savers receive monitors dispersed over the network-feature space (Figure~\ref{fig:net_setup}).
The supplementary online appendix describes the data-generating process and both matching procedures in detail.
We also investigate iid randomization, which draws assignments independently from the non-uniform $F$, and stratification.
Stratification uses matched groups of $k = R = 64$ savers, one group per village, assigning each monitor once per group, and because it assigns monitors uniformly rather than with the marginal $F$, its estimator is reweighted to target the same estimand.

\paragraph{Efficiency and inference.}
Figure~\ref{fig:net_efficiency} reports the efficiency gain $1 - \var_G/\var_{\giid}$ for the estimated dose-response surface as a function of the tuple size $k$, with coverage for the centrality coefficient given in Table~\ref{tab:net}.\footnote{For the estimated profile, $\var_G$ centers each assignment draw's fitted surface at its $F$-weighted mean and averages the design variance of the centered surface over the pooled catalog under $F$. Each design cell uses $25{,}000$ Monte-Carlo assignment draws ($5{,}000$ for coverage).}
In the highly predictive regime the coupling designs deliver gains of up to roughly $28\%$ at moderate-to-large tuple sizes under pooled matching, versus far less under within-village matching.
In the weakly predictive regime the gains are smaller throughout, around $8\%$ under pooled matching and only a few percent within villages.
Pooled matching dominates because of match quality: it draws each tuple's members from all $4800$ savers rather than from the $64$ savers of a single village, so tuples remain tight in covariate space even at large $k$, where a village's own pool is exhausted.
Forced stratification does far worse than iid randomization, inflating the variance by about half (Table~\ref{tab:net}), because the implicit marginal distribution it implements deviates from $F$, necessitating reweighting.

\begin{figure}[ht]
\centering
\includegraphics[width=\figwquad]{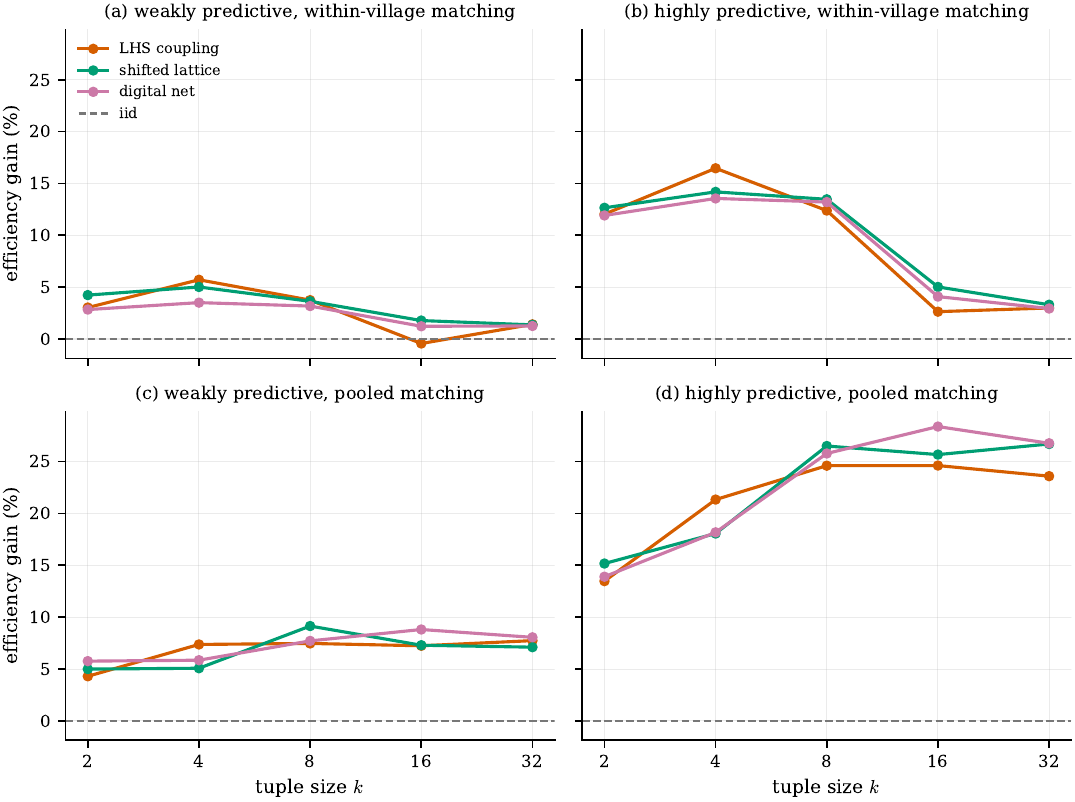}
\caption{Efficiency gain by tuple size $k$, in the weakly predictive (left) and highly predictive (right) regimes, under within-village (top) and pooled (bottom) matching.}
\label{fig:net_efficiency}
\end{figure}

Confidence intervals built from the collapsed-strata variance estimator, with a critical value from a $t$ distribution whose degrees of freedom equal half the total number of matched groups (the finite-sample refinement of Proposition~\ref{prop:ci}), stay at or above the nominal $95\%$ level under both matching procedures (Table~\ref{tab:net}).
Coverage is essentially flat in $k$ under pooled matching and grows conservative at the largest tuple sizes under within-village matching.
The iid benchmark, whose intervals use the same matched groups, shows the same pattern, so the conservativeness reflects the variance estimator rather than the couplings.

\begin{table}[ht]
\centering
\caption{Efficiency gain and coverage of the confidence interval for the centrality coefficient, by matching procedure and tuple size, for the scrambled digital net coupling \citep{owen1995scrambled}. Forced stratification included for comparison.}
\label{tab:net}
\begin{tabular}{llcccc}
\toprule
 & & \multicolumn{2}{c}{efficiency gain (\%)} & \multicolumn{2}{c}{coverage (\%)} \\
\cmidrule(lr){3-4}\cmidrule(lr){5-6}
matching & $k$ & weak & predictive & weak & predictive \\
\midrule
Within village & 2 & 2.8 & 11.9 & 96.1 & 96.0 \\
 & 4 & 3.5 & 13.6 & 96.6 & 96.4 \\
 & 8 & 3.2 & 13.2 & 96.2 & 97.1 \\
 & 16 & 1.2 & 4.1 & 96.9 & 97.9 \\
 & 32 & 1.3 & 2.9 & 98.6 & 99.8 \\
\addlinespace
Pooled & 2 & 5.8 & 13.9 & 96.0 & 95.7 \\
 & 4 & 5.9 & 18.2 & 96.5 & 95.6 \\
 & 8 & 7.7 & 25.8 & 96.3 & 96.1 \\
 & 16 & 8.8 & 28.4 & 96.0 & 96.6 \\
 & 32 & 8.1 & 26.8 & 95.9 & 97.0 \\
\addlinespace
Stratification & 64 & -52.3 & -52.8 & -- & -- \\
\bottomrule
\end{tabular}

\end{table}

The application traces out the trade-off between dispersion and match quality at the heart of the paper's theory (Theorem~\ref{thm:eigenspace-decomposition}).
Larger tuples disperse monitors more evenly over each village's catalog, raising dispersion, but match quality deteriorates as tuples grow, at a rate governed by the pool from which the tuples are formed.
Under within-village matching that pool is a single village's $64$ savers, so match quality erodes quickly and the gains peak at small $k$ before fading.
Under pooled matching each tuple draws on all $4800$ savers, match quality stays high through $k = 32$, and the gains rise and then plateau, mirroring the frontier of Section~\ref{subsection:efficiency}.
Forced stratification is the lexicographic extreme of Example~\ref{ex:stratified-randomization}: it attains perfect dispersion, but only with groups as large as the catalog and therefore excessively poor match quality, and its uniform assignment cannot accommodate the non-uniform marginal $F$, forcing an ex post reweighting that further inflates the variance.

\section{Concluding Remarks}

Coupling designs provide a powerful approach for improving the efficiency of experimentation over complex treatment spaces. 
We showed that such treatments naturally arise in a wide variety of economic applications.
Some examples covered in Section \ref{section:overview} include cash transfer and network experiments in development economics, behavioral experiments with multivariate information structures, and discrete choice experiments in marketing and industrial organization. 

The key insight of our work is that the mechanism underlying conventional stratification can be extended by first matching similar units into groups, then assigning within-group treatments to be highly dispersed over the treatment space $\suppf$. 
The efficiency gain from coupling designs is proportional to the product of sample dispersion and match quality, with a tradeoff where dispersion is increasing and match quality decreasing in the tuple size $k$.
The attained dispersion depends on well-approximation of the influence functions $\si(\cdot)$ on the high dispersion eigenspaces of the coupling operator $U_G$, generally amounting to a mild smoothness requirement.

Several directions for future work remain.
First, we focused on exchangeable couplings with fixed marginals for tractability and exposition.
Alternative designs that use non-exchangeable couplings or allow for flexible marginal distributions may offer further efficiency improvements, but require new tools to analyze their properties.
Second, the mechanism of producing high sample dispersion among similar experimental units applies very broadly, and the stratified structure of coupling designs is not essential to leverage this insight.
Alternative designs that do not partition units into groups, but still assign highly dispersed treatments to similar units, may be possible and could offer further efficiency improvements.
Finally, integrating coupling designs with response-adaptive methods that update the treatment distribution over successive experimental waves is a promising avenue for combining the efficiency gains from both approaches.

{\small
\bibliography{design_references.bib}}

\clearpage

\pagenumbering{arabic}\renewcommand{\thepage}{\arabic{page}}

\appendix
\renewcommand{\thesection}{\Alph{section}}

\section{Appendix} \label{section:extra-results}

\subsection{Multivariate Transport Maps} \label{section:ot-maps}

This appendix formalizes the geometry-preservation property motivating the Brenier map $\tstar$ in Equation \eqref{eqn:monge-ot-intro} and discusses its computation.

\textbf{Rank Preservation.}
In the univariate case, the quantile transform $T = F\inv$ is the unique nondecreasing map with $T(U) \sim F$, up to almost-everywhere equivalence.
Being nondecreasing is equivalent to $(u - v)(T(u) - T(v)) \ge 0$ for all $u, v \in [0, 1]$.
For $m \ge 1$, we can generalize this by requiring that $T$ satisfies $(u - v)'(T(u) - T(v)) \ge 0$ for all $u, v \in [0, 1]^m$, a geometric constraint requiring $T(\cdot)$ to preserve the relative orientation of input points.
A slightly stronger condition is \emph{cyclic monotonicity}.
For any finite set $\{u_1, \dots, u_L\} \sub [0,1]^m$ and permutation $\sigma$ of $[L]$, we can require that
\begin{equation} \label{eqn:cyclic-monotonicity}
\sum_{l \in [L]} u_l' \, T(u_l) \ge \sum_{l \in [L]} u_l' \, T(u_{\sigma(l)}).
\end{equation}
This requires that each output $T(u)$ is the ``most aligned'' with $u$ in the sense that no reassignment among finitely many points can increase the total inner product.
The Brenier map $\tstar$ in Equation \eqref{eqn:monge-ot-intro} is the unique map satisfying both the marginal condition $\tstar(U) \sim F$ and cyclic monotonicity, up to almost-everywhere equivalence \citep{brenier1991polar, mccann1995existence}, and can be viewed as a multivariate generalization of the quantile transform $F\inv$, see e.g.\ \cite{chernozhukov2017monge}.
For $F = \otimes_{j=1}^m F_j$, the componentwise quantile map in Equation \eqref{eqn:product-quantile-map} is the gradient of the separable convex function $\psi(u) = \sum_{j=1}^m \int_0^{u_j} F_j\inv(t) \, dt$ and pushes $\unif[0, 1]^m$ forward to $F$, so it coincides with the Brenier map.

\textbf{Computation.}
For the discrete treatment catalogs $\suppf = \{d_1, \dots, d_R\} \sub \mr^m$ of Section \ref{ex:irregular-uber-eats}, and for the constrained spaces of Example \ref{ex:development-constrained} when the pre-space $\suppfpre$ is discrete, finding $\tstar$ in Equation \eqref{eqn:monge-ot-intro} is a \emph{semi-discrete} optimal transport problem, efficiently solvable by convex optimization \citep{merigot2011multiscale}.
The map has a simple form in this case, partitioning $[0, 1]^m$ into convex regions $[0, 1]^m = \cup_j C_j$, disjoint up to null sets, with each region mapped to a unique point $\tstar(C_j) = d_j$ in $\suppf$.
This partition is known as a Laguerre tessellation.
If $\suppf$ is fully continuous, one can construct a fine discrete approximation $\suppf_{approx} = \{d_1, \dots, d_N\}$ by rejection sampling, proposing $d \sim \unif(\suppfpre)$ and accepting if $C(d) \le B$, then compute the semi-discrete Brenier map onto the accepted points with mass proportional to $F$.
Since the discretization preserves the underlying geometry of $\suppf$, we still expect $\corrg(\phi(\Di), \phi(\Dj)) < 0$ for smooth functions $\phi: \suppf \to \mr$.
If the constraint functions $C(d) \le B$ are simple enough, it is sometimes possible to analytically construct a transport map $T: [0, 1]^m \to \suppf$ using Knothe-Rosenblatt transport \citep{carlier2010knothe} instead.

\subsection{Antithetic Variates} \label{subsection:av}

In this section, we show how our framework characterizes the exact efficiency of the antithetic variates coupling for $k=2$, generalizing the discussion in Section \ref{paragraph:antithetic-variates} beyond monotone influence functions.
Throughout this subsection, let $F = \unif[0, 1]$. The results extend to general univariate $F$ by applying them to the effective influence functions $\si \circ F\inv$ on $[0, 1]$ (Remark~\ref{rem:canonical-marginals}).
Define the even and odd functions in $\ltwonot$ by $E_{even} = \{\phi \in \ltwonot: \phi(d) = \phi(1-d) \}$ and $E_{odd} = E_{even}^\perp$ in $\ltwonot$ to be functions with $\phi(d) = -\phi(1-d)$.
We show the decomposition $\ltwonot = E_{even} \oplus E_{odd}$, which are eigenspaces of $U_G$ with dispersions $\disp_G(\eodd) = 1$ and $\dispg(\eeven) = -1$.
The projections have closed forms, with $\phi^{even}(d) = (1/2)(\phi(d) + \phi(1-d)) - \ef[\phi(D)]$ and $\phi^{odd}(d) = (1/2)(\phi(d) - \phi(1-d))$.
Applying Theorem \ref{thm:eigenspace-decomposition} yields the following exact efficiency of antithetic variates, proved in the supplementary online appendix.

\begin{thm}[AV] \label{thm:av-variance}
The variance $n \var_{G}(\est) = \vmatch(s^{odd}) + 2 \cdot \vg(s^{even})$ and
\begin{equation} \label{eqn:av-efficiency}
1 - \frac{\var_G(\est)}{\var_{\giid}(\est)} = w_{odd} \cdot \matchcoeffk(s^{odd}) - w_{even} \cdot \matchcoeffk(s^{even}).
\end{equation}
\end{thm}

Antithetic variates performs well if $w_{odd}$ is close to $1$ and units are well-matched.
Similar to rotation sampling, efficiency is penalized if the $\si(\cdot)$ have large weight on the negative dispersion eigenspace.
To see why $\dispg(E_{even}) = -1$, note that $\done = U$ and $\dtwo = 1-U$, even functions produce perfectly correlated samples $\phi(\done) = \phi(\dtwo)$.
The term $2 \cdot \vg(s^{even})$ in the nominal variance reflects this clustering effect, which reduces effective sample size by a factor of $2$.

\begin{remark}[Gaussian Equivalence] \label{rem:gaussian-equivalence}
For $k=2$, the Gaussian coupling is equivalent to antithetic variates randomization.
To see this, note if $(Z_1, Z_2) \sim G$ with $\corrg(Z_1, Z_2)=-1$ then $Z_1 = -Z_2$.
The rank transform satisfies $\Phi(z) = 1-\Phi(-z)$ for $z \in \mr$, so for $U_i = \Phi(Z_i)$ we have $U_1 = 1 - U_2$.
\end{remark}

\subsection{Dispersion of Stratified Randomization} \label{subsection:dispersion-stratified}

Stratified randomization is a particular example of a coupling design, where the coupling is given by complete randomization.
By definition, complete randomization assigns equal probability to all treatment allocations $d_{1:k} \in \suppf^k$ that exactly reproduce the target distribution $F$.
We can formalize this as follows:

\begin{defn}[Complete Randomization]\label{defn:stratification}
Let $\suppf = [m]$ and $f_j = F(D=j)$.
For an allocation $d_{1:k} = (d_1, \dots, d_k)$, define the realized proportion of units $i \in [k]$ assigned to $d=j$ by $\wh f_j(d_{1:k}) = k\inv \sum_{i=1}^k \one(d_i = j)$.
Then define the complete randomization coupling $\gs \equiv \unif \{d_{1:k}: \wh f(d_{1:k}) = f \}$.
\end{defn}

As the distribution $F$ becomes more complex, we may need very large group sizes $k$ to satisfy this condition, as the next result shows:

\begin{thm}[Stratified Randomization] \label{thm:classic-stratification}
The coupling $\gs \in \transportsym(F)$ exists if and only if $k \cdot f_j \in \mathbb{N}$ for all $j \in [m]$.
If $\gs$ exists, then $\disp_{\gs}(\phi) = 1$ for any $0 < \varf(\phi) < \infty$.
\end{thm}

The theorem shows that for discrete $F$, complete randomization is a simple heuristic to attain perfect dispersion $\dispg(\cdot) = 1$ within groups.
However, in settings with many treatments or irregular treatment probabilities, this may only exist for very large $k$-tuples, potentially destroying match quality.
By Theorem \ref{thm:variance-parametric}, stratified randomization will generally be suboptimal in such settings since efficiency increases with the product dispersion $\times$ match quality, not dispersion alone.

%\begin{ex}[Response-Adaptive Treatments] \label{ex:response-adaptive}
%There is a large literature on response adaptive designs that optimize assignment probabilities $F_t(\cdot)$ over time using data from previous experimental waves.
%Some recent papers include \cite{kasy2021adaptive}, \cite{tabord-meehan2020}, \cite{cytrynbaum2023} and \cite{zhao2023adaptive}.
%Such estimated probabilities often take irregular forms, e.g.\ $F_t(\cdot) = (0.23, 0.417, \dots)$.
%To balance covariates using classic stratification in such settings, one would need to either (1) round estimated optimal probabilities to rational values with a small common denominator or (2) use very large tuple size $k$ to stratify exactly.
%By contrast, coupling designs allow randomization within tightly matched $k$-tuples for \emph{any} choice of $k \ge 2$ and any distribution $F_t(\cdot)$.
%For a simple example, if $F = \bern(p)$ for $p=0.23$, then $k \cdot p \in \mathbb{N}$ in Theorem \ref{thm:classic-stratification} requires $k \ge 100$.
%By contrast, if we use Latin hypercube (Example \ref{ex:latin-hypercube}) with $k=4$, we will have $\dispg(H) \approx 0.9$ for the Horvitz-Thompson weights $H(d)$, with radically improved match quality.
%\end{ex}

\subsection{Estimator Asymptotic Theory} \label{subsection:estimator-asymptotics}

We first record a simple pointwise consistency lemma, which follows from the uniform consistency of Theorem~\ref{thm:uniform-consistency-rates}.

\begin{lem}[Pointwise Consistency] \label{lem:generic-rootn-consistency}
Let $G \in \transportsym(F)$ with $k = O(1)$.
If $\viid(s) = O(1)$, then $\en[\si(\Di)] = \thetan + \Op(\negrootn)$.
\end{lem}

The proof of the following proposition follows by Lemma~\ref{lem:generic-rootn-consistency} and standard algebraic manipulations.

\begin{prop}[OLS Linearization] \label{prop:ols-linearization}
Suppose $k = O(1)$, $\ef[|D|_2^4] < \infty$, $\en \ef[\yi(D)^4] = O(1)$, and $\varf(D) \succ 0$.
Then the OLS coefficient $\wh \beta$ of Example~\ref{ex:blp} satisfies the linearization in Equation~\eqref{eqn:ols-influence-function}.
\end{prop}

The asymptotic linearization for the logit estimator also follows by standard arguments that upgrade the pointwise consistency in Lemma~\ref{lem:generic-rootn-consistency} to a uniform law of large numbers.

\begin{prop}[Logit Linearization] \label{prop:logit-linearization}
Consider the discrete-choice setting of Example~\ref{ex:logit}, with binary outcomes $\yi(d) \in \{0, 1\}$ and logit MLE $\wh \beta$.
Suppose $k = O(1)$, $\ef[|D|_2^4] < \infty$, $\sup_n |\betan|_2 < \infty$, and $\ef[DD'] \succ 0$.
Then Equation~\eqref{eqn:logit-kl-projection} holds and $\wh \beta$ satisfies the linearization $\wh \beta - \betan = \en[\si(\Di)] + \Op(n\inv)$ with influence function $\si(\cdot)$ as in Example~\ref{ex:logit}.
\end{prop}

The proofs of both Proposition~\ref{prop:ols-linearization} and Proposition~\ref{prop:logit-linearization} can be found in the supplementary online appendix.

\section{Proofs}

\subsection{Proofs for Section \ref{section:dispersion-match-quality}} \label{proofs:dispersion-match-quality}

\begin{proof}[Proof of Theorem \ref{thm:variance-parametric}]
For $\est = \en[\si(\Di)]$ with $\si(d) = \ci + \ai \phi(d)$, we have
\begin{align*}
n \var_G(\est)
&= \viid(s) + n \inv \sum_{g} \sum_{i \neq j \in [k]} \covg(\sig(\dig), \sjg(\djg)) \\
&= \viid(s) + n \inv \sum_{g} \sum_{i \neq j \in [k]} \aig \ajg \covg(\phi(\dig), \phi(\djg)).
\end{align*}
The first equality is by definition of coupling designs, with $\est = n\inv \sum_g \sum_i \sig(\dig)$.
The second equality is by our parametric assumption.
By Proposition \ref{prop:dispersion-interpretation}, this equals $\viid(s) - \dispg(\phi) \cdot (k-1)\inv n \inv \sum_{g} \sum_{i \neq j \in [k]} \aig \ajg \varf(\phi)$.
By Lemma~\ref{lem:var-components}, match quality $\matchcoeffk(s) = (k-1)\inv c(s) / \viid(s)$ for $c(s) = n\inv \sum_g \sum_{i\neq j \in [k]} \covf(\sig, \sjg) = \varf(\phi) \cdot n\inv \sum_g \sum_{i\neq j \in [k]} \aig \ajg$.
Then continuing the calculation above, this is
\begin{align*}
&= \viid(s) - \dispg(\phi) \cdot (k-1)\inv c(s) = \viid(s) - \viid(s) \dispg(\phi) \matchcoeffk(s) \\
&= \viid(s)(1 - \dispg(\phi) \matchcoeffk(s)).
\end{align*}
Under the iid design the covariance terms vanish, so the same calculation gives $n \var_{\giid}(\est) = \viid(s)$.
Dividing the two expressions finishes the proof.
\end{proof}

\begin{proof}[Proof of Theorem \ref{thm:classic-stratification}]
First, suppose such $\gs$ exists.
Then there exists $z \in [m]^k$ with $\wh f(z) = f$.
In particular, $\wh f_j(z) = k\inv \sum_{i \in [k]} \one(z_i = j) = f_j$, so $k \cdot f_j = \sum_{i \in [k]} \one(z_i = j) \in \mathbb N$.
Conversely, if $n_j \equiv k \cdot f_j \in \mathbb{N}$ for all $j$, we have $\sum_j n_j = k$.
Then we can construct an allocation $z \in [m]^k$ with $z_i = j$ for exactly $n_j$ indices $i$.
This allocation has $\wh f_j(z) = n_j / k = f_j$ for all $j \in [m]$, so $z \in \{z: \wh f(z) = f \}$.
Then the set $\{z: \wh f(z) = f \}$ is non-empty and $\gs = \unif \{z: \wh f(z) = f \}$ exists.

Next, we show $\gs \in \transportsym(F)$.
Claim $\gs$ is exchangeable.
We must show $P_{\gs}(Z_{\sigma} = z) = P_{\gs}(Z = z)$ for all $z \in [m]^k$ and any permutation $\sigma \in S_k$.
Note $P_{\gs}(Z_{\sigma} = z) = P_{\gs}(Z = z_{\sigma\inv})$.
Since $\gs$ is uniform on $S \equiv \{z: \wh f(z) = f\}$, we have $P_{\gs}(Z = z) = |S|\inv \one(z \in S)$.
Then it suffices to show $z \in S \iff z_{\sigma} \in S$ for all $z \in [m]^k$ and $\sigma \in S_k$, noting $\{\sigma: \sigma \in S_k\} = \{\sigma \inv: \sigma \in S_k\}$.
For $z \in [m]^k$ and $\sigma \in S_k$,
\[
\wh f_j(z_{\sigma}) = k\inv \sum_{i \in [k]} \one(z_{\sigma(i)} = j) = k\inv \sum_{l \in [k]} \one(z_l = j) = \wh f_j(z)
\]
where the second equality is by the substitution $\sigma\inv(l) = i$.
Then $\wh f(z_{\sigma}) = f$ if and only if $\wh f(z) = f$, so $z \in S \iff z_{\sigma} \in S$.
This proves the claim.
Next, we show the marginal $\gs_l = F$.
Fix $j \in [m]$ and note that
\begin{align*}
f_j &= E_{\gs} \bigl [\wh f_j(Z) \bigr ] = k\inv \sum_{i=1}^k P_{\gs}(Z_i=j) = P_{\gs}(Z_l=j) \quad \text{for any} \; l \in [k].
\end{align*}
The first equality holds by definition of $\gs$, the final equality by exchangeability.
This finishes the proof of the first result.
Next consider the claim about dispersion.
Let $(\Di)_{i \in [k]} \sim \gs$.
By definition,
\[
k\inv \sum_{i \in [k]} \phi(\Di) = \sum_{j \in [m]} \phi(j) \cdot \wh f_j(D) = \sum_{j \in [m]} \phi(j) \cdot f_j = \ef[\phi(D)].
\]
Then by Equation \eqref{eqn:pure-dispersion}, we have $0 = k \var_{\gs}(\est) = \varf(\phi)(1-\disp_{\gs}(\phi))$, so $\disp_{\gs}(\phi) = 1$.
This finishes the proof.
\end{proof}

\begin{proof}[Proof of Corollary \ref{cor:covariate-balance}]
Define $\psi(d) = \phi(d) - \ef[\phi(D)]$ and $z_i = b_i - \bar b$.
Set $\sfn_i(d) = z_i \cdot \psi(d)$, so that $\cov_n(\phi(\Di), b_i) = \en[\psi(\Di) \cdot z_i] = \en[\si(\Di)]$.
This has the form studied in Theorem \ref{thm:variance-parametric}.
Since $\ef[\psi(D)] = 0$, we have $\ef[\sfn_i(\Di)] = 0$ for $i \in [n]$, so $\eg[\cov_n(\phi(\Di), b_i)^2] = \eg[(\en[\si(\Di)])^2] = \var_G(\en[\sfn_i(\Di)])$.
By Theorem~\ref{thm:variance-parametric},
\begin{equation} \label{eqn:proof-balance-calculation}
n \cdot \eg[\cov_n(\phi(\Di), b_i)^2] = \viid(s)\big(1 - \dispg(\phi) \cdot \matchcoeffk(s)\big).
\end{equation}
The iid variance is $\viid(s) = \en[z_i^2 \varf(\psi)] = \varf(\phi) \cdot \var_n(b_i)$.
For match quality, by the last identity in Lemma~\ref{lem:var-components},
$\vmatch(s) = (k-1)\inv (n/k)\inv \sum_g \sum_{i \in [k]} \varf(\sig - \sgbar)$.
Since $\sig(d) - \sgbar(d) = (z_{ig} - \bar z_g) \psi(d)$ and $z_{ig} - \bar z_g = b_{ig} - \bar b_g$, we have $\varf(\sig - \sgbar) = (b_{ig} - \bar b_g)^2 \varf(\phi)$.
Then by equality above $\vmatch(s) = \varf(\phi) \cdot (n/k)\inv \sum_g \vark(b_{ig})$.
Then
\[
\matchcoeffk(s) = 1 - \frac{\vmatch(s)}{\viid(s)} = 1 - \frac{(n/k)\inv \sum_g \vark(b_{ig})}{\var_n(b_i)} = \matchcoeffk(b).
\]
By definition, $\imb_G(\phi, b) = \eg[\cov_n(\phi(\Di), b_i)^2]$, so Equation~\eqref{eqn:proof-balance-calculation} and $\matchcoeffk(s) = \matchcoeffk(b)$ give
$n \cdot \imb_G(\phi, b) = \viid(s)(1 - \dispg(\phi) \cdot \matchcoeffk(b))$.
Setting $G = \giid$ and using $\disp_{\giid}(\phi) = 0$, we obtain $n \cdot \imb_{\giid}(\phi, b) = \viid(s)$.
Dividing, $\imb_G(\phi, b) / \imb_{\giid}(\phi, b) = 1 - \dispg(\phi) \cdot \matchcoeffk(b)$, which finishes the proof.
\end{proof}

\subsection{Proofs for Section \ref{section:efficiency-analysis}}

\begin{lem}[Operator Properties]\label{lem:operator-properties}
For $G \in \transportsym(F)$, the coupling operator $U_G$ is a self-adjoint contraction on $\elltwo(F)$.
If $E$ is an eigenspace of $U_G$, then $E$ is closed in $\elltwo(F)$.
If $E, E'$ are eigenspaces of $U_G$ with $\lambda_E \neq \lambda_{E'}$, then $E \perp E'$.
\end{lem}

\begin{proof}[Proof of Lemma \ref{lem:operator-properties}]
Write $U = U_G$.
First, we show $U$ is a contraction on $\elltwo(F)$.
For any $\phi \in L^2(F)$, by conditional Jensen's inequality and tower law
\begin{align*}
E_F[(U\phi)(D)^2] &= E_F[E_G[\phi(D_1)|D_2]^2] \le E_F[E_G[\phi(D_1)^2|D_2]] \\
&= E_G[\phi(D_1)^2] = E_F[\phi(D)^2].
\end{align*}
Next we show self-adjointness.
For $\phi, \psi \in L^2(F)$,
\begin{align*}
\langle U\phi, \psi \rangle_F &= E_F[(U\phi)(D)\psi(D)] = \ef[\eg[\phi(\done) | \dtwo] \psi(\dtwo)] = \eg[\phi(\done) \psi(\dtwo)] \\
&= \eg[\phi(\dtwo) \psi(\done)] = \lf U \psi, \phi \rf.
\end{align*}
The second equality follows since $\dtwo \sim F$.
The fourth equality since by exchangeability, $(\done, \dtwo) \sim (\dtwo, \done)$, so $(\phi(\done), \psi(\dtwo)) \sim (\phi(\dtwo), \psi(\done))$.
For closedness, since $U$ is a contraction it is bounded, and the eigenspace $E = \ker(U - \lambda I)$ is the kernel of a bounded linear operator, hence closed.
For orthogonality, let $\phi \in E$ and $\psi \in E'$ with $\lambda_E \neq \lambda_{E'}$.
Then by self-adjointness, $\lambda_E \lf \phi, \psi \rf = \lf U\phi, \psi \rf = \lf \phi, U\psi \rf = \lambda_{E'} \lf \phi, \psi \rf$.
Since $\lambda_E \neq \lambda_{E'}$, we have $\lf \phi, \psi \rf = 0$.
\end{proof}

\begin{proof}[Proof of Theorem \ref{thm:dispersion-decomposition}]
First consider (a).
For nonzero $\phi \in E \sub \ltwonot$, we have $\varf(\phi) > 0$, so
\begin{align*}
\dispg(\phi) &= -(k-1)\corrg(\phi(\done), \phi(\dtwo)) = -(k-1) \frac{\eg[\phi(\done)\phi(\dtwo)]}{\varf(\phi)} \\
&= -(k-1) \frac{\lf U_G \phi, \phi \rf}{\varf(\phi)} = -(k-1) \frac{\varf(\phi)}{\varf(\phi)} \cdot \eval = -(k-1) \cdot \eval.
\end{align*}
The third equality by calculations in Lemma \ref{lem:operator-properties}.
The fourth equality since $\phi \in E \sub \ltwonot$.
This proves the claim.
Next consider (b).
First, we show the projection $P_m$ exists and has the required form.
By Lemma \ref{lem:operator-properties}, each eigenspace $E_m$ is closed in $\elltwo(F)$, so the projection $P_m \phi = \argmin_{f \in E_m} \ef[(\phi(D) - f(D))^2]$ exists and is unique.
Moreover, recall $\ef[f(D)] = 0$ for any $f \in E_m \subset \ltwonot$.
Let $\phi' = \phi - \ef[\phi]$.
Then
\begin{align*}
&\ef[(\phi(D) - f(D))^2] = \ef\big[(\phi'(D) + \ef[\phi] - f(D))^2 \big] \\
&= \ef[\phi]^2 + \ef\big[(\phi'(D) - f(D))^2 \big] = \varf(\phi - f) + \ef[\phi]^2.
\end{align*}
The second equality since $\ef[\phi' - f] = 0$, so the cross-term cancels.
This shows that $P_m \phi = \argmin_{f \in E_m} \varf(\phi - f)$.
For the main result, let $n=k$ and define $\si = \phi$ for each $i \in [k]$.
We use Theorem \ref{thm:eigenspace-decomposition}, whose proof below relies only on part (a) of the present theorem, so there is no circularity: the dependency order is part (a), then Theorem \ref{thm:eigenspace-decomposition}, then part (b).
Combining Equation \eqref{eqn:pure-dispersion} and Theorem \ref{thm:eigenspace-decomposition},
\begin{equation}
\dispg(\phi) = 1- \frac{\var_G(\est)}{\var_{\giid}(\est)} = \sum_{m \ge 1} w_m \cdot \disp_G(E_m) \matchcoeffk(\phi^m).
\end{equation}
Since $\sfn_i = \phi$ for all $i$, the match quality $\matchcoeffk(\phi^m) = 1$.
The projection weight $w_m = \en \varf(P_m \si) / \en \varf(\si) = \varf(P_m \phi) / \varf(\phi)$.
This finishes the proof.
\end{proof}

\begin{proof}[Proof of Corollary~\ref{cor:extremal-directions}]
Let $\ltwonot = E \oplus E^\perp$ be eigenspaces of $U_G$ with dispersions $\dispg(E)$ and $\dispg(E^\perp)$.
Without loss, suppose $\dispg(E) > \dispg(E^\perp)$.
Let $\phi \in \elltwo(F)$ with $\varf(\phi) > 0$.
By Theorem~\ref{thm:dispersion-decomposition}, we have $\dispg(\phi) = w \cdot \dispg(E) + (1-w) \cdot \dispg(E^\perp)$, for $w = \varf(P_E \phi) / \varf(\phi) \in [0,1]$.
Then apparently $\dispg(E^\perp) \le \dispg(\phi) \le \dispg(E)$.
To show the upper bound is achieved, take $\phi \in E$ with $\varf(\phi) > 0$.
Then $P_E \phi = \phi$ and $P_{E^\perp} \phi = 0$, so $w = \varf(\phi)/\varf(\phi) = 1$, so $\dispg(\phi) = \dispg(E)$ is achieved.
Conversely, suppose $\dispg(\phi) = \dispg(E)$ for some $\phi \neq c$.
Then $w \cdot \dispg(E) + (1-w) \cdot \dispg(E^\perp) = \dispg(E)$, which requires $w = 1$ since $\dispg(E) \ne \dispg(E^\perp)$.
Then $P_{E^\perp} \phi = 0$ and $\phi \in E$.
This completes the proof.
\end{proof}

\begin{proof}[Proof of Theorem \ref{thm:eigenspace-decomposition}, Corollary \ref{cor:efficiency-convex}]
Define $\ri \equiv \si - \ef[\si] \in \ltwonot$ and let $\rim \equiv P_m \ri$ and $\si^m = P_m \si$.
We have $P_m \ef[\si] = 0$ since $\ef[\si] \perp E_m$.
Then $P_m \ri = P_m \si$ by linearity.
Since $\rim \in E_m$, by definition $U \rim = \evalm \rim$ for all $m$, and (b) $\rim \perp \ri^{l}$ for $m \not =l$, by orthogonality of eigenspaces.
By Assumption \ref{assump:direct-sum}, $\ri = \sum_m P_m \ri = \sum_m \rim$.
Note the key fact $\covg(\phi(\done), \psi(\dtwo)) = \eg[\phi(\done) \psi(\dtwo)] = \eg[\eg[\phi(\done)|\dtwo]\psi(\dtwo)] = \ef[(U\phi)(D)\psi(D)] = \lf U \phi, \psi \rf$ for any $\phi, \psi \in \ltwonot$.
Denote $U_G = U$ and calculate
\begin{align*}
&n \var_G(\est) = n\inv \sum_i  \varf(\si(D)) + n \inv \sum_{g} \sum_{i \neq j \in [k]} \covg(\sig(\done), \sjg(\dtwo)) \\
&= n\inv \sum_i  \varf(\ri(D)) + n \inv \sum_{g} \sum_{i \neq j \in [k]} \covg(\rig(\done), \rjg(\dtwo)).
\end{align*}
Continuing, this is
\begin{align*}
&n\inv \sum_i  |\sum_m \ri^m|_F^2 + n \inv \sum_{g} \sum_{i \neq j \in [k]} \covg \big (\sum_m \rigm(\done), \sum_l \rjg^l(\dtwo) \big ) \\
&= n\inv \sum_i  \big |\sum_m \rim \big |_F^2 + n \inv \sum_{g} \sum_{i \neq j \in [k]} \big \lf U \sum_m \rigm, \sum_l \rjg^l \big \rf.
\end{align*}
The first equality by definition of the design.
The second equality since $G \in \transportsym(F)$ so $\sig(\done) - E_G[\sig(\done)] = \sig(\done) - \ef[\sig(D)] = \rig(\done)$.
The third equality since $\ri = \sum_m \ri^m$ by Assumption \ref{assump:direct-sum}.
The fourth equality by the key fact.
Continuing,
\begin{align*}
&= n\inv \sum_i  \big |\sum_m \rim \big |_F^2 + n \inv \sum_{g} \sum_{i \neq j \in [k]} \big \lf \sum_m \evalm \rigm, \sum_l \rjg^l \big \rf \\
&= n\inv \sum_i \sum_m |\rim|_F^2 + n \inv \sum_{g} \sum_{i \neq j \in [k]} \sum_m \evalm \lf \rigm, \rjgm \rf.
\end{align*}
This equals $n\inv \sum_m \sum_i |\rim|_F^2 + n \inv \sum_m \evalm \sum_{g} \sum_{i \neq j \in [k]} \lf \rigm, \rjgm \rf$.
The first equality since $U \sum_m \rigm = \sum_m U \rigm = \sum_m \evalm \rigm$ by continuity of $U$, since $U$ is a linear contraction by Lemma \ref{lem:operator-properties}.
The first term in the second equality follows by Parseval's theorem.
The second term follows by continuity of the inner product map $(\phi, \psi) \to \lf \phi, \psi \rf$ and orthogonality.
To see this, note $\lf \sum_m \evalm \rigm, \sum_l \rjg^l \rf = \lim_{K \to \infty} \lf \sum_{m=1}^K \evalm \rigm, \sum_{l=1}^K \rjg^l \rf = \lim_{K \to \infty} \sum_{m=1}^K \evalm \lf \rigm, \rjgm \rf$.
For the third equality above, the exchange of sums can be rigorously justified by a standard Fubini argument, using $\en \varf(\si) < \infty$.
Continuing from above, expand $\lf \rigm, \rjgm \rf = -(1/2)(|\rigm - \rjgm|_F^2 - |\rigm|_F^2 - |\rjgm|_F^2)$.
Then $n \inv \sum_m \evalm \sum_{g} \sum_{i \neq j \in [k]} \lf \rigm, \rjgm \rf$ is
\begin{align*}
&- (2n)\inv \sum_m \evalm \sum_g \sum_{i \neq j \in [k]} (|\rigm - \rjgm|_F^2 - |\rigm|_F^2 - |\rjgm|_F^2) \\
&=  \sum_m \evalm \big [(k-1) n\inv \sum_i |\rim|_F^2  - (2n)\inv \sum_g \sum_{i \neq j \in [k]} |\rigm - \rjgm|_F^2 \big].
\end{align*}
This equals $\sum_m \dispg(m) \big [(2n(k-1))\inv \sum_g \sum_{i \neq j \in [k]} |\rigm - \rjgm|_F^2 - \en |\rim|_F^2 \big]$.
The last equality substitutes $\evalm = -\dispg(m)/(k-1)$, by Theorem~\ref{thm:dispersion-decomposition}(a).
Note that $|\rim|_F^2 = \varf(\si^m)$.
Then putting both terms together, we get
\begin{align*}
n\var_G(\est) &= \sum_m (1- \dispg(m)) \cdot \en \varf(\si^m) \\
&+ \dispg(m) \cdot (2n(k-1))\inv \sum_g \sum_{i \neq j \in [k]} \varf(\sig^m - \sjg^m) \\
&= \sum_m (1- \dispg(m)) \cdot \viid(s^m) + \dispg(m) \cdot \vmatch(s^m).
\end{align*}
This proves Corollary \ref{cor:efficiency-convex}.
Then $\releff(G) = 1-\var_G(\est) / \var_{\giid}(\est)$ has
\begin{align*}
\releff(G) &= 1 - \viid(s)\inv \sum_m \left[ \viid(s^m)(1- \dispg(m)) + \vmatch(s^m) \dispg(m) \right] \\
&= \viid(s)\inv \sum_m \dispg(m)(\viid(s^m) - \vmatch(s^m)) \\
&= \sum_m \frac{\viid(s^m)}{\viid(s)} \dispg(m) \left( 1 - \frac{\vmatch(s^m)}{\viid(s^m)} \right) = \sum_m w_m \dispg(m) \matchcoeffk(s^m).
\end{align*}
This proves Theorem \ref{thm:eigenspace-decomposition}.
Finally, we prove an extra decomposition needed for Theorem \ref{thm:gaussian-efficiency}.
By Lemma~\ref{lem:var-components}, $\vmatch(s^m) = \viid(s^m) - (k-1)\inv c(s^m)$ and by Theorem \ref{thm:dispersion-decomposition}(a) we have $\dispg(m) = -(k-1)\evalm$, so we can write
\begin{align*}
n\var_G(\est) &= \sum_m \viid(s^m) + \dispg(m) \cdot (\vmatch(s^m) - \viid(s^m)) \\
&= \sum_m \viid(s^m) + \dispg(m) \cdot (-(k-1)\inv c(s^m)) \\
&= \sum_m \viid(s^m) + \evalm c(s^m) = \viid(s) + \sum_m \evalm \cdot c(s^m).
\end{align*}
This finishes the proof.
\end{proof}

\textbf{LHS Coupling.}
Next, we consider Examples \ref{ex:lhs-dispersion}, \ref{ex:lhs-analysis}, and \ref{ex:lhs-analysis-2}.
Let $G$ be the LHS coupling for canonical marginal $F=\unif[0,1]$.
Define $J(l) = [(l-1)/k, l/k)$ for $l \in [k]$ and $I(d) = \sum_l J(l) \one(d \in J(l))$.
Define the demeaned histogram space $\ehist = \{\phi \in \ltwonot: \phi(x) = \sum_l a_l \one(x \in J(l)) \}$, constants $E_1 = \{\phi: \phi(x) = c\}$, and remainder $\ecomp = \ehist^\perp$ (orthogonal complement in $\ltwonot$).
First, we show that Assumption \ref{assump:direct-sum} holds and characterize the projection operators.

\begin{lem}[LHS]\label{lem:lhs-operator}
Let $F = \unif[0, 1]$ and $G$ the LHS coupling.
The operator $(U_G \phi)(d) = E[\phi(D) | D \not \in I(d)]$.
The direct sum $\ltwonot = \ehist \oplus \ehist^\perp$ holds, eigenspaces of $U_G$ with eigenvalues $\eval = -(k-1)\inv$, and $0$, verifying Assumption \ref{assump:direct-sum}.
The projection is given by $\phist \phi(d) = E[\phi(D) | D \in I(d)] - \ef[\phi(D)]$.
\end{lem}

\begin{proof}
First, we characterize the operator $U_G$.
Recall that $\Di = k\inv (\pii - 1 + U_i)$, where $\pi$ is a random permutation of $[k]$ and $U_i \sim \unif[0,1]$, with $\pi \indep U$.
By tower law $E[\phi(\dtwo) | \done] = E[E[\phi(\dtwo) | \done, \pi] | \done]$.
Since $(U_1, U_2, \pi)$ are jointly independent, we have $U_2 \indep U_1 | \pi$.
Moreover, $\done \in \sigma(U_1, \pi)$ and $\dtwo \in \sigma(U_2, \pi_2)$, so $\done \indep \dtwo | \pi$.
Then we conclude $E[\phi(\dtwo) | \done, \pi] = E[\phi(\dtwo) | \pi]$.
Next, we analyze $E[\phi(\dtwo) | \pi]$.
The distribution $\dtwo | \pi \sim k\inv (\pi_2 - 1) + \unif[0, k\inv]$.
Then $E[\phi(\dtwo) | \pi] = m(\pi_2)$ for discrete function $m(l) \equiv k \int_{(l-1)/k}^{l/k} \phi(x) dx$.
Then we have
\begin{align*}
E[\phi(\dtwo) | \done] &= E[E[\phi(\dtwo) | \pi] | \done] = E[m(\pi_2) | \done] = E[E[m(\pi_2) | U_1, \pi_1] | \done] \\
&= E[E[m(\pi_2) | \pi_1] | \done] = E[m(\pi_2) | \pi_1].
\end{align*}
The third equality by tower law, since $\done \in \sigma(U_1, \pi_1)$.
The fourth equality since $\pi \indep U$.
The last equality since $\pi_1 = f(\done)$ for $f(\done) = \lfloor k \done \rfloor + 1$, so $\sigma(\pi_1) \sub \sigma(\done)$.
By calculations in Lemma~\ref{lem:rs-distribution}, $\pi_2 | \pi_1$ is uniform on $[k] \setminus \{\pi_1\}$.
Then the expectation is $E[m(\pi_2) | \pi_1] = (k-1)^{-1} \sum_{l \neq \pi_1} m(l) = (k-1)^{-1} \sum_{l \neq \pi_1} k \int_{(l-1)/k}^{l/k} \phi(x) \, dx$.
Note $\pi_1 = \lfloor k \done \rfloor + 1$.
Set $I(d) = J(l)$ if $d \in J(l)$.
So $(U_G \phi)(d) = E[\phi(\dtwo) | \done=d]$:
\begin{align*}
(U_G \phi)(d) &= (k-1)^{-1} \sum_{l \neq \lfloor k d \rfloor + 1} k \int_{(l-1)/k}^{l/k} \phi(t) dt =  \frac{k}{k-1} \int_{I(d)^c} \phi(t) dt \\
&= E[\phi(D) | D \not \in I(d)] = (k-1)\inv(k \cdot \ef[\phi] - \ef[\phi(D) | D \in I(d)]).
\end{align*}
The final equality follows by tower law.
This finishes the characterization of $U_G$.

Next, we show the claimed decomposition $\ltwonot = \ehist \oplus \ehist^\perp$.
For $\phi \in \ehist$ we have $\ef[\phi] = 0$ by definition and clearly $\ef[\phi(D) | D \in I(d)] = \phi(d)$ so $(U_G\phi)(d) = -(k-1)\inv \phi(d)$ by the calculation above.
Then $\ehist$ is an eigenspace with $\eval = -(k-1)\inv$.
Finally, note that we can write $(U_G \phi)(d) = k(k-1)\inv \lf \phi, t_d \rf$ for $t_d(x) = \one(x \not \in I(d))$.
Since $t_d \in \linearspan(1) \oplus \ehist$, we have $\lf \phi, t_d \rf = 0$ for $\phi \in \ehist^\perp$, noting the orthocomplement is taken within $\ltwonot$.
Then $\ehist^\perp$ is an eigenspace with $\eval = 0$.
Note $\ehist$ is closed as a finite-dimensional linear subspace, so $\ltwonot = \ehist \oplus \ehist^\perp$, verifying Assumption \ref{assump:direct-sum}.

Finally, let $P\phi(x) = \ef[\phi(D) | D \in I(x)] - \ef[\phi(D)]$.
We claim $P = \phist$.
Let $\phi \in \ehist$, then since $\phi$ is piecewise constant on $J(l)$, $\ef[\phi(D) | D \in I(x)] - \ef[\phi(D)] = \phi(x) - 0 = \phi(x)$, so $P\phi = \phi$.
Also if $\phi = c$ then $P \phi = 0$.
Third, note if $\phi \perp \ehist$ and $\phi \in \ltwonot$, we have $\phi \perp \one(D \in J(l))$ for each $l$, so $\ef[\phi(D) | D \in I(x)] = 0$, and $\ef[\phi(D)] = 0$, giving $P\phi = 0$.
For any $\phi \in \elltwo(F) = \linearspan(1) \oplus \ehist \oplus \ehist^\perp$, we have $\phi = \ef[\phi] + \phist \phi + (\phi - \ef[\phi] - \phist \phi)$, where the last term lies in $\ehist^\perp$.
We have shown $P$ is a linear operator that fixes $\ehist$ and annihilates $\ehist^\perp \oplus \linearspan(1)$, so applying $P$ to equation above gives $P\phi = P \phist \phi = \phist \phi$ for any $\phi \in \elltwo(F)$.
This finishes the proof.
\end{proof}

\begin{lem} \label{lem:direct-sum-satisfied}
Assumption \ref{assump:direct-sum} is satisfied for any $G \in \transportsym(F)$ if $|\supp(F)| < \infty$.
If $F$ is continuous, then the assumption holds for the antithetic variates, Latin hypercube, rotation sampling, and Gaussian couplings in Section \ref{subsection:matching-couplings}.
\end{lem}

\begin{proof}[Proof of Lemma \ref{lem:direct-sum-satisfied}]
Consider the first statement.
We claim $U_G|_{\ltwonot}$ is compact.
To see this, let $B \sub \ltwonot$ be a bounded set.
We need to show $U_G B$ is relatively compact in $\ltwonot$.
Let $\supp(F) = \{d_1, \dots, d_M\}$.
Then $L^2(F) = \linearspan\{ \one(D=d_l) : l \in [M] \}$ is finite dimensional, so $\ltwonot \sub L^2(F)$ is also finite dimensional.
Then the Heine-Borel theorem applies, so the closure $\bar B$ is compact.
By Lemma \ref{lem:operator-properties}, $U_G$ is a linear contraction, hence bounded and continuous.
Then $U_G(B) \sub U_G(\bar B)$ compact, so $U_G(B)$ is contained in a compact set, so it is relatively compact.
Also by Lemma \ref{lem:operator-properties}, $U_G|_{\ltwonot}$ is self-adjoint on $\ltwonot$.
By the spectral theorem for compact self-adjoint operators, $\ltwonot = \oplus_{m \ge 1} E_m$, verifying Assumption \ref{assump:direct-sum}.

Next, consider the second statement.
Lemmas \ref{lem:av-operator}, \ref{lem:rs-operator}, \ref{lem:lhs-operator}, and \ref{lem:gaussian-operator} show that $L^2_0(Q) = \oplus_{m \ge 1} W_m$ for eigenspaces $W_m$ of the operator $T_H$, where $H = $ antithetic variates, rotation sampling, Latin hypercube, and Gaussian respectively with canonical marginals $Q = \normal(0, 1)$ in the Gaussian case and $Q = \unif[0, 1]$ otherwise.
For the Gaussian coupling, Lemma \ref{lem:gaussian-operator} covers $k \ge 3$. For $k = 2$, the Gaussian coupling coincides with the antithetic variates coupling after the rank transform (Remark \ref{rem:gaussian-equivalence}), so this case is covered by Lemma \ref{lem:av-operator}.
The distributions $(\Di)_{i=1}^k \sim G \in \transportsym(F)$ for all of these couplings satisfy the conditions of Lemma \ref{lem:canonical-space}, with $\Di = v(U_i)$ for $(U_i)_{i=1}^k \sim H$ with $v = F\inv$ for the uniform cases and $v = F\inv \circ \Phi$ for the Gaussian coupling.
We claim that if $F$ is continuous, then in each case $v$ is injective with measurable inverse on its image.
For injectivity, note that $F \circ F\inv = I$ on $(0, 1)$ since $F$ is continuous, so $F\inv$ is injective with left-inverse $F$ on its range.
By continuity, the left-inverses $v\inv = F$ and $v\inv = \Phi\inv \circ F$ are both measurable.
Then by Lemma \ref{lem:canonical-space}, Assumption \ref{assump:direct-sum} is satisfied for each coupling $G$.
This finishes the proof.
\end{proof}

\textbf{Canonical Spaces.}
Here, we develop the machinery about canonical spaces needed for Lemma \ref{lem:direct-sum-satisfied}.
Suppose for some $Q$ the space $L^2_0(Q) = \oplus_{m \ge 1} W_m$, where $W_m$ are eigenspaces of the operator $(T_H \phi)(u) = E_H[\phi(U_i) | U_j=u]$ for $i \neq j$ with $H \in \transportsym(Q)$, as in Assumption \ref{assump:direct-sum}.
Define the coupling $G \in \transportsym(F)$ by $(\Di)_{i=1}^k \sim G$ for $\Di = v(U_i)$ for $(U_i)_{i=1}^k \sim H$ and some $v: \supp(Q) \to \supp(F)$.
In particular, this implies $D = v(U) \sim F$ for $U \sim Q$.
Define the linear pushforward operator $R: \elltwo(F) \to \elltwo(Q)$ by $(R\phi)(u) = (\phi \circ v)(u)$.
The proof of the following lemma is a functional analysis exercise and is deferred to the supplementary online appendix for space.

\begin{lem}[Canonical Spaces]\label{lem:canonical-space}
The coupling operator has $U_G = R^* T_H R$.
If the map $v(\cdot)$ is injective with measurable inverse $v\inv$ on $v(\supp(Q))$, then $R^* = R\inv$ and $\ltwonot = \oplus_{m \ge 1} E_m$ for eigenspaces $E_m = R \inv(W_m)$ of $U_G$ with the same eigenvalues $\eval(E_m) = \eval(W_m)$.
In particular, Assumption \ref{assump:direct-sum} is satisfied for $G$.
\end{lem}

\subsection{Proofs for Section \ref{section:coupling-analysis}}

First, consider rotation sampling (RS).
Recall that for $U \sim \unif[0, 1]$ and $\pi$ a random permutation of $\{1, \dots, k\}$ with $U \indep \pi$, we set $\Di = U \oplus k\inv \pi_i$.
Then $(\Di)_{i=1}^k \sim \grs$ is the RS coupling.
We begin with a lemma characterizing the bivariate marginal $(\Di, \Dj)$ for $i \not = j$.
Recall $a \oplus b \equiv a + b \pmod{1}$ for $a, b \in \mr$.

\begin{lem}\label{lem:rs-distribution}
Let $(\Di)_{i=1}^k \sim \grs$.
Then $(\done, \dtwo) \sim (V, V \oplus R)$, where $V \sim \unif[0, 1]$ and $R \sim  \unif\{1/k, \dots, (k-1)/k\}$ with $V \indep R$.
\end{lem}
\begin{proof}

Observe that $\dtwo = U \oplus k\inv \pi_2 = U \oplus k\inv \pi_1 \oplus k\inv (\pi_2 - \pi_1) = \done \oplus R$ with $R \equiv k\inv (\pi_2 - \pi_1) \pmod 1$.
We claim that $\done \indep R$ and $\done \sim \unif[0, 1]$.
Note that $U + a \pmod 1 \sim \unif[0, 1]$ for any fixed $a \in \mr$.
Then $\done | \pi = U \oplus k\inv \pi_1 | \pi \sim \unif[0, 1]$ since $U \indep \pi$.
Then $\done \sim \unif[0, 1]$.
This also shows $\done \indep \pi$, so $\done \indep R$ since $R \in \sigma(\pi)$, proving the claim.
Then $(\done, \dtwo) = (\done, \done \oplus R) \sim (V, V \oplus R)$ with $V \indep R$ and $V \sim \unif[0, 1]$.
Finally, we show the distribution of $R$.
A calculation shows $\pi_2 | \pi_1$ is uniform on $[k] \setminus \{\pi_1\}$.
Define $f(\pi_2) = \pi_2 - \pi_1 \pmod k$.
This is a bijection from $[k] \setminus \{\pi_1\}$ to $\{1, \dots, k-1\}$.
Then $\pi_2 - \pi_1 \pmod k  | \pi_1 \sim \unif \{1, \dots, k-1\}$.
Then this also holds marginally, completing the proof.
\end{proof}

\begin{lem}[RS] \label{lem:rs-operator}
Let $F = \unif[0, 1]$ and $G$ the RS coupling.
Then $(U_G\phi)(d) = (k-1)\inv \sum_{l=1}^{k-1} \phi(d \oplus lk\inv)$.
Also, $\ltwonot = \ecyclic \oplus \eacyclic$ is a direct sum of eigenspaces of $U_G$ with eigenvalues $1$ on $\ecyclic$ and $-(k-1)\inv$ on $\eacyclic$.
In particular, Assumption \ref{assump:direct-sum} holds.
Finally, the projection operators $P_c$ and $P_a$ on $\elltwo(F)$ have $(P_c \phi)(d) \equiv k\inv \sum_{l=1}^k \phi(d \oplus lk \inv) - \ef[\phi(D)]$ and $(P_a \phi)(d) = \phi(d) - k\inv \sum_{l=1}^k \phi(d \oplus lk \inv)$.
\end{lem}
\begin{proof}

First, we establish the direct sum and projection formulas.
For $\phi \in \elltwo(F)$, define $(P\phi)(d) \equiv k\inv \sum_{l=1}^k \phi(d \oplus lk\inv) - \ef[\phi(D)]$.
We claim that $\pcyclic = P$.
Write $(S\phi)(d) \equiv k\inv \sum_{l=1}^k \phi(d \oplus lk\inv)$, so that $P\phi = S\phi - \ef[\phi(D)]$.
First, we will show that $P\phi$ is cyclic.
For any $m \ge 1$, note
\[
(S\phi)(d \oplus m k\inv) = k\inv \sum_{l \in [k]} \phi(d \oplus lk\inv \oplus m k \inv) = k\inv \sum_{l \in [k]} \phi(d \oplus lk\inv) = (S\phi)(d).
\]
Then $S\phi$ is cyclic, so $P\phi = S\phi - \ef[\phi(D)]$ is as well.
Next, we show $P\phi$ is mean-zero.
Note $\ef[(S\phi)(D)] = k\inv \sum_{l=1}^k \ef[\phi(D \oplus lk\inv)] = \ef[\phi(D)]$, since $D \oplus lk\inv \sim \unif[0,1]$ whenever $D \sim \unif[0,1]$.
Then $\ef[(P\phi)(D)] = \ef[(S\phi)(D)] - \ef[\phi(D)] = 0$.
Then $P\phi \in \ecyclic$ for any $\phi \in \elltwo(F)$.
Next, we show idempotency, $P^2 = P$.
Note if $\phi \in \ecyclic$, then $\phi(d \oplus lk\inv) = \phi(d)$ for all $l$, so $(S\phi)(d) = \phi(d)$.
Also $\ef[\phi(D)] = 0$ since $\phi \in \ecyclic \sub \ltwonot$.
Then $(P\phi)(d) = \phi(d) - 0 = \phi(d)$.
Then for any $\phi \in \elltwo(F)$, $P \phi \in \ecyclic$, so $P(P\phi) = P\phi$, showing $P^2 = P$.
This also implies $\mathrm{Im}(P) = \ecyclic$.
Next, we show $P$ is self-adjoint.
For any $\phi, \psi \in \elltwo(F)$, we have $\lf P\phi, \psi \rf = \lf S\phi, \psi \rf - \ef[\phi(D)] \lf 1, \psi \rf = \lf S\phi, \psi \rf - \ef[\phi(D)] \ef[\psi(D)]$.
Similarly, $\lf \phi, P\psi \rf = \lf \phi, S\psi \rf - \ef[\psi(D)] \ef[\phi(D)]$, so it suffices to show $S$ is self-adjoint.
For $\phi, \psi \in \elltwo(F)$ and $a \ominus b = a-b \pmod 1$,
\begin{align*}
\lf S \phi, \psi \rf &= \int_0^1 (S\phi)(x) \psi(x) dx = \int_0^1 k\inv \sum_{l=1}^k \phi(x \oplus lk\inv) \psi(x) dx \\
&= k\inv \sum_{l=1}^k \int_0^1 \phi(u) \psi(u \ominus l k\inv )  du =  \int_0^1 \phi(u) k\inv \sum_{l=1}^k \psi(u \ominus l k\inv )  du.
\end{align*}
This equals $\int_0^1 \phi(u) (S\psi)(u) du = \lf \phi, S \psi \rf$.
Then $\lf P\phi, \psi \rf = \lf \phi, P\psi \rf$, so $P$ is self-adjoint on $\elltwo(F)$.
Since $P$ is a self-adjoint idempotent on $\elltwo(F)$ with image $\ecyclic$, it is the orthogonal projection onto $\ecyclic$ in $\elltwo(F)$.
Then $\ecyclic$ is closed in $\elltwo(F)$, hence also closed in $\ltwonot$.
Then the projection theorem for Hilbert spaces gives $\ltwonot = \ecyclic \oplus \eacyclic$, where $\eacyclic = \ecyclic^\perp$ is the orthocomplement within $\ltwonot$.
The projection formula $\pcyclic = P$ holds on $\elltwo(F)$.
Note also that $P$ annihilates $\eacyclic$, since for $\phi \in \eacyclic \sub \ltwonot$, $P\phi = S\phi - 0$, and $S\phi \in \ecyclic$ (since $P\phi \in \ecyclic$ and $\ef[\phi(D)] = 0$), so $\lf S\phi, S\phi \rf = \lf S^2 \phi, \phi \rf = \lf S\phi, \phi \rf = 0$ since $\phi \in \eacyclic \perp S\phi$.
Then $S\phi = P\phi = 0$ on $\eacyclic$.
We claim $I-S = \pacyclic$.
To see this, note $\phi = \ef[\phi] + \pcyclic\phi + \pacyclic \phi$.
It's clear $I-S$ annihilates $\linearspan(1) \oplus \ecyclic$.
Then it suffices to show $I-S$ fixes $\eacyclic$.
For $\phi \in \eacyclic$, we showed $S\phi = 0$ above, so $(I-S)\phi = \phi$.
This finishes our characterization of the projections.

\emph{Eigenspaces.}
Next, we characterize the eigenspaces of $U_G$.
Denote $U_G = U$.
Lemma~\ref{lem:rs-distribution} showed that bivariate marginal $(\done, \dtwo) \sim (V, V \oplus R)$, with $V \sim \unif[0, 1]$, $R \sim \unif\{1/k, \dots, (k-1)/k\}$ and $V \indep R$.
Then
\begin{align*}
(U\phi)(d) &= E[\phi(V \oplus R) | V=d] = E[\phi(d \oplus R)] = (k-1)\inv \sum_{l \in [k-1]} \phi(d \oplus lk\inv).
\end{align*}
The second equality since $V \indep R$.
This proves the first claim.
Since $\phi(d \oplus m k\inv) = \phi(d)$ for any $m \geq 1$ and $\phi \in \ecyclic$, then $U \phi(d) = (k-1)\inv \sum_{l=1}^{k-1} \phi(d) = \phi(d)$ and $\ecyclic$ is an eigenspace of $U$ with eigenvalue $1$.
Next, note the identity
\begin{align*}
(U \phi)(d) &= (k-1)\inv \sum_{l\in [k-1]} \phi(d \oplus lk \inv) = (k-1)\inv \big ( \sum_{l \in [k]} \phi(d \oplus lk \inv) - \phi(d) \big) \\
&= (k-1)\inv (k (S\phi)(d) - \phi(d)).
\end{align*}
If $\phi \in \eacyclic$, we showed $S\phi = 0$ above, so $U \phi = -(k-1)\inv \phi$, showing that $\eacyclic$ is an eigenspace with $\eval = -(k-1)\inv$.
This completes the proof.
\end{proof}

\begin{proof}[Proof of Theorem \ref{thm:rs-variance}]
Relative efficiency follows from Lemma~\ref{lem:rs-operator} and Theorem~\ref{thm:eigenspace-decomposition}.
For the marginal variance, by Corollary~\ref{cor:efficiency-convex} and Lemma \ref{lem:var-components}, we have $n \var_G(\est) = (1-\dispg(s^c)) \viid(s^c) + \dispg(s^c) \vmatch(s^c) + (1-\dispg(s^a)) \viid(s^a) + \dispg(s^a) \vmatch(s^a) = k \cdot \viid(s^c) - (k-1) \vmatch(s^c) + \vmatch(s^a) = k \cdot \viid(s^c) - (k-1) (\viid(s^c) - (k-1)\inv c(s^c)) + \vmatch(s^a) = \viid(s^c) + c(s^c) + \vmatch(s^a) = k \cdot \vg(s^c) + \vmatch(s^a)$.
\end{proof}

\begin{lem}[Gaussian]\label{lem:gaussian-operator}
Let $F = \normal(0, 1)$ and $G$ the Gaussian coupling with correlation $\rho = -(k-1)\inv$ and $k \ge 3$.
Write $U_G = U$
Let $(\hm)_{m \geq 0}$ be the normalized probabilist's Hermite polynomials.
Then $L^2(F) = \oplus_{m \geq 0} \linearspan(h_m)$ with $U \hm = \evalm \hm$ for $\evalm = (-1)^m (k-1)^{-m}$.
Restricting to $\ltwonot$ gives $\ltwonot = \oplus_{m \geq 1} \linearspan(h_m)$, verifying Assumption \ref{assump:direct-sum}.
\end{lem}
\begin{proof}
Denote $\rho = -(k-1)\inv$.
Then $\dtwo | \done = x \sim \normal(\rho x, (1-\rho^2))$, so the operator
\begin{align*}
&U\phi(x) = E[\phi(\dtwo) | \done=x] = \int_{\mr} \phi(y) \frac{1}{\sqrt{2\pi(1-\rho^2)}} \exp\left(-\frac{(y - \rho x)^2}{2(1-\rho^2)}\right) dy \\
&= \int_{\mr} \phi(y) \frac{1}{\sqrt{1-\rho^2}} \exp\left(-\frac{\rho^2(x^2+y^2) - 2\rho xy}{2(1-\rho^2)}\right) dF(y) \equiv \int_{\mr} \phi(y) K(y, x) dF(y).
\end{align*}
The third equality follows by algebra, using $dF(y) = (2 \pi)^{-1/2} e^{-y^2/2}$.
The function $K(x,y)$ is the Mehler kernel, with $K(x,y) = \sum_{m=0}^\infty \rho^m h_m(x) h_m(y)$.
For example, see \cite{thangavelu1993}.
Define the finite sum $K_M = \sum_{m=0}^M \rho^m h_m(x) h_m(y)$ and $U_M \phi(x) = \ef[K_M(x, Y)\phi(Y)]$.
For $\phi \in L_2(F)$, we have $|U_M \phi - U \phi|_F \leq |U_M - U|_{op} |\phi|_F$.
Moreover, since $U_M - U$ is a Hilbert-Schmidt operator with kernel $K_M - K$, we have $|U_M - U|_{op} \leq |K_M-K|_{L_2(F \otimes F)}$, where $(K-K_M)(x, y) = \sum_{m=M+1}^\infty \rho^m h_m(x) h_m(y)$.
Note that $e_m(x, y) = \hm(x) \hm(y)$ is an ON collection in $L_2(F \otimes F)$.
Moreover, $K = \sum_m \rho^m e_m$ with $\sum_m (\rho^m)^2 < \infty$ since $|\rho| < 1$ for $k \geq 3$.
Then by Riesz-Fischer, $K_M \to K$ in $L_2(F \otimes F)$, so that $|U_M - U|_{op} \to 0$.
Clearly $U_M h_m = \rho^m h_m$ for $m \leq M$ by orthonormality, so $U \hm = \lim_M U_M \hm = \lim_M \rho^m \hm = \rho^m \hm$.
Then we have shown $L_2(F) = \oplus_{m \geq 0} E_m$ for $E_m = \linearspan(h_m)$ with eigenvalue $\evalm = \rho^m = (-1)^m (k-1)^{-m}$.
Restricting to $\ltwonot$ gives $\ltwonot = \oplus_{m \geq 1} E_m$, since $E_0 = \linearspan(h_0) = \{\text{constants}\}$.
This completes the proof
\end{proof}

\begin{proof}[Proof of Theorem \ref{thm:gaussian-efficiency}]
From the proof of Theorem \ref{thm:eigenspace-decomposition}, we have the variance identity $n\var_G(\est) = \viid(s) + \sum_{m \ge 1} \evalm c(s^m)$.
By Lemma \ref{lem:gaussian-operator}, the eigenspaces correspond to Hermite polynomials with eigenvalues $\evalm = (-1)^m (k-1)^{-m}$.
By Lemma \ref{lem:var-components}, $c(s^m) = (k-1) \viid(s^m) \matchcoeffk(s^m)$.
Substituting these into the variance expression yields
\[
n\var_G(\est) = \viid(s) + \sum_{m \ge 1} (-1)^m (k-1)^{-(m-1)} \viid(s^m) \matchcoeffk(s^m).
\]
Then dividing by $n\var_{\giid}(\est) = \viid(s)$ and using weights $w_m = \viid(s^m) / \viid(s)$, we have
\[
\frac{\var_G(\est)}{\var_{\giid}(\est)} = 1 + \sum_{m \ge 1} (-1)^m (k-1)^{-(m-1)} w_m \matchcoeffk(s^m).
\]
Rearranging for relative efficiency $\releff(G) = 1 - \var_G(\est) / \var_{\giid}(\est)$ gives the result, noting $-(-1)^m = (-1)^{m-1}$.
\end{proof}

\begin{proof}[Proof of Corollary \ref{cor:gaussian-limit}]
We show $n \var_G(\est) = \vmatch(s^L) + \viid(s - s^L) + R_k(s)$ with remainder $|R_k(s)| \leq (k-1)^{-1} \viid(s - s^L)$.
Write $\si^m = \lf \si, \hm \rf \cdot \hm$ for the projection onto $\hm$.
By the proof of Theorem~\ref{thm:eigenspace-decomposition}, $n \var_G(\est) = \sum_m \viid(s^m) + \evalm c(s^m)$.
Then by Lemma \ref{lem:var-components}, since $\vmatch(s) = \viid(s) - (k-1)\inv c(s)$ we calculate
\begin{align*}
&n \var_G(\est) = \vmatch(s^1) + \sum_{m \ge 2} \left( \viid(s^m) + \eval_m c(s^m) \right) \\
&= \vmatch(s^1) + \viid(s - s^1) + \sum_{m \ge 2} (-1)^m (k-1)^{-m} c(s^m).
\end{align*}
The second equality uses $\viid(s - s^1) = \sum_{m \ge 2} \viid(s^m)$, which follows by the Parseval argument displayed later in this proof.
Define $R_k(s) \equiv \sum_{m \ge 2} (-1)^m (k-1)^{-m} c(s^m)$.
For the first two terms, note that $\si^1(D) = \lf h_1, \si \rf \cdot h_1(D) = \lf D, \si \rf \cdot D = \covf(D, \si(D)) \cdot D$ since $\ef[D] = 0$.
We have $\si^L(D) = a_i^* + b_i^* D$, where
\[
\bi^* = \argmin_b \min_a \varf(\si(D) - a - bD) = \argmin_b \varf(\si(D) - bD).
\]
The solution is $\bi^* = \cov_F(\si(D), D) = \lf h_1, \si \rf$.
Then we have shown $\si^L = \ai^* + \si^1$, so $\vmatch(s^1) = \vmatch(s^L)$ and $\viid(s - s^1) = \viid(s - s^L)$.
Summarizing, we showed that $n \var(\est) = \vmatch(s^L) + \viid(s - s^L) + R_k(s)$.
Next, we show the bound on $R_k(s)$.
We have $|c(s^j)| \leq (k-1) \viid(s^j)$.
To see this, recall $c(s^m) = n\inv \sum_g \sum_{i\neq j} \covf(\sig^m,\sjg^m)$.
By the triangle inequality, Cauchy-Schwarz and Young's inequality
\begin{align*}
&|c(s^m)| \le n\inv \sum_g \sum_{i\neq j} |\covf(\sig^m,\sjg^m)| \le n\inv \sum_g \sum_{i\neq j} (1/2)(\varf(\sig^m)+\varf(\sjg^m)) \\
&= n\inv (k-1) \sum_g  \sum_{i \in [k]} \varf(\sig^m) = (k-1) \en \varf(\si^m) = (k-1) \cdot \viid(s^m).
\end{align*}
Then by the triangle inequality, we have
\begin{align*}
&|R_k(s)| \le \sum_{m=2}^\infty (k-1)^{-m} |c(s^m)| \le \sum_{m=2}^\infty (k-1)^{-m} (k-1) \cdot \viid(s^m) \\
&= (k-1)^{-1} \sum_{m=2}^\infty (k-1)^{-(m-2)} \viid(s^m) \le (k-1)^{-1} \sum_{m=2}^\infty \viid(s^m) = (k-1)^{-1} \viid(s - s^L).
\end{align*}
The last equality $\viid(s - s^L) = \sum_{m=2}^\infty \viid(s^m)$ follows since
\begin{align*}
\viid(s - s^L) &= \en \varf\bigg(\sum_{m=2}^\infty \si^m\bigg) = \en \bigg|\sum_{m=2}^\infty \si^m\bigg|_F^2 \\
&= \en \sum_{m=2}^\infty |\si^m|_F^2 = \sum_{m=2}^\infty \en |\si^m|_F^2 = \sum_{m=2}^\infty \viid(s^m).
\end{align*}
The third equality is Parseval's theorem.
The second to last equality is Tonelli's Theorem.
This shows the claimed bound.
Finally, we derive the second (efficiency) display of the corollary.
Since $\si^L = \ai^* + \si^1$ differs from $\si^1$ by a unit-specific constant, we have $\viid(s^L) = \viid(s^1)$ and $\viid(s - s^L) = \viid(s - s^1)$, and Lemma~\ref{lem:variance-components-orthogonal} applied with the closed subspace $H = \linearspan(h_1) \sub \ltwonot$ gives $\viid(s - s^1) = \viid(s) - \viid(s^1)$.
Write $w_L = \viid(s^L) / \viid(s)$.
Dividing the expansion $n \var_G(\est) = \vmatch(s^L) + \viid(s - s^L) + R_k(s)$ by $n \var_{\giid}(\est) = \viid(s)$ gives
\[
1 - \frac{\var_G(\est)}{\var_{\giid}(\est)} = 1 - \frac{\vmatch(s^L)}{\viid(s)} - \bigg(1 - \frac{\viid(s^L)}{\viid(s)}\bigg) - \frac{R_k(s)}{\viid(s)} = w_L \cdot \matchcoeffk(s^L) - \frac{R_k(s)}{\viid(s)},
\]
where the second equality uses $w_L \cdot \matchcoeffk(s^L) = w_L (1 - \vmatch(s^L)/\viid(s^L)) = w_L - \vmatch(s^L)/\viid(s)$.
By the remainder bound above, $|R_k(s)| / \viid(s) \le (k-1)\inv \viid(s - s^L) / \viid(s) \le (k-1)\inv = O(k\inv)$.
This completes the proof.
\end{proof}

\begin{lem}[Antithetic Variates]\label{lem:av-operator}
Let $F = \unif[0, 1]$ and $G$ the antithetic variates coupling.
The coupling operator $U_G\phi(x) = E_G[\phi(D_1) | D_2=x] = \phi(1-x)$.
Also $\ltwonot = \eeven \oplus \eodd$, which are eigenspaces of $U_G$ with eigenvalues $1$ on $\eeven$ and $-1$ on $\eodd$.
In particular, Assumption \ref{assump:direct-sum} holds.
Define $S_e \phi \equiv (1/2)(\phi(x) + \phi(1-x)) - \ef[\phi(D)]$ and $S_o \phi \equiv (1/2)(\phi(x) - \phi(1-x))$.
Then $\peven \phi = S_e \phi$ and $\podd \phi = S_o \phi$.
\end{lem}

The proof is similar to those of Lemma~\ref{lem:rs-operator} (rotation sampling) and Lemma~\ref{lem:lhs-operator} (Latin hypercube), and is deferred to the supplementary online appendix for space.

\begin{lem}[RS Uniform Bound]\label{lem:rs-uniform-bound}
Let $F = \unif[0, 1]$ and $\ecyclic$ be the $1/k$-cyclic subspace.
Then for any $(\si(\cdot))_{i=1}^n \sub \elltwo(F)$ we have $k \cdot \viid(s^c) \le k \inv \en V_{[0, 1]}(\si)^2$.
In particular, $k \cdot \vg(s^c) \le k \inv \en V_{[0, 1]}(\si)^2$.
\end{lem}

\begin{proof}
Our strategy will be to show that $\si^c(x) = k^{-1} \sum_{l \in [k]} \si(x \oplus l/k) - \ef[\si]$ is of order $O(k\inv)$, uniformly over $x \in [0, 1]$.
Define sets $E_l = [(l-1)/k, l/k) \sub [0, 1]$ and $\mc S(x) = \{x \oplus l/k: l \in [k]\}$.
Clearly $|\mc S(x) \cap E_l| = 1$ for $l \in [k]$.
Let $t_l(x)$ denote this point and define the histogram approximation $f_i(x, t) = \sum_{l \in [k]} \one(t \in E_l) \cdot \si(t_l(x))$.
Then we have $\int_0^1 f_i(x, t) dt = k\inv \sum_{l \in [k]} \si(t_l(x)) = k\inv \sum_{l \in [k]} \si(x\oplus l/k)$.
The projection $\delta \si(x) \equiv \si^c(x)$ may be written
\begin{align*}
\delta \si(x) = \si^c(x) = k\inv \sum_{l \in [k]} \si(x\oplus l/k)-\ef[\si] = \int_0^1 [f_i(x, t) - \si(t)] dt
\end{align*}
Note $|f_i(x, t) - \si(t)| \leq \sum_{l \in [k]} \one(t \in E_l) |\si(t_l(x)) - \si(t)| \leq \sum_{l \in [k]} \one(t \in E_l) V_{E_l}(\si)$.
The last inequality follows since $t_l(x), t \in E_l$ on if $\one(t \in E_l) = 1$.
Then above
\begin{align*}
|\delta \si(x)| &\leq \int_0^1 |f_i(x, t) - \si(t)| dt \leq \int_0^1 \sum_{l \in [k]} \one(t \in E_l) V_{E_l}(\si) dt \\
&= k^{-1} \sum_{l \in [k]} V_{E_l}(\si) \leq k \inv V_{[0,1]}(\si).
\end{align*}
In particular, $\sup_{x \in [0, 1]} |\delta \si(x)| \leq k \inv V_{[0,1]}(\si)$.
Then the variance
\[
\varf(\si^c) = \ef[\delta \si(D)^2] \leq \sup_x |\delta \si(x)|^2 \leq k^{-2} V(\si)^2.
\]
Then $k \cdot \viid(s^c) = k \en \varf(\si^c) \leq k \inv \en V_{[0, 1]}(\si)^2$.
The second claim follows since $k \cdot \vg(s^c) \le k \cdot \viid(s^c)$ by Lemma \ref{lem:var-components}.
This finishes the proof.
\end{proof}

\begin{lem}[LHS Uniform Bound]\label{lem:lhs-uniform-bound}
Let $F = \unif[0, 1]$ and let $\ehist$ be the $k$-histogram subspace on bins $E_l = [(l-1)/k, l/k)$ for $l \in [k]$.
Let $\si^{hist} \equiv \phist \si$ be the $L^2(F)$ projection in
Equation~\eqref{eqn:histogram-projection-disp}.
Then for any $(\si(\cdot))_{i=1}^n \sub \elltwo(F)$, we have $\viid(s - s^{hist}) \le k\inv \en V_{[0, 1]}(\si)^2$.
\end{lem}
\begin{proof}
Fix a unit $i$ and write $\si = \si(\cdot)$.
For each bin $E_l$, define the within-bin mean $m_l \equiv E_F[\si(D) \mid D \in E_l] = k \int_{E_l} \si(t)  dt$.
By Equation~\eqref{eqn:histogram-projection-disp}, we have
$\si^{hist}(t) = \sum_{l \in [k]} m_l \one(t \in E_l) - \ef[\si(D)]$.
Define residual $r_i(t) \equiv \si(t) - \sum_{l \in [k]} m_l \one(t \in E_l)$, which differs from $\si - \si^{hist}$ by the constant $\ef[\si(D)]$, so that $\varf(\si - \si^{hist}) = \varf(r_i)$.
Clearly $\ef[r_i(D)] = 0$ so $\varf(r_i) = E_F[r_i(D)^2] = \int_0^1 r_i(t)^2 dt$.
Next, we calculate for $t \in E_l$
\[
|r_i(t)| = |\si(t) - m_l| \le k \int_{E_l} |\si(t) - \si(t')|  dt'
\le \sup_{u,v \in E_l} |\si(u) - \si(v)|
\le V_{E_l}(\si).
\]
Then, $\int_{E_l} r_i(t)^2 dt \le \int_{E_l} V_{E_l}(\si)^2  dt
= k\inv \, V_{E_l}(\si)^2$.
Summing over bins gives
\[
\varf(\si - \si^{hist})
= \int_0^1 r_i(t)^2 dt
\le k\inv \sum_{l \in [k]} V_{E_l}(\si)^2
\le k\inv \Big( \sum_{l \in [k]} V_{E_l}(\si) \Big)^2
\le k\inv V_{[0,1]}(\si)^2.
\]
The second inequality since $(\sum_l a_l)^2 \ge \sum_l a_l^2$ for $a_l \ge 0$, and the final inequality using $\sum_l V_{E_l}(\si) \le V_{[0,1]}(\si)$.
Finally, $\viid(s - s^{hist}) = \en \varf(\si - \si^{hist}) \le k\inv \en V_{[0,1]}(\si)^2$.
This finishes the proof.
\end{proof}

\begin{proof}[Proof of Theorem \ref{thm:dispersion-limits}]
For the first claim, consider $G = \text{RS}$.
By the dispersion formula in Section \ref{subsection:rs}, letting $w_c = \varf(P_c \phi) / \varf(\phi)$, we have
\begin{align*}
\dispg(\phi) &= w_a - (k-1) w_c = (1-w_c) - (k-1)w_c = 1 - k w_c.
\end{align*}
By Lemma~\ref{lem:rs-uniform-bound} applied with $n = k$ and $\si = \phi$ for all $i$, we have $k \varf(P_c \phi) \leq k\inv V_{[0, 1]}(\phi)^2$, where $V_{[0, 1]}(\phi)$ is the total variation of $\phi$ on $[0, 1]$.
Then $\sup_{\phi \in \mc H(b, \eps)} k w_c(\phi) \le b^2 / (k \eps) = o(1)$, so that $\inf_{\phi \in \mc H(b, \eps)} \dispg(\phi) = 1 + o(1)$, as claimed.
Next, consider $G = \text{LHS}$.
By Equation \eqref{eqn:lhs-dispersion-val}, we have $\dispg(\phi) = w_h = \varf(\phist \phi)/\varf(\phi) = 1 - \varf(\phi - \phist \phi)/\varf(\phi)$.
By Lemma \ref{lem:lhs-uniform-bound}, $\varf(\phi - \phist \phi) \le k \inv V_{[0, 1]}(\phi)^2$.
The conclusion follows.
Finally, consider the statement about the Gaussian coupling.
Let $n=k$ and $\sfn_i = \phi$ for all $i \in [k]$.
Then match quality $\matchcoeffk(\phi^L) = 1$ and weights $w_L = \varf(P_L \phi) / \varf(\phi)$.
By Equation \eqref{eqn:pure-dispersion} and Corollary \ref{cor:gaussian-limit}, we have
\[
\dispg(\phi) = 1 - \frac{\var_G(\est)}{\var_{\giid}(\est)} = w_L \cdot 1 + O(k\inv) = \frac{\varf(P_L \phi)}{\varf(\phi)} + o(1).
\]
This proves the second claim.
\end{proof}

\begin{proof}[Proof of Theorem~\ref{thm:structured-worst-case}]
First consider $G = $ LHS.
By Corollary \ref{cor:efficiency-convex}, $n\var_G(\est) = \vmatch(s^{hist}) + \viid(s - s^{hist})$.
Using the orthogonal decomposition $\ltwonot = \ehist \oplus \ehist^\perp$ and Lemma \ref{lem:variance-components-orthogonal}, we have $\vmatch(s^{hist}) = \vmatch(s) - \vmatch(s - s^{hist})$, so $n\var_G(\est) = \vmatch(s) + \viid(s - s^{hist}) - \vmatch(s - s^{hist})$.
Since $\vmatch(\cdot) \ge 0$ by Lemma~\ref{lem:var-components}, we obtain $n\var_G(\est) \le \vmatch(s) + \viid(s - s^{hist})$.
By Lemma~\ref{lem:lhs-uniform-bound}, $\viid(s - s^{hist}) \le k\inv \en[V_{[0,1]}(\si)^2]$.
Dividing by $n\var_{\giid}(\est) = \viid(s)$ and using $\matchcoeffk(s) = 1 - \vmatch(s)/\viid(s)$:
\[
1 - \frac{\var_G(\est)}{\var_{\giid}(\est)} \ge \matchcoeffk(s) - \frac{\en[V_{[0,1]}(\si)^2]}{k \cdot \viid(s)} = \matchcoeffk(s) - \frac{\etatv(s)}{k}.
\]
Next consider $G = $ RS.
By the proof of Theorem~\ref{thm:rs-variance}, $n\var_G(\est) = \vmatch(s^a) + k \cdot \vg(s^c)$.
Using the decomposition $\ltwonot = \ecyclic \oplus \eacyclic$ and Lemma \ref{lem:variance-components-orthogonal}, $\vmatch(s^a) = \vmatch(s) - \vmatch(s^c)$, so
\[
n\var_G(\est) = \vmatch(s) + k \cdot \vg(s^c) - \vmatch(s^c).
\]
Since $\vmatch(s^c) \ge 0$, we obtain $n\var_G(\est) \le \vmatch(s) + k \cdot \vg(s^c)$.
By Lemma \ref{lem:var-components}, $k \cdot \vg(s^c) \le k \cdot \viid(s^c)$, and by Lemma \ref{lem:rs-uniform-bound}, $k \cdot \viid(s^c) \le k\inv \en[V_{[0,1]}(\si)^2]$.
Dividing by $\viid(s)$ as before gives $1 - \var_G(\est)/\var_{\giid}(\est) \ge \matchcoeffk(s) - \etatv(s)/k$.
\end{proof}

\subsection{Proofs of Uniform Consistency Results}

\begin{proof}[Proof of Theorem~\ref{thm:uniform-consistency-rates}]
First, we show the main result in Equation \eqref{eqn:maxrate-sup-inf}.
We begin with an upper bound on the worst case variance, then show it is attained for some functions $\si$.
For eigenspace $E_m$ of $U_G$, let $s^m = (P_m \si)_{i=1}^n$ be the orthogonal projections and $\evalm$ the eigenvalue.
From the proof of Theorem~\ref{thm:eigenspace-decomposition}, we have the key equation $n\var_G(\est) = \viid(s) + \sum_{m \ge 1} \evalm \cdot c(s^m)$.
By Lemma~\ref{lem:var-components}, $-\viid(s^m) \le c(s^m) \le (k-1)\viid(s^m)$ for each $m$.
If $\evalm \ge 0$, then $\evalm c(s^m) \le (k-1)\evalm \viid(s^m)$.
If $\evalm < 0$, then $\evalm c(s^m) \le -\evalm \viid(s^m)$.
Recall from Theorem~\ref{thm:dispersion-decomposition} that $\dispg(m) = -(k-1)\evalm$.
Putting this together, we have
\[
\evalm \cdot c(s^m) \le \max(-\dispg(m), \dispg(m)/(k-1)) \cdot \viid(s^m).
\]
Then from the key equation above and using $\viid(s) = \sum_m \viid(s^m)$, which follows by Parseval's theorem as in the proof of Theorem~\ref{thm:eigenspace-decomposition}:
\begin{align*}
n\var_G(\est) &\le \sum_m [1 + \max(-\dispg(m), \dispg(m)/(k-1))] \cdot \viid(s^m).
\end{align*}
Next, we claim that $\inf_{m \ge 1} \dispg(m) = \infdispg$ and $\sup_{m \ge 1} \dispg(m) = \supdispg$.
Write $\phi \ne c$ for $\phi$ non-constant.
By Theorem~\ref{thm:dispersion-decomposition}, for $\sum_{m \ge 1} w_m = 1$, we have $\dispg(\phi) = \sum_{m \ge 1} w_m \dispg(m) \ge \inf_m \dispg(m)$.
Then $\inf_{\phi \ne c}\dispg(\phi) \ge \inf_m \dispg(m)$.
Conversely, picking $\phi_m \in E_m$ non-constant for $m \ge 1$, we have $\inf_{\phi \ne c} \dispg(\phi) \le \inf_m \dispg(\phi_m) = \inf_m \dispg(m)$.
This proves the claim for the inf, and the sup follows by symmetry.
Using this, we have $-\dispg(m) \le - \inf_m \dispg(m) = - \infdispg$ and $\dispg(m) \le \sup_m \dispg(m) = \supdispg$, so the above is
\begin{align*}
n\var_G(\est) &\le \sum_m [1 + \max(-\infdispg, \supdispg/(k-1))] \cdot \viid(s^m) \\
&= [1 + \max(-\infdispg, \supdispg/(k-1))] \cdot \viid(s).
\end{align*}
This establishes the upper bound.

Next we show that the bound is achieved.
\emph{Case 1:}
Suppose $-\infdispg \ge \supdispg/(k-1)$.
Choose $l$ with $\dispg(l) < \infdispg + \epsilon$.
Pick $\phi \in E_{l}$ with $\varf(\phi) = 1$ and set $\si = \phi$ for all $i \in [n]$.
Then $\viid(s) = 1$ and, since all influence functions are identical,
we have $c(s^l) = n\inv \sum_g \sum_{i \neq j} \covf(\phi, \phi) = (k-1)\varf(\phi) = (k-1)$ and $c(s^m) = 0$ for $m \ne l$ by Lemma \ref{lem:variance-components-orthogonal}.
Since only $s^{l}$ is nonzero, the key equation above gives
\[
n\var_G(\est) = 1 + (k-1)\eval_{l} = 1 - \dispg(l) > 1 - \infdispg - \epsilon.
\]
\emph{Case 2:}
Suppose $\supdispg/(k-1) > -\infdispg$.
Choose $l$ with $\dispg(l) > \supdispg - \epsilon$.
Pick $\phi \in E_{l}$ with $\varf(\phi) = 1$ and choose coefficients $\alpha_1, \ldots, \alpha_k \in \mr$ satisfying $\sum_{i=1}^k \alpha_i = 0$ and $k\inv \sum_{i=1}^k \alpha_i^2 = 1$.
Within each group $g$, set $\sig(d) = \alpha_i \phi(d)$ for the $i$th unit.
Then $\viid(s) = k\inv \sum_i \alpha_i^2 \varf(\phi) = 1$, and
\[
c(s) = k\inv \sum_{i \neq j} \alpha_i \alpha_j \varf(\phi) = k\inv \bigg[\bigg(\sum_i \alpha_i\bigg)^{\!2} - \sum_i \alpha_i^2 \bigg] = -1.
\]
Thus $n\var_G(\est) = 1 - \eval_{l} = 1 + \dispg(l)/(k-1) > 1 + \supdispg/(k-1) - \epsilon/(k-1)$.
Since $\epsilon > 0$ is arbitrary, the supremum equals $1 + \max(-\infdispg,\; \supdispg/(k-1))$.

Next, we prove the remaining claims of the theorem, together with the claims about the iid and RS couplings in Remark~\ref{rem:minimaxity}.
Recall from Section~\ref{subsection:dispersion} the dispersion bounds $-(k-1) \le \dispg(\phi) \le 1$.
For the upper bound, by \eqref{eqn:maxrate-sup-inf}, $\maxraten(G) = 1 + \max(-\infdispg,\; \supdispg/(k-1)) \le 1 + \max(k-1,\; 1/(k-1)) = 1 + (k-1) = k$.
For the lower bound, we show $\max(-\infdispg,\; \supdispg/(k-1)) \ge 0$.
If $\infdispg \le 0$, then $-\infdispg \ge 0$ and the conclusion is immediate.
If $\infdispg > 0$, then $\supdispg \ge \infdispg > 0$, so $\supdispg/(k-1) > 0$.
In either case $\maxraten(G) \ge 1$.
Minimaxity of $\giid$ is immediate since $\maxraten(\giid) = 1 + \max(0, 0) = 1$.
For continuous $F$, the transport $v$ of Lemma~\ref{lem:canonical-space} is injective, so by that lemma the operator $U_G$ has the same eigenvalues as the canonical operator $T_H$. We compute $\infdispg$ and $\supdispg$ from these eigenvalues below.
For $G = $ RS, Lemma~\ref{lem:rs-operator} gives eigenvalues $\eval = 1$ on $\ecyclic$ and $\eval = -(k-1)\inv$ on $\eacyclic$, so by Theorem~\ref{thm:dispersion-decomposition} the dispersions are $\dispg(\ecyclic) = -(k-1)$ and $\dispg(\eacyclic) = 1$.
Hence $\infdispg = -(k-1)$ and $\supdispg = 1$, giving $\maxraten(G) = 1 + \max(k-1,\; 1/(k-1)) = k$.
Finally, if $k = O(1)$, then $\maxraten(G) \le k = O(1)$ for any coupling sequence, so $\est$ is uniformly $\rootn$-consistent.

Finally, we prove the claims about the LHS and Gaussian couplings in Remark~\ref{rem:minimaxity}.
The argument above showed $\infdispg = \inf_m \dispg(m)$ and $\supdispg = \sup_m \dispg(m)$, with $\dispg(m) = -(k-1)\evalm$ by Theorem~\ref{thm:dispersion-decomposition}.
Substituting the eigenvalues of Lemma~\ref{lem:lhs-operator} gives $\infdispg = 0$ and $\supdispg = 1$ for LHS, and substituting the eigenvalues of Lemma~\ref{lem:gaussian-operator} gives $\infdispg = -(k-1)\inv$ and $\supdispg = 1$ for the Gaussian coupling.
In both cases Equation \eqref{eqn:maxrate-sup-inf} gives $\maxraten(G) = 1 + \max(-\infdispg, \supdispg/(k-1)) = k/(k-1) = O(1)$, so uniform $\rootn$-consistency holds for any sequence $k(n)$.
If $k \to \infty$, then $\maxraten(G) \to 1$, and since $\maxraten(G) \ge 1$ as shown above, both designs are asymptotically minimax.
\end{proof}

\begin{proof}[Proof of Theorem~\ref{thm:consistency-under-regularity}]
By Theorem~\ref{thm:structured-worst-case}, for $G \in \{\text{LHS}, \text{RS}\}$, we have $1 - \var_G(\est)/\var_{\giid}(\est) \ge \matchcoeffk(s) - k\inv \etatv(s)$.
Rearranging and using $n \var_{\giid}(\est) = \viid(s)$, we have $n \var_G(\est) \le \viid(s)(1 - \matchcoeffk(s) + k\inv \etatv(s))$.
For $s \in \sreg$, we have $\viid(s) \le 1$, $\matchcoeffk(s) \ge q_0$, and $\etatv(s) \le \bar \eta$, so $n \var_G(\est) \le 1 - q_0 + k\inv \bar \eta$.
Taking the supremum over $s \in \sreg$ completes the proof.
\end{proof}

\section{Asymptotics Proofs}

\subsection{Proof of Proposition~\ref*{prop:clt}}

\begin{proof}[{Proof of Proposition~\ref*{prop:clt}}]
For each matched group $g \in [n/k]$, define the centered scaled group sum $S_g = n^{-1} \sum_{i=1}^{k} [\sig(\Di) - \ef[\sig(D)]]$.
Since treatments are independent across matched groups, $\{S_g\}_{g=1}^{n/k}$ are independent with $\est - \thetan = \sum_{g=1}^{n/k} S_g$ and $\hksd_n^2 = \var_G(\est) = \sum_{g=1}^{n/k} \var_G(S_g)$.
We verify the Lindeberg condition: $\hksd_n^{-2} \sum_{g=1}^{n/k} \eg\bigl[S_g^2 \one(|S_g| > \varepsilon \hksd_n)\bigr] \to 0$ for every $\varepsilon > 0$.
By Markov's inequality, $\eg\bigl[S_g^2 \one(|S_g| > \varepsilon \hksd_n)\bigr] \leq \eg[S_g^4] / \varepsilon^2 \hksd_n^2$.
By the power mean inequality applied pointwise,
\begin{align*}
\eg[S_g^4]
= \frac{1}{n^4} \eg\Bigg[\bigg(\sum_{i=1}^{k} [\sig(\Di) - \ef[\sig(D)]]\bigg)^{\!4}\,\Bigg]
&\leq \frac{k^3}{n^4} \sum_{i=1}^{k} \ef\bigl[(\sig(D) - \ef[\sig(D)])^4\bigr] \\
&\leq \frac{16 k^3}{n^4} \sum_{i=1}^{k} \ef[\si(D)^4].
\end{align*}
The first inequality also uses that $\Di \sim F$ marginally under $G$.
The second uses $\ef[(\sig(D) - \ef[\sig(D)])^4] \leq 16 \ef[\sig(D)^4]$, which follows since $\ef[(X - \ef[X])^4]^{1/4} \leq \ef[X^4]^{1/4} + |\ef[X]| \leq 2 \ef[X^4]^{1/4}$ by Minkowski's and Jensen's inequalities.
Summing over all groups,
\[
\frac{1}{\hksd_n^2} \sum_{g=1}^{n/k} \eg\bigl[S_g^2 \one(|S_g| > \varepsilon \hksd_n )\bigr]
\leq \frac{16 k^3}{\varepsilon^2 n^4 \hksd_n^4} \sum_{i=1}^{n} \ef[\si(D)^4]
= \frac{16 k^3 M_{4,n}}{\varepsilon^2 n^3 \hksd_n^4}
\]
for $M_{4,n} = \en\bigl[\ef[\si(D)^4]\bigr]$.
By parts 1 and 2 of Assumption~\ref{assumption:clt}, $M_{4,n} = O(1)$ and $\hksd_n^2 = \Omega(n^{-1})$, so $1/\hksd_n^4 = O(n^2)$ and the bound above is $O(k^3 / n) \to 0$, using $k^3 / n \to 0$ by part 3.
The Lindeberg--Feller CLT can therefore be applied using the group sums $S_g$ as the independent variates.
\end{proof}

\begin{proof}[{Proof of Corollary~\ref*{cor:parametric-clt}}]
Write $c'(\wh\beta - \betan) = \en[c'\si(\Di)] + R_n$ with $R_n = \Op(n\inv)$.
The estimator $\en[c'\si(\Di)]$ has estimand $\en \ef[c'\si(D)] = 0$, so Proposition~\ref{prop:clt} applied with the influence functions $(c'\si)_{i=1}^n$ gives $\en[c'\si(\Di)] / \hksd_n \convd \normal(0, 1)$.
By part 2 of Assumption~\ref{assumption:clt}, $\hksd_n^2 = \Omega(n\inv)$, so $\hksd_n\inv = O(n^{1/2})$ and $R_n / \hksd_n = \Op(\negrootn) = \op(1)$.
The claim then follows from Slutsky's theorem.
\end{proof}

\subsection{Proof of Proposition~\ref*{prop:var-est-expectation-variance}}

\begin{proof}[{Proof of Proposition~\ref*{prop:var-est-expectation-variance}}]
\emph{Expectation.}
Because groups $g$ and $\pi(g)$ are independent (since $\pi$ has no fixed points), we have $\eg[\est_{g} \est_{\pi(g)}] = \eg[\est_{g}] \eg[\est_{\pi(g)}] = \theta_{g} \theta_{\pi(g)}$.
Therefore,
\begin{align*}
\eg\big[(\est_{g} - \est_{\pi(g)})^2\big]
&= \eg[\est_{g}^2] + \eg[\est_{\pi(g)}^2] - 2 \theta_{g} \theta_{\pi(g)} \\
&= \var_G(\est_{g}) + \var_G(\est_{\pi(g)}) + (\theta_{g} - \theta_{\pi(g)})^2.
\end{align*}
Summing over $g$ and using that $\pi$ is a permutation, $\sum_{g=1}^{n/k} \eg[(\est_{g} - \est_{\pi(g)})^2] = 2 \sum_{g=1}^{n/k} \var_G(\est_{g}) + \sum_{g=1}^{n/k} (\theta_{g} - \theta_{\pi(g)})^2$.
Since $\est = (n/k)^{-1} \sum_{g} \est_{g}$ and the groups are independent, $\hksd_n^2 = \var_G(\est) = (n/k)^{-2} \sum_{g=1}^{n/k} \var_G(\est_{g})$, so $\sum_{g} \var_G(\est_{g}) = (n/k)^2 \hksd_n^2$.
Substituting and using $\sum_{g} (\theta_{g} - \theta_{\pi(g)})^2 = (n/k) \Delta_n^2$,
\[
\eg[\widehat{\hksd}_n^2]
= \frac{k^2}{2n^2} \bigg[ 2 (n/k)^2 \hksd_n^2 + (n/k) \Delta_n^2 \bigg]
= \hksd_n^2 + \frac{k \Delta_n^2}{2n}.
\]
Dividing by $\hksd_n^2$ gives the expectation result.

\emph{Variance.}
Let $Q_g = (\est_{g} - \est_{\pi(g)})^2$.
Since $Q_g$ depends only on groups $g$ and $\pi(g)$, the terms $Q_g$ and $Q_h$ are independent unless $\{g, \pi(g)\} \cap \{h, \pi(h)\} \neq \emptyset$.
For each $g$, the dependent indices are $h \in \{g, \pi(g), \pi^{-1}(g)\}$, so
\[
\var_G(\widehat{\hksd}_n^2)
= \frac{k^4}{4n^4} \sum_{g = 1}^{n / k} \sum_{h \in \{g, \pi(g), \pi^{-1}(g)\}} \cov_G(Q_g, Q_h).
\]
By the Cauchy--Schwarz and AM--GM inequalities, $\cov_G(Q_g, Q_h) \leq [\var_G(Q_g) + \var_G(Q_h)] / 2$.
Each inner sum has at most three terms, and reindexing via the permutation $\pi$ preserves the sum, so $\var_G(\widehat{\hksd}_n^2) \leq (3k^4 / 4n^4) \sum_{g} \var_G(Q_g)$.
Next, $\var_G(Q_g) \leq \eg[Q_g^2] = \eg[(\est_{g} - \est_{\pi(g)})^4] \leq 8(\eg[\est_{g}^4] + \eg[\est_{\pi(g)}^4])$, where the last step uses the convexity of $x \mapsto |x|^4$.
Because $\pi$ is a permutation, $\sum_{g} (\eg[\est_{g}^4] + \eg[\est_{\pi(g)}^4]) = 2 \sum_{g} \eg[\est_{g}^4]$, so $\var_G(\widehat{\hksd}_n^2) \leq (12 k^4 / n^4) \sum_{g} \eg[\est_{g}^4]$.
By Jensen's inequality applied to $x \mapsto |x|^4$, we have $\eg[\est_{g}^4] = \eg[(k^{-1} \sum_{i=1}^{k} \sig(\Di))^4] \leq k^{-1} \sum_{i=1}^{k} \eg[\sig(\Di)^4]$.
Therefore,
\[
\var_G(\widehat{\hksd}_n^2)
\leq \frac{12 k^3}{n^4} \sum_{i = 1}^{n} \ef[\si(D)^4]
= \frac{12 k^3}{n^3} M_{4,n}.
\]
Dividing by $\hksd_n^4$ gives the variance result.
\end{proof}

\subsection{Proof of Proposition~\ref*{prop:ci}}

\begin{proof}[{Proof of Proposition~\ref*{prop:ci}}]
Let $T_n = (\est - \thetan) / \hksd_n$ and $R_n = \widehat{\hksd}_n / \hksd_n$.
By Proposition~\ref{prop:clt}, $T_n \convd \normal(0, 1)$, and by Proposition~\ref{prop:var-est-expectation-variance}, $\eg[R_n^2] \geq 1$ and $\var_G(R_n^2) = o(1)$.
Note that $\thetan \in \mathrm{CI}_{1 - \alpha}$ if and only if $|T_n| \leq z_{1 - \alpha/2}\, R_n$.
For any $\delta > 0$,
\begin{align*}
\pg(|T_n| \leq z_{1 - \alpha/2}\, R_n)
&\geq \pg(|T_n| \leq z_{1 - \alpha/2} \sqrt{1 - \delta},\; R_n^2 \geq 1 - \delta) \\
&\geq \pg(|T_n| \leq z_{1 - \alpha/2} \sqrt{1 - \delta}) - o(1).
\end{align*}
The second inequality follows since $\eg[R_n^2] \geq 1$, so Chebyshev's inequality gives $\pg(R_n^2 < 1 - \delta) \leq \var_G(R_n^2) / \delta^2 = o(1)$.
Taking limits and using $T_n \convd \normal(0, 1)$, we obtain $\liminf_{n \to \infty} \pg(\thetan \in \mathrm{CI}_{1 - \alpha}) \geq \pg(|Z| \leq z_{1 - \alpha/2} \sqrt{1 - \delta})$ for $Z \sim \normal(0, 1)$.
Since this holds for all $\delta > 0$, taking $\delta \to 0$ gives $\liminf_{n} \pg(\thetan \in \mathrm{CI}_{1 - \alpha}) \geq 1 - \alpha$.
\end{proof}

\subsection{Lemmas} \label{section:lemmas}

Define $\viid(s) \equiv \en \varf(\si)$ and $\vmatch(s) \equiv (2n(k-1))\inv \sum_g \sum_{i\neq j} \varf(\sig - \sjg)$.
Also let $\sgbar \equiv k\inv \sum_{i\in[k]} \sig$ and define $\vg(s) \equiv (n/k) \inv \sum_g \varf(\sgbar)$.
Also define $c(s) \equiv n\inv \sum_g \sum_{i\neq j} \covf(\sig,\sjg)$.
The proofs of the following lemmas are algebra exercises, omitted for space.

\begin{lem}[Variance Identities]\label{lem:var-components}
We have $\vmatch(s) = \viid(s) - (k-1)\inv c(s)$ and $k \cdot \vg(s) = \viid(s) + c(s)$.
Moreover, $k \vg(s) \le k \viid(s)$ and $\vmatch(s) \leq 2 \viid(s)$.
In particular, $-\viid(s) \le c(s) \le (k-1)\viid(s)$.
The match quality coefficient $\matchcoeffk(s) = (k-1)\inv c(s) / \viid(s)$.
Consequently, $\matchcoeffk(s) \in [-1/(k-1),\, 1]$.
Finally, for $k \ge 2$, we have $\vmatch(s) = (k-1)\inv (n/k)\inv \sum_g \sum_{i\in[k]} \varf(\sig - \sgbar )$.
\end{lem}

\begin{lem}[Orthogonal Decomposition]\label{lem:variance-components-orthogonal}
Let $H$ be a closed linear subspace of $\ltwonot$ and let $\sih = P_H \si$ be the orthogonal projection of $\si$ onto $H$ in $\elltwo(F)$.
Then (1) $\viid(s) = \viid(\sh) + \viid(s - \sh)$, (2) $c(s) = c(\sh) + c(s - \sh)$, and (3) $\vmatch(s) = \vmatch(\sh) + \vmatch(s - \sh)$.
Finally, for $w_H = \viid(\sh) / \viid(s)$, we have
\begin{equation*}
\matchcoeffk(s) = w_H \cdot \matchcoeffk(\sh) + (1-w_H) \cdot \matchcoeffk(s - \sh).
\end{equation*}
\end{lem}

%% ======================= Supplementary Online Appendix =======================
%% Removed from the submission build to meet the page limit; this material is
%% posted as an online-only supplement. It is self-contained (restarts page
%% numbering and the theorem/equation counters). To compile the full paper WITH
%% the supplement, simply uncomment the \input line below.
\clearpage

\pagenumbering{arabic}\renewcommand{\thepage}{\arabic{page}}

\begin{center}
{\Large Supplementary Online Appendix to ``Coupling Designs for Randomized Experiments with Complex Treatments''}
\vskip 24pt
{\large Max Cytrynbaum and Fredrik S{\"a}vje}
\end{center}

This supplementary appendix is not intended for publication.

\setcounter{thm}{0}
\setcounter{equation}{0}

\section{Supplementary Results and Proofs}

\subsection{Details of the Savings-Monitor Application} \label{appendix:matching-detail}

This appendix collects details of the application of Section~\ref{subsection:savings-monitors-application}: its data-generating process and its two matching procedures.

\emph{Data-generating process.}
Each saver has two baseline covariates $X_i \sim N(0, I_2)$, and its three response coefficients are generated from them.
Writing $r$ for the regime's variance share, the intercept, centrality, and curvature coefficients are $\beta_{i0} = 0.2\,\zeta_{i0}$, $\beta_{i1} = 0.5 + 9\,\zeta_{i1}$, and $\beta_{i2} = 18\,\zeta_{i2}$, where $\zeta_{ij} = \sqrt{r}\, s_{ij} + \sqrt{1-r}\, e_{ij}$ mixes a unit-variance signal $s_{ij}$ (the covariate $X_{i1}$ for the centrality coefficient, $X_{i2}$ for the curvature coefficient, and their average for the intercept) with independent standard-normal noise $e_{ij}$.
By construction the covariates account for exactly a fraction $r$ of each coefficient's variance.
Equivalently, each covariate has correlation $\rho = \sqrt{r}$ with the coefficient it drives, so $\rho \approx 0.45$ in the weakly predictive regime ($r = 0.20$) and $\rho \approx 0.77$ in the highly predictive one ($r = 0.60$).
The treatment basis $t(d) = \big(1,\ 2(d_1 - \tfrac12),\ 4[(d_1 - \tfrac12)^2 + (d_2 - \tfrac12)^2]\big)$ collects an intercept, a centered linear centrality term, and a radial curvature term, the latter two scaled so they vary over comparable ranges on $[0,1]^2$.
The small intercept dispersion $0.2$ sits well below the sensitivity scale $9$, keeping savers similar in overall savings level, which no design can help estimate, so the informative object is the dose-response profile.
The curvature coefficient is twice the sensitivity scale, so the surface is more sharply curved and dispersing monitors over the feature space is more valuable.
The features of the drawn monitors are jittered by independent $N(0, 0.01^2)$ noise to break ties between households with identical network statistics.
Finally, the target marginal $F$ places a $10 \times 10$ grid over $[0,1]^2$ and splits equal mass among the monitors in each occupied cell, and the transport weights are fit by entropic optimal transport, annealing the regularization down to $10^{-3}$, at which the hard Laguerre map reproduces $F$ to within a few percent in total variation.

\emph{Matching step.}
The two procedures each produce a partition of savers into tuples of size $k$ and then apply the identical assignment step, so they differ only in which savers share a tuple.
Let $X_i \in \mr^2$ collect a saver's baseline covariates, standardized coordinate-wise.
Following the matched-tuples construction of \citet{cytrynbaum2023}, we sort savers along a boustrophedon (serpentine) space-filling curve through the covariate space and cut the sorted list into consecutive blocks of $k$, so that each tuple collects $k$ savers that are close in covariate space.
We omit the algorithm's optional balanced-$k$-means polishing refinement.
Under \emph{within-village} matching, this sort-and-block step is applied separately to the $64$ savers of each village, producing $64/k$ tuples per village whose members all reside in the same village.
Under \emph{pooled} matching, the step is applied once to all $n = 4800$ savers, producing $n/k$ tuples that may combine savers from different villages.
Pooled matching can only improve match quality, since the pool of candidates for each tuple is strictly larger, though its tuples no longer respect village boundaries.

\emph{Assignment step.}
Fix a matched tuple $(i_1, \dots, i_k)$.
The design draws one realization $(U_1, \dots, U_k)$ of the chosen uniform coupling on $[0,1]^2$ (Latin hypercube, shifted lattice, or scrambled digital net), exchangeably assigned to the tuple's members, and member $i_j$ receives the monitor $T_{v(i_j)}(U_j)$, where $v(i)$ denotes saver $i$'s village and $T_v$ is the semi-discrete optimal-transport (Brenier) map from $\unif[0,1]^2$ to village $v$'s catalog with target marginal $F$ (Section~\ref{subsection:transport}).
Draws are independent across tuples.
Because every village's map pushes the uniform distribution to (approximately) its catalog's marginal $F$, each saver's monitor has the same marginal distribution under both procedures, and the two procedures differ only in the joint distribution of assignments within tuples: within-village tuples disperse their monitors over a single village's catalog, while pooled tuples disperse the members' uniforms jointly and each member's uniform is realized in the member's own village.
Stratified assignment and iid randomization are as described in the main text.
Under iid each saver receives an independent uniform, which corresponds to $k = 1$.

\emph{Inference.}
The collapsed-strata variance estimator of Proposition~\ref{prop:ci} is computed on the matched tuples of the design in use (within-village or pooled), pairing adjacent tuples in the space-filling order, with the critical value from a $t$ distribution whose degrees of freedom equal half the total number of tuples, a finite-sample refinement of the proposition's normal quantile.

\subsection{Proofs for Section \ref*{section:overview}} \label{section:proofs-1}

\begin{proof}[Proof of Equation \eqref{eqn:relative-efficiency-toy}]
Under perfect matching $\yi(\cdot) = \yj(\cdot) = y(\cdot)$, the estimator is $\est = (1/2)(y(\Di) + y(\Dj))$ with $\thetatrue = \ef[y(D)]$.
Since $G_i = G_j = F$, both $y(\Di)$ and $y(\Dj)$ have variance $\varf(y(D))$, so $\var_G(\est)$ is
\begin{align*}
&= (1/4)\var_G(y(\Di) + y(\Dj)) = (1/4)[2\varf(y(D)) + 2\covg(y(\Di), y(\Dj))] \\
&= (1/2)\varf(y(D))[1 + \corrg(y(\Di), y(\Dj))].
\end{align*}
Under iid randomization, $\cov(y(\Di), y(\Dj)) = 0$, so $\var_{\giid}(\est) = (1/2)\varf(y(D))$.
Dividing through gives the result.
\end{proof}

\begin{proof}[Proof of Matched Pairs Equivalence]
Let $F = \bern(1/2)$, so $F\inv(u) = \one(u > 1/2)$.
Under the antithetic variates coupling, $U \sim \unif[0,1]$ and $\Di^* = F\inv(U)$, $\Dj^* = F\inv(1-U)$.
If $U < 1/2$, then $\Di^* = 0$ and $1-U > 1/2$ so $\Dj^* = 1$.
If $U > 1/2$, then $\Di^* = 1$ and $1-U < 1/2$ so $\Dj^* = 0$.
Since $U = 1/2$ occurs with probability zero, almost surely $\Di^* = 1 - \Dj^*$, and $(\Di^*, \Dj^*)$ is supported on $\{(0,1), (1,0)\}$ each with probability $1/2$.
This is exactly the matched pairs distribution.
\end{proof}

\subsection{Proofs for Section \ref*{section:coupling-designs}}\label{section:proofs-coupling-designs}

Throughout this subsection, we maintain the design setting of Section~\ref{section:coupling-designs}: treatments are assigned by a coupling design with $G \in \transportsym(F)$, drawn independently across the matched groups.

\begin{proof}[Proof of Lemma \ref{lem:generic-rootn-consistency}]
Since treatments are independent across groups and $\est = n\inv \sum_g \sum_{i \in [k]} \sig(\dig)$,
\[
n \var_G(\est) = \viid(s) + n\inv \sum_g \sum_{i \neq j \in [k]} \covg(\sig(\dig), \sjg(\djg)).
\]
By the fixed marginals $G_i = F$, each $\sig(\dig)$ has variance $\varf(\sig)$, so Cauchy-Schwarz and Young's inequality give $|\covg(\sig(\dig), \sjg(\djg))| \le \varf(\sig)^{1/2} \varf(\sjg)^{1/2} \le \tfrac{1}{2}(\varf(\sig) + \varf(\sjg))$.
Summing over the pairs, $n\var_G(\est) \le \viid(s) + (k-1)\viid(s) = k \viid(s)$ for any $G \in \transportsym(F)$.
Since $k = O(1)$ and $\viid(s) = O(1)$, we have $\var_G(\en[\si(\Di)]) = O(n\inv)$.
By Chebyshev's inequality, $\en[\si(\Di)] - \thetan = \Op(\negrootn)$.
\end{proof}

\begin{proof}[Proof of Identification in Example \ref{ex:blp}]
First, consider $\thetan$.
We have
\begin{align*}
\thetan &= \en\ef[\si(D)] = \ef[\ynbar(D) H(D)] = \varf(D)\inv \ef[\ynbar(D)(D - \ef[D])] \\
&= \varf(D)\inv \covf(\ynbar(D), D).
\end{align*}
Then by concentrating out, $\thetablp = \argmin_{\theta \in \mr^m} \min_{\alpha} \ef[(\ynbar(D) - \alpha - \theta'D)^2] = \argmin_{\theta \in \mr^m} \varf(\ynbar(D) - \theta'D)$.
The FOC gives $\varf(D)\thetablp = \covf(\ynbar(D), D)$, so $\thetablp = \varf(D)\inv \covf(\ynbar(D), D) = \thetan$.
\end{proof}

\begin{proof}[Proof of Proposition \ref{prop:ols-linearization}]

By Frisch-Waugh, the OLS coefficient $\wh \beta = \varn(\Di)\inv \covn(\Di, \yi)$.
Define $y_n = \en \ef[\yi(D)]$.
We have
\begin{align*}
\covn(\Di, \yi) &= \en[(\Di - \en[\Di])\yi] = \en[(\Di - \en[\Di])(\yi - y_n)] \\
&= \en[(\Di - \ef[D] + \ef[D] - \en[\Di])(\yi - y_n)] \\
&= \en[(\Di - \ef[D])(\yi - y_n)] + (\ef[D] - \en[\Di])(\en[\yi] - y_n)
\end{align*}
Write this as $\covn(\Di, \yi) = A_n + B_n$.
By Lemma \ref{lem:generic-rootn-consistency} and the conditions of Proposition \ref{prop:ols-linearization}, we have $\en[\Di] = \ef[D] + \Op(\negrootn)$ and $\en[\yi(\Di)] = y_n + \Op(\negrootn)$.
Then $B_n = \Op(n\inv)$.
\begin{align*}
A_n &= \en[(\Di - \ef[D])(\yi - y_n - \thetan'(\Di-\ef[D]))] \\
&+ \en[(\Di - \ef[D])(\Di-\ef[D])'] \thetan = A_n^1 + A_n^2
\end{align*}
Define the residual $\ei(d) = \yi(d) - y_n - \thetan'(d - \ef[D])$ and $\ri(d) = (d - \ef[D])\ei(d)$, so that $A_n^1 = \en[\ri(\Di)]$.
We claim $A_n^1 = \Op(\negrootn)$ and $A_n^2 = \varn(\Di) \thetan +  \Op(n\inv)$.
Note that for $A_n^1$,
\begin{align*}
\en \ef[(D-\ef[D])\ei(D)] &= \en \covf(D, \yi(D)) - \varf(D)\thetan = 0.
\end{align*}

To apply Lemma \ref{lem:generic-rootn-consistency} we show $\en \ef[|\ri(D)|_2^2]= O(1)$.
This is $\ef[|\ri(D)|_2^2] = \ef[|D-\ef[D]|_2^2 \ei(D)^2] \le (\ef[|D-\ef[D]|_2^4])^{1/2} (\ef[\ei(D)^4])^{1/2}$.
By the $c_r$ inequality, $\ef[\ei(D)^4] \lesssim \ef[\yi(D)^4] + y_n^4 + |\thetan|_2^4 \ef[|D-\ef[D]|_2^4]$,
so $\en \ef[\ri(D)^2] = O(1)$ under the conditions of Proposition \ref{prop:ols-linearization}.
Then $A_n^1 = \Op(\negrootn)$.
The proof for $A_n^2$ is similar.
Applying Lemma \ref{lem:generic-rootn-consistency} entrywise to the functions $d \mapsto (d - \ef[D])(d - \ef[D])'$, whose average second moments are bounded by the conditions of Proposition \ref{prop:ols-linearization}, and using $\en[\Di] = \ef[D] + \Op(\negrootn)$, we get $\varn(\Di) = \varf(D) + \Op(\negrootn)$.
Since $\varf(D) \succ 0$, the matrix $\varn(\Di)$ is invertible with probability approaching one, with $\varn(\Di)\inv = \varf(D)\inv + \Op(\negrootn)$.
Then overall
\[
\wh \beta = \varn(\Di)\inv[A_n^1 + A_n^2 + B_n] = \thetan + \varf(D)\inv A_n^1 + \Op(n\inv),
\]
using $A_n^1 = \Op(\negrootn)$ and $B_n = \Op(n\inv)$.
This proves the claim with influence function $\si(d) = \varf(D)\inv(d-\ef[D])\ei(d)$.
\end{proof}

We now consider the discrete choice setting of Example~\ref{ex:logit}, where outcomes are binary, $\yi(d) \in \{0, 1\}$.
Define the sample and population log-likelihood criteria
\begin{align*}
\objest(\beta) &= \en[\yi\log(L(\beta'\Di)) + (1-\yi)\log(1-L(\beta'\Di))] \\
\objn(\beta) &= \ef[\ynbar(D)\log(L(\beta'D)) + (1-\ynbar(D))\log(1-L(\beta'D))],
\end{align*}
so that the MLE $\wh \beta$ maximizes $\objest(\cdot)$ and the best logistic approximation $\betan$ is the maximizer of $\objn(\cdot)$.
As part of the setting, we take a maximizer $\betan$ to exist for each $n$, which rules out degenerate configurations such as perfect separation.

\begin{proof}[Proof of Proposition \ref{prop:logit-linearization}]

\emph{Step 1:} First, we show consistency.
We claim $\objest(\beta) = \objn(\beta) + \Op(\negrootn)$ for each fixed $\beta$.
To see this, fix $\beta$ and define $\ell_i(d) = \yi(d)\log L(\beta'd) + (1-\yi(d))\log(1-L(\beta'd))$, so that $\objest(\beta) = \en[\ell_i(\Di)]$ and $\objn(\beta) = \en \ef[\ell_i(D)]$.
By Lemma~\ref{lem:generic-rootn-consistency} it suffices to check $\viid(\ell) = O(1)$.
Note $|\log L(v)| \le |v| + \log 2$ and $|\log(1-L(v))| \le |v| + \log 2$.
Since $\yi(d) \in \{0, 1\}$, exactly one of the two terms in $\ell_i(d)$ is nonzero, so $|\ell_i(d)| \le |\beta'd| + \log 2$.
Then by Cauchy-Schwarz,
\begin{align*}
\ef[\ell_i(D)^2] \le \ef[(|\beta'D| + \log 2)^2] \lesssim 1 + |\beta|_2^2 \, \ef[|D|_2^2] = O(1)
\end{align*}
under the conditions of Proposition~\ref{prop:logit-linearization}, so $\viid(\ell) \le \en \ef[\ell_i(D)^2] = O(1)$.
This proves the claim.
Note that $\objest(\cdot)$ is concave.
Also observe that $|\nabla_\beta \objn(\beta)|_2 = |\ef[(\ynbar(D) - L(\beta'D))D]|_2 \le \ef[|D|_2] < \infty$ by the conditions of Proposition~\ref{prop:logit-linearization}.
Then $\objn(\beta)$ is uniformly Lipschitz in $n$, hence equicontinuous.
A suitably modified version of the convexity lemma in \cite{pollard1991asymptotics} then implies $\sup_{\beta \in K} |\objest(\beta) - \objn(\beta)| = \op(1)$ for any compact $K \sub \mr^m$.

We now prove $\wh \beta \convp \betan$.
Since $\sup_n |\betan|_2 < \infty$ by assumption, choose $R$ large enough that $B(\betan, 1) \sub \operatorname{int}(B_R)$ for all $n$, where $B_R = \{\beta : |\beta|_2 \le R\}$.
Let $\wh \beta_R = \argmax_{\beta \in B_R} \objest(\beta)$.
The ULLN on $B_R$ gives $\sup_{\beta \in B_R} |\objest(\beta) - \objn(\beta)| = \op(1)$.
Since $\wh \beta_R$ maximizes $\objest$ over $B_R$ and $\betan \in B_R$,
\begin{align*}
\objn(\betan) - \objn(\wh \beta_R) &\le [\objn(\betan) - \objest(\betan)] + [\objest(\wh \beta_R) - \objn(\wh \beta_R)] \\
&\le 2 \sup_{\beta \in B_R} |\objest(\beta) - \objn(\beta)| = \op(1).
\end{align*}
We claim for any $\eps > 0$, there exists $\eta > 0$
such that $\objn(\betan) - \objn(\beta) \ge \eta$ whenever $|\beta - \betan|_2 \ge \eps$ and $\beta \in B_R$.
To see this, $\objn(\betan) - \objn(\beta) = -\tfrac{1}{2}(\beta - \betan)'\nabla^2 \objn(\bar \beta)(\beta - \betan)$ for some $\bar \beta \in [\betan, \beta]$ by Taylor's theorem, using $\nabla \objn(\betan) = 0$.
Since $\bar \beta \in B_R$, we have $|\bar \beta' D| \le R|D|_2$, so $L'(\bar \beta'D) = L(\bar \beta'D)(1 - L(\bar \beta'D)) \ge L(R|D|_2)(1 - L(R|D|_2)) \equiv w_R(D) > 0$.
Then $-\nabla^2 \objn(\bar \beta) = \ef[L'(\bar \beta'D)DD'] \succeq \ef[w_R(D) DD'] \equiv M_R$.
Note that $M_R \succ 0$, since if $M_R v = 0$ for some $v \ne 0$, then $w_R(D)(v'D)^2 = 0$.
Since $w_R(D) > 0$, we have $v'D = 0$, contradicting $\ef[DD'] \succ 0$.
Also, $M_R$ does not depend on $n$.
Then $\objn(\betan) - \objn(\beta) \ge \tfrac{1}{2} \lambda_{\min}(M_R) |\beta - \betan|_2^2 \ge \tfrac{1}{2} \lambda_{\min}(M_R) \eps^2 \equiv \eta > 0$.
Then $P(|\wh \beta_R - \betan|_2 > \eps) \le P(\objn(\betan) - \objn(\wh \beta_R) > \eta) = o(1)$, so that $\wh \beta_R \convp \betan$.
Then $P(\wh \beta = \wh \beta_R) \ge P(|\wh \beta_R - \betan|_2 < 1) \to 1$, so $\wh \beta \convp \betan$.

\emph{Step 2:} Next we derive the influence function.
Define $\ri(d, \beta) = (\yi(d) - L(\beta'd))d$, so the MLE solves $\en[\ri(\Di, \wh \beta)] = 0$.
Observe that $\nabla_\beta \ri(d, \beta) = -L'(\beta'd) dd'$ where $L'(v) = L(v)(1-L(v))$.
Let $\beta_t \equiv \betan + t(\wh \beta - \betan)$ for $t \in [0, 1]$.
Applying the fundamental theorem of calculus to $t \mapsto \en[\ri(\Di, \beta_t)]$ and using $\en[\ri(\Di, \wh \beta)] = 0$ gives
\[
0 = \en[\ri(\Di, \betan)] - \wh J_n (\wh \beta - \betan), \quad \quad \wh J_n \equiv \int_0^1 \en[L'(\beta_t'\Di)\Di\Di'] \, dt.
\]
Since $|L''(v)| \le 1/4$, the mean value theorem gives $|L'(\beta_t'\Di) - L'(\betan'\Di)| \le \tfrac{1}{4} t \, |\wh \beta - \betan|_2 |\Di|_2$ for each $t \in [0, 1]$, so
\[
|\wh J_n - \en[L'(\betan'\Di)\Di\Di']|_{op} \le \tfrac{1}{4}|\wh \beta - \betan|_2 \cdot \en[|\Di|_2^3].
\]
By Jensen, $\en[|\Di|_2^3] \le (\en[|\Di|_2^4])^{3/4}$.
The latter is $\Op(1)$ by Markov's inequality, since $\eg\big[\en[|\Di|_2^4]\big] = \ef[|D|_2^4] < \infty$.
Working entry-wise with $z_i(d) \equiv L'(\betan'd)d_jd_l$, we have $\viid(z) \le \ef[D_j^2D_l^2] \le \ef[|D|_2^4] = O(1)$ since $L' \le 1$.
Since $k = O(1)$, this implies $\en[L'(\betan'\Di)\Di\Di'] = J_n + \Op(\negrootn)$ by Lemma \ref{lem:generic-rootn-consistency}.
Writing $\wh J_n = J_n + R_n$, the triangle inequality combines the last two bounds into
\begin{equation} \label{eqn:logit-jacobian-error}
|R_n|_{op} \le \tfrac{1}{4}|\wh \beta - \betan|_2 \cdot \en[|\Di|_2^3] + \Op(\negrootn) = |\wh \beta - \betan|_2 \cdot \Op(1) + \Op(\negrootn).
\end{equation}
In particular $R_n = \op(1)$, since $\wh \beta \convp \betan$ by Step~1.
Next we show $\en[\ri(\Di, \betan)] = \Op(\negrootn)$.
Since $\en \ef[\ri(D, \betan)] = 0$ by the first-order condition $\nabla \objn(\betan) = 0$, it suffices by Lemma~\ref{lem:generic-rootn-consistency} to check $\en \ef[\ri^j(D, \betan)^2] = O(1)$ for each component $\ri^j(d, \betan) = (\yi(d) - L(\betan'd))d_j$.
Since $\yi(d) \in \{0, 1\}$ implies $|\yi(d) - L(\betan'd)| \le 1$, we have $\ef[\ri^j(D, \betan)^2] \le \ef[D_j^2] = O(1)$.
Since $k = O(1)$, the claim follows.
Since $\sup_n |\betan|_2 < \infty$, there exists $R$ such that $|\betan'D| \le R|D|_2$ for all $n$, so $L'(\betan'D) \ge w_R(D) \equiv L(R|D|_2)(1-L(R|D|_2)) > 0$.
Then $J_n \succeq M_R \equiv \ef[w_R(D)DD'] \succ 0$ for all $n$, where $M_R \succ 0$ by $\ef[DD'] \succ 0$ as in Step~1.
In particular $|J_n\inv|_{op} \le \lambda_{\min}(M_R)\inv = O(1)$.
Since $\wh J_n = J_n + R_n$ is symmetric with $R_n = \op(1)$, Weyl's inequality gives $\lambda_{\min}(\wh J_n) \ge \lambda_{\min}(J_n) - |R_n|_{op} \ge \tfrac{1}{2}\lambda_{\min}(M_R)$ with probability approaching one, so $\wh J_n$ is invertible with $|\wh J_n\inv|_{op} \le 2\lambda_{\min}(M_R)\inv = \Op(1)$.
Then
\[
\wh \beta - \betan = \wh J_n\inv \en[\ri(\Di, \betan)] = \Op(1) \cdot \Op(\negrootn) = \Op(\negrootn).
\]
Feeding this rate back into Equation \eqref{eqn:logit-jacobian-error} gives $R_n = \Op(\negrootn)$.
Then by the resolvent identity $\wh J_n\inv = J_n\inv - J_n\inv R_n \wh J_n\inv$,
\[
|\wh J_n\inv - J_n\inv|_{op} \le |J_n\inv|_{op} \, |R_n|_{op} \, |\wh J_n\inv|_{op} = \Op(\negrootn).
\]
Therefore
\[
\wh \beta - \betan = J_n\inv \en[\ri(\Di, \betan)] + (\wh J_n\inv - J_n\inv) \en[\ri(\Di, \betan)] = \en[\si(\Di)] + \Op(n\inv),
\]
where $\si(d) = J_n\inv \ei(d) \cdot d$ for residual $\ei(d) = \yi(d) - L(\betan'd)$.
\emph{Step 3:} Finally, we show identification, i.e., that Equation \eqref{eqn:logit-kl-projection} holds.
Note that $\kl(p \| q) = p\log(p/q) + (1-p)\log((1-p)/(1-q))$, so we have $p \log q + (1-p)\log(1-q) = -\kl(p \| q) + h(p)$ where $h(p) = p\log p + (1-p)\log(1-p)$ does not depend on $q$.
Then
\begin{align*}
\betan &= \argmax_\beta \; \ef[\ynbar(D)\log L(\beta'D) + (1-\ynbar(D))\log(1-L(\beta'D))] \\
&= \argmin_\beta \; \ef[\kl(\ynbar(D) \| L(\beta'D))].
\end{align*}
This finishes the proof.
\end{proof}

\subsection{Proofs for Section \ref*{section:dispersion-match-quality}}

\begin{proof}[Proof of Proposition \ref{prop:dispersion-interpretation}]
Expanding the sample variance directly, we have
\begin{align*}
\vark(\phi(\Di)) &= k\inv \sum_{i \in [k]} \phi(\Di)^2 - (k(k-1))\inv \sum_{i \neq j \in [k]} \phi(\Di)\phi(\Dj).
\end{align*}
Since $G \in \transportsym(F)$, by fixed marginals and exchangeability,
\begin{align*}
\eg \vark(\phi(\Di)) &= \ef[\phi(D)^2] - \eg[\phi(\done)\phi(\dtwo)] \\
&= \ef[\phi(D)^2] - \ef[\phi(D)]^2 + \ef[\phi(D)]^2 - \eg[\phi(\done)\phi(\dtwo)] \\
&= \varf(\phi) - \covg(\phi(\done), \phi(\dtwo)).
\end{align*}
Then to finish the proof, observe
\begin{align*}
(k-1) \bigg (\frac{\eg \vark(\phi(\Di))}{\varf(\phi)} - 1 \bigg ) = -(k-1) \corrg(\phi(\done), \phi(\dtwo)) = \dispg(\phi).
\end{align*}
\end{proof}

\begin{proof}[Proof of Lemma \ref{lem:var-components}]
Consider the first identity.
We again expand the variance of the difference:
\begin{align*}
    \vmatch(s) &= (2n(k-1))\inv \sum_g \sum_{i \neq j \in [k]} \big( \varf(\sig) + \varf(\sjg) - 2\covf(\sig, \sjg) \big) \\
    &= (n(k-1))\inv \sum_g \big( (k-1)\sum_{i \in [k]} \varf(\sig) - \sum_{i \neq j \in [k]} \covf(\sig, \sjg) \big) \\
    &= n\inv \sum_g \sum_{i \in [k]} \varf(\sig) - (n(k-1))\inv \sum_g \sum_{i \neq j \in [k]} \covf(\sig, \sjg) \\
    &= \viid(s) - (n(k-1))\inv \sum_g \sum_{i \neq j \in [k]} \covf(\sig, \sjg).
\end{align*}
Consider the second identity.
We start from the variance of the group mean $\sgbar$:
\[
k^2 \varf(\sgbar) = \varf\sum_{i \in [k]} \sig = \sum_{i \in [k]} \varf(\sig) + \sum_{i \neq j \in [k]} \covf(\sig, \sjg).
\]
Summing over all groups $g$ and multiplying by $n\inv$ gives
\begin{align*}
    n\inv k^2 \sum_g \varf(\sgbar) &= n\inv \sum_g \sum_{i \in [k]} \varf(\sig) + n\inv \sum_g \sum_{i \neq j \in [k]} \covf(\sig, \sjg) \\
    &= \viid(s) + n\inv \sum_g \sum_{i \neq j \in [k]} \covf(\sig, \sjg).
\end{align*}
For the next inequality, note that
\begin{align*}
|c(s)| &= \bigg| n\inv \sum_g \sum_{i \neq j \in [k]} \covf(\sig, \sjg) \bigg| \le n\inv \sum_g \sum_{i \neq j \in [k]} |\covf(\sig, \sjg)| \\
&\le n\inv \sum_g \sum_{i \neq j \in [k]} (1/2)(\varf(\sig) + \varf(\sjg)) \\
&= n\inv \sum_g (k-1) \sum_{i \in [k]} \varf(\sig) = (k-1) \en \varf(\si) = (k-1) \viid(s).
\end{align*}
The second inequality is by Cauchy-Schwarz and Young's inequality.
Then $k \vg(s) = \viid(s) + c(s) \leq k \viid(s)$ and $\vmatch(s) = \viid(s) - (k-1)\inv c(s) \leq 2 \viid(s)$.
For the bounds on $c(s)$: the lower bound $c(s) \ge -\viid(s)$ follows from $k\vg(s) = \viid(s) + c(s) \ge 0$, and the upper bound $c(s) \le (k-1)\viid(s)$ from $\vmatch(s) = \viid(s) - (k-1)\inv c(s) \ge 0$.
The statement about match quality follows by rearranging the identity $\vmatch(s) = \viid(s) - (k-1)\inv c(s)$, noting $\matchcoeffk(s) = 1-\vmatch(s)/\viid(s)$.
Finally consider the last statement.
Expand
\begin{align*}
    \sum_{i \neq j \in [k]} \varf(\sig - \sjg) &= \sum_{i \neq j \in [k]} \big( \varf(\sig) + \varf(\sjg) - 2\covf(\sig, \sjg) \big) \\
    &= 2(k-1)\sum_{i \in [k]} \varf(\sig) - 2\sum_{i \neq j \in [k]} \covf(\sig, \sjg).
\end{align*}
Next, we use $ \sum_{i \in [k]} \varf(\sig) = \sum_{i \in [k]} \varf(\sig - \sgbar) + k \varf(\sgbar)$.
Also, from the definition of $\sgbar$, we have $k^2 \varf(\sgbar) = \sum_{i \in [k]} \varf(\sig) + \sum_{i \neq j \in [k]} \covf(\sig, \sjg)$.
Combining these gives
\begin{align*}
    \sum_{i \neq j \in [k]} \varf(\sig - \sjg) &= 2(k-1)\sum_{i \in [k]} \varf(\sig) - 2\big( k^2 \varf(\sgbar) - \sum_{i \in [k]} \varf(\sig) \big) \\
    &= 2k \sum_{i \in [k]} \varf(\sig) - 2k^2 \varf(\sgbar) \\
    &= 2k \big( \sum_{i \in [k]} \varf(\sig) - k \varf(\sgbar) \big) = 2k \sum_{i \in [k]} \varf(\sig - \sgbar).
\end{align*}
Summing over $g$ and multiplying by $(2n(k-1))\inv$ yields
\begin{align*}
    \vmatch(s) &= (2n(k-1))\inv \sum_g 2k \sum_{i \in [k]} \varf(\sig - \sgbar) \\
    &= (k-1)\inv (n/k)\inv \sum_g \sum_{i \in [k]} \varf(\sig - \sgbar).
\end{align*}
This proves the first identity.
\end{proof}

\subsubsection{Random Matching}

\begin{prop}[Random Matching]\label{prop:random-matching-baseline}
Suppose $\varf(\si) < \infty$ for $i \in [n]$.
Let $\match$ be a uniform random partition of $[n]$ into $n/k$ groups.
Define $\snbar(d) \equiv n\inv \sum_{i=1}^n \si(d)$.
\begin{equation}\label{eqn:random-matching-baseline}
E_{\match}[\matchcoeffk(s)] = \frac{n}{n-1}\cdot \frac{\varf(\snbar)}{\viid(s)} - \frac{1}{n-1}.
\end{equation}
\end{prop}

\begin{proof}
Define $c(s) \equiv n\inv \sum_g \sum_{i\neq j \in [k]} \covf(\sig, \sjg)$.
By Lemma~\ref{lem:var-components}, the match coefficient $\matchcoeffk(s) = (k-1)\inv c(s) / \viid(s)$.
Then we calculate
\begin{align*}
E_{\match}[c(s)] &= n\inv E_{\match} \sum_{i\neq j \in [n]} \one(g(i)=g(j)) \covf(\si, \sj) \\
 &= n\inv  \sum_{i\neq j \in [n]} P_{\match}(g(i)=g(j)) \covf(\si, \sj) = \frac{k-1}{n(n-1)}  \sum_{i\neq j \in [n]} \covf(\si, \sj).
\end{align*}
The third equality since $P_{\match}(g(i)=g(j)) = (k-1) / (n-1)$ by a counting argument.
Then $E_{\match}[\matchcoeffk(s)] = n\inv (n-1)\inv \sum_{i\neq j \in [n]} \covf(\si, \sj) / \viid(s)$.
\begin{align*}
\varf(\snbar) = n^{-2} \sum_{i,j} \covf(\si, \sj) = n^{-2} \sum_{i} \varf(\si) + n^{-2} \sum_{i \neq j} \covf(\si, \sj).
\end{align*}
Rearranging gives $\sum_{i \neq j} \covf(\si, \sj) = n^2 \varf(\snbar) - n \viid(s)$.
Then
\begin{align*}
E_{\match}[\matchcoeffk(s)] &= \frac{1}{n(n-1)\viid(s)} (n^2 \varf(\snbar) - n \viid(s)) \\
&= \frac{n}{n-1} \frac{\varf(\snbar)}{\viid(s)} - \frac{1}{n-1}.
\end{align*}
This completes the proof.
\end{proof}

\subsubsection{Covariate Power} \label{appendix:covariate-power}

To formalize the relationship between match quality and covariate predictive power, suppose just for illustration in the remainder of this subsection that $(X_i, \si(\cdot)) \sim P$ iid for some fixed measure $P$.
This guarantees that the units in our experiment are ``typical'' relative to a fixed relationship $P$ between $X_i$ and $\si(\cdot)$.
Let $\mu(d, X_i) \equiv \ep[\si(d) | X_i]$ and the predictive power of covariates by
\begin{equation}
R^2_{s|X} \equiv \frac{\ep \varf(\mu(D, X_i))}{\ep \varf(\si(D))}.
\end{equation}
This measures how well $X_i$ predicts heterogeneity in influence functions $\si(d)$ across treatment levels, the predictive power for ``design heterogeneity.''
Note how this differs from the typical superpopulation $R^2$ coefficient, say in a linear regression.
For example, if $\mu(d, X_i) = \mu(X_i)$, then $R^2_{s|X} = 0$, even if $\mu(X_i)$ is highly predictive of $\si(\cdot)$ over the distribution $P$.

\begin{prop}[Covariate Power]\label{prop:match-quality-decomposition}
Let $(X_i, \si(\cdot)) \sim P$ iid with $\ep[\varf(\si)] > 0$ and $\ep[\ef[\si(D)^4]] < \infty$.
Suppose the matching $\match$ is $\sigma(\Xn)$-measurable.
Then if $k = O(1)$ as $n \to \infty$
\begin{equation} \label{eqn:match-quality-decomposition}
\matchcoeffk(s) = R^2_{s|X} \cdot \matchcoeffk(\mu) + \op(1).
\end{equation}
\end{prop}

The term $\matchcoeffk(\mu)$ measures within-group match quality on the covariate features $\mu_i(\cdot) = \mu(\cdot, X_i)$, on average over $D \sim F$.
This proposition formalizes the intuition above, showing that, under slightly stronger assumptions, $\matchcoeffk(s)$ is increasing in the product of both $R^2_{s|X}$ and $\matchcoeffk(\mu)$.

\textbf{Practical Recommendations.}
Proposition \ref{prop:match-quality-decomposition} suggests that we want to measure baseline covariates $X_i$ with large predictive power $R^2_{s|X}$.
In particular, we want to match tightly on the features of those covariates $\mu(\cdot, X_i)$ most predictive of $\si(\cdot)$.
Of course, both of these objects are unknown at design time before the experimenter has access to data.
One could try to estimate them with a pilot, but the literature has raised concerns about the finite sample properties of pilot-based designs that try to estimate related quantities \citep{cai2024, cytrynbaum2023}.

\begin{proof}[Proof of Proposition \ref{prop:match-quality-decomposition}]
We claim that (1) $c(s) = c(\mu) + \opone$ for covariance $c(s) = n\inv \sum_g \sum_{i\neq j \in [k]} \covf(\sig, \sjg)$ and (2) $\viid(\mu) = \ep[\varf(\mui)] + \opone$ and $\viid(s) = \ep[\varf(\si)] + \opone$.
Under these assumptions, by Lemma \ref{lem:var-components}, $\matchcoeffk(\mu) = (k-1)\inv c(\mu) / \viid(\mu)$, so
\begin{align*}
\matchcoeffk(s) &= \frac{(k-1)\inv c(s)}{\viid(s)} = \frac{(k-1)\inv [c(\mu) + \opone]}{\viid(s)} \\
&= \frac{(k-1)\inv c(\mu)}{\viid(s)} + \opone = \frac{\viid(\mu)}{\viid(s)} \cdot \matchcoeffk(\mu) + \opone.
\end{align*}
The third equality since $\viid(s) \convp \ep[\varf(\si)] > 0$ by assumption.
Since $\matchcoeffk(\mu) \in [-(k-1)\inv, 1]$ is bounded and $\viid(\mu) / \viid(s) \convp R^2_{s|X}$ by the claims above and continuous mapping theorem, we conclude $\matchcoeffk(s) = R^2_{s|X} \cdot \matchcoeffk(\mu) + \opone$.

Consider claim (1).
We expand $\covf(\sig, \sjg) = \covf(\muig, \mujg) + \covf(\muig, e_{jg}) + \covf(e_{ig}, \mujg) + \covf(e_{ig}, e_{jg})$ for $i \neq j$.
We can decompose $c(s) = c(\mu) + A_n + B_n$, where $A_n$ collects the cross covariance terms $\covf(\muig, e_{jg})$ and $\covf(e_{ig}, \mujg)$, and $B_n$ collects the $\covf(e_{ig}, e_{jg})$ terms.
We show $A_n, B_n = \opone$.
To see this, note
\[
E[A_n | \Xn] = n\inv \sum_g \sum_{i \neq j \in [k]} \big[E[\covf(\muig, e_{jg})|\Xn] + E[\covf(e_{ig}, \mujg) | \Xn] \big].
\]
We have $\covf(\muig, e_{jg}) = \ef[\muig(D) e_{jg}(D)] - \ef[\muig(D)]\ef[e_{jg}(D)]$.
Recall matching $\match: [n] \to [k] \times [n/k]$ with $\match \in \sigma(\Xn)$.
Then $\ef[\muig(D) e_{jg}(D)] = \sum_{r,l} \one(\match(r)=ig,\, \match(l)=jg)\, \ef[\mu(X_r, D)\, e_l(D)]$.
Note $\one(\match(r)=ig,\, \match(l)=jg)$ and $\mu(X_r, d)$ are $\sigma(\Xn)$-measurable for each $d$.
Moreover, we claim $\ep[\ef[\mu(X_r, D)\, e_l(D)] | \Xn] = 0$ for each $r \neq l$.
It suffices to show $\ep[\ef[\mu(X_r, D)\, e_l(D)] \cdot h(\Xn)] = 0$ for any bounded measurable $h(\cdot)$.
To see this, note
\begin{align*}
&\ep\big[\ef[|\mu(X_r, D)\, e_l(D) \, h(\Xn)|]\big] \lesssim \ep\big[\ef[|\mu(X_r, D)\, e_l(D)|]\big] \\
&= \ef\big[\ep[|\mu(X_r, D)\, e_l(D)|]\big] \leq \ef\big[\ep[\mu(X_r, D)^2]^{1/2} \, \ep[e_l(D)^2]^{1/2}\big] \\
&\leq \ef\big[\ep[\mu(X_r, D)^2] + \ep[e_l(D)^2]\big] \leq \ef\big[\ep[\si(D)^2]\big] = \ep\big[\ef[\si(D)^2]\big] < \infty.
\end{align*}
The first inequality since $|h| \leq M$ for some $M < \infty$.
The first equality is Tonelli's theorem.
The second and third inequalities are Cauchy-Schwarz and Young's.
The fourth since $\ep[\mu(X_r, d)^2] = \ep[(\ep[\si(d) | X_r])^2] \leq \ep[\si(d)^2]$ by Jensen's inequality, and $\ep[e_l(d)^2] \leq \ep[\si(d)^2]$.
The last equality is Tonelli's theorem, and finiteness holds by assumption.
Then by Fubini's theorem,
\begin{align*}
\ep\big[\ef[\mu(X_r, D)\, e_l(D)] \cdot h(\Xn)\big] &= \ef\big[\ep[\mu(X_r, D)\, e_l(D) \cdot h(\Xn)]\big].
\end{align*}
For each fixed $d$, since $\mu(X_r, d)$ and $h(\Xn)$ are $\sigma(\Xn)$-measurable, the tower property gives $\ep[\mu(X_r, d)\, e_l(d) \cdot h(\Xn)] = \ep\big[\mu(X_r, d) \, h(\Xn) \, \ep[e_l(d) | \Xn]\big]$.
Since $(X_i, \si(\cdot)) \sim P$ iid, the residual $e_l(\cdot)$ satisfies $\ep[e_l(d) | \Xn] = \ep[e_l(d) | X_l] = 0$ for all $d$.
Hence the integrand is zero for all $d$, so $\ep[\ef[\mu(X_r, D)\, e_l(D)] \cdot h(\Xn)] = 0$.
Since $h(\cdot)$ was arbitrary, we conclude $\ep[\ef[\mu(X_r, D)\, e_l(D)] | \Xn] = 0$.
Then
\begin{align*}
&\ep[\ef[\muig(D) e_{jg}(D)] | \Xn] \\
&= \sum_{r,l} \one(\match(r)=ig,\, \match(l)=jg)\, \ep[\ef[\mu(X_r, D)\, e_l(D)] | \Xn] = 0.
\end{align*}
By the same argument, $\ep[\ef[\muig(D)] \ef[e_{jg}(D)] | \Xn] = 0$, since $\ef[e_{jg}(D)] = \sum_l \one(\match(l) = jg)\, \ef[e_l(D)]$ and $\ep[\ef[e_l(D)] | \Xn] = 0$.
Then we have shown that $\ep[\covf(\muig, e_{jg}) | \Xn] = 0$ for each $i \neq j \in [k]$.
By symmetry, $\ep[A_n | \Xn] = 0$.
Similarly, one can show $\varp(A_n | \Xn) = \Op(n\inv)$.
Then by conditional Chebyshev, $A_n = \Op(\negrootn)$.
The proof that $B_n = \Op(\negrootn)$ is similar.

Next, we show claim (2).
Consider $\viid(\mu) = \en[\varf(\mui)]$.
Note $\ep[\varf(\mui)] \leq \ep[\ef[\mui(D)^2]] = \ef[\ep[\mui(D)^2]] \leq \ef[\ep[\si(D)^2]] = \ep[\ef[\si(D)^2]] < \infty$.
The equalities are by Tonelli's theorem to exchange $\ep$ and $\ef$.
The second inequality since $\ep[\mui(d)^2] = \ep[(\ep[\si(d) | X_i])^2] \leq \ep[\si(d)^2]$ by Jensen's inequality.
Then by WLLN, $\viid(\mu) = \ep[\varf(\mui)] + \opone$.
Similarly, $\viid(s) = \en[\varf(\si)] = \ep[\varf(\si)] + \opone$.
This proves the claim.
\end{proof}

\subsection{Proofs for Section \ref*{section:efficiency-analysis}}

\begin{proof}[Proof of Lemma \ref{lem:canonical-space}]
First, we show $R$ is well-defined.
To see this, $E_Q[(R\phi)(U)^2] = E_Q[\phi(v(U))^2] = \ef[\phi(D)^2] < \infty$ for any $\phi \in \elltwo(F)$.
Next, claim $R: L^2(F) \to L^2(Q)$ is a linear isometry.
Linearity is clear.
For the isometry property,
\begin{align*}
\lf \phi, \psi \rf &= \ef[\phi(D)\psi(D)] = E_Q[\phi(v(U))\psi(v(U))] \\
&= E_Q[(R\phi)(U)(R\psi)(U)] = \lag R\phi, R\psi \rag_Q.
\end{align*}
In particular, $|R|_{op} = 1$, so $R$ is bounded and the adjoint $R^*$ exists by Riesz representation.
Note $R$ maps $\ltwonot$ into $L^2_0(Q)$, since $E_Q[(R\phi)(U)] = \ef[\phi(D)] = 0$ for $\phi \in \ltwonot$.
Next, we show $\lag R^* T_H R \phi, \psi \rag_F = \lag U_G \phi, \psi \rag_F$ for $\phi, \psi \in \elltwo(F)$.
\begin{align*}
\lag U_G \phi, \psi \rag &= \ef[U_G\phi(D)\psi(D)] = \eg[\phi(\done) \psi(\dtwo)] = \eh[(R\phi)(U_1)(R\psi)(U_2)] \\
&= E_H[E[(R\phi)(U_1)|U_2](R\psi)(U_2)] = E_Q[T_H (R\phi)(U) (R\psi)(U)] \\
&= \lag T_H R \phi, R \psi \rag_Q = \lag R^* T_H R \phi, \psi \rag_F.
\end{align*}
Observe that if $\lag a - b, \psi \rag = 0$ for any $\psi$, then $a=b$ by taking $\psi = a-b$.
By this observation, we have shown $R^* T_H R \phi = U_G \phi$ for each $\phi \in \elltwo(F)$, so $R^* T_H R = U_G$, proving the first claim.

Next, impose injectivity of $v(\cdot)$.
Note that $D = v(U)$, so $F(D \in \image(v)) = 1$.
Define $S: \elltwo(Q) \to \elltwo(F)$ by $(S \phi)(d) = \phi(v\inv(d))$ for $d \in \image(v)$ and $(S \phi)(d) = 0$ for $d \not \in \image(v)$.
Note that $\image(v)$ is measurable, since the image of a Borel set under an injective Borel-measurable map is Borel by the Lusin--Suslin theorem.
Then $S\phi$ is measurable, since $v\inv$ is measurable on $\image(v)$ by assumption.
Moreover, $\ef[(S\phi)(D)^2] = \ef[\phi(v\inv(D))^2] = E_Q[\phi(U)^2] < \infty$ for any $\phi \in \elltwo(Q)$, showing $S$ is well-defined.
Then we have $RS \phi(u) = (S \phi)(v(u)) = \phi(v\inv (v(u))) = \phi(u)$.
Then $RS = I$ on $\elltwo(Q)$.
In particular, $R$ is surjective and $S$ is injective.
Moreover, $R^* R = I$ on $\elltwo(F)$, so $R$ is injective.
Then $R^*RS = R^*$, so $S = R^*$.
Then $SR = RS = I$, so $R^* = S = R\inv$.
Clearly $R\inv$ is also an isometry, hence continuous.

We will show $\ltwonot = \oplus_{m \ge 1} E_m$.
Let $\phi \in E_m$.
Then $T_H R \phi = \evalm R \phi$, so $R\inv T_H R \phi = \evalm R\inv R \phi = \evalm \phi$, so $E_m$ is an eigenspace with eigenvalue $\evalm$.
Note $R\inv$ is continuous, since $R\inv$ is a linear isometry, hence a bounded linear map.
Each $W_m$ is closed by assumption, so $E_m = R\inv(W_m)$ is closed by continuity of $R\inv$.
Let $\phi \in \ltwonot$.
Then by assumption $R \phi = \sum_{m \ge 1} P^W_m R \phi = \lim_{l \to \infty} \sum_{m \ge 1}^l P^W_m R \phi$.
By linearity and continuity of $R\inv$, $\phi = \sum_{m \ge 1} R\inv P^W_m R \phi$.
Since $P^W_m R \phi \in W_m$, then $R\inv P^W_m R \phi \in R\inv(W_m) = E_m$.
Let $\phi \in E_m$ and $\psi \in E_l$ for $l \not = m$.
Then $R \phi \in W_m$ and $R \psi \in W_l$, so by isometry $\lag \phi, \psi \rag_F = \lag R \phi, R \psi \rag_Q = 0$ since $W_m \perp W_l$ by assumption.
This shows $E_m \perp E_l$ for $l \not = m$.
Putting this all together, we have shown $\ltwonot = \oplus_{m \ge 1} E_m$ with the required properties.
\end{proof}

\begin{proof}[Proof of Lemma \ref{lem:variance-components-orthogonal}]
For (1), note that $\sih \in H$ and $\si - \sih \perp H$ by definition of projection.
Since $H \sub \ltwonot$, we have $\ef[\sih] = 0$.
Then $\varf(\si) = \varf(\si - \sih + \sih) = \varf(\si - \sih) + 2 \covf(\si - \sih, \sih) + \varf(\sih)$.
The middle term $\covf(\si - \sih, \sih) = \ef[(\si - \sih)\sih] - \ef[\si - \sih]\ef[\sih] = \lf \si - \sih, \sih \rf = 0$ by orthogonality, using $\ef[\sih] = 0$.
The conclusion follows since $\viid(s) = \en \varf(\si)$.
Next consider the second claim.
We expand the covariance term:
\begin{align*}
\covf(\sig, \sjg) &= \covf\big(\sig - \sig^H + \sig^H, \sjg - \sjg^H + \sjg^H \big) \\
&= \covf(\sig^H, \sjg^H) \\
&\quad + \covf(\sig - \sig^H, \sjg - \sjg^H) + \covf(\sig^H, \sjg - \sjg^H) + \covf(\sig - \sig^H, \sjg^H).
\end{align*}
The cross term is $\covf(\sig^H, \sjg - \sjg^H) = \lf \sig^H, \sjg - \sjg^H \rf = 0$ since $\sig^H \in H$ with $\ef[\sig^H] = 0$ and $\sjg - \sjg^H \perp H$, and similarly for the second term $\covf(\sig - \sig^H, \sjg^H)$.
The claim now follows since $c(s) = n\inv \sum_g \sum_{i\neq j} \covf(\sig, \sjg)$.
Finally, the third claim follows from the identity $\vmatch(s) = \viid(s) - (k-1)\inv c(s)$ in Lemma~\ref{lem:var-components} and the additivity of $\viid(\cdot)$ and $c(\cdot)$ shown above.
For the final claim, write $w_\perp = \viid(s - \sh)/\viid(s) = 1 - w_H$ by (1).
If $\viid(\sh) = 0$ or $\viid(s - \sh) = 0$, then $w_H = 0$ or $w_\perp = 0$ respectively, and the corresponding $\vmatch$ term also vanishes since $\vmatch \le 2\viid$ by Lemma~\ref{lem:var-components}, so the claim holds trivially.
Otherwise, dividing (3) by $\viid(s)$ and using (1), $\frac{\vmatch(s)}{\viid(s)} = w_H \frac{\vmatch(\sh)}{\viid(\sh)} + w_\perp \frac{\vmatch(s - \sh)}{\viid(s - \sh)} = w_H(1 - \matchcoeffk(\sh)) + w_\perp(1 - \matchcoeffk(s - \sh))$.
Then $\matchcoeffk(s) = 1 - \vmatch(s)/\viid(s) = 1 - w_H - w_\perp + w_H \matchcoeffk(\sh) + w_\perp \matchcoeffk(s - \sh) = w_H \matchcoeffk(\sh) + w_\perp \matchcoeffk(s - \sh)$, using $w_H + w_\perp = 1$.
\end{proof}

\subsection{Proofs for Section \ref*{section:coupling-analysis}}

\begin{proof}[Proof of Lemma \ref{lem:av-operator}]
Let $\phi \in \elltwo(F)$, then clearly $\phi = S_e \phi + S_o \phi + \ef[\phi(D)]$ with $S_e \phi \in \eeven$ and $S_o \phi \in \eodd$ and $\ef[\phi(D)] \in \linearspan(1)$.
If $\phi \in \eeven$ then $S_e \phi = \phi$.
If $\phi \in \eodd$ then $\ef[\phi(D)] = -\ef[\phi(1-D)]$, so $\ef[(S_e \phi)(D)] = 0$.
Then $S_e \phi = (1/2)(\phi(x) + \phi(1-x)) = (1/2)(-\phi(1-x) + \phi(1-x)) = 0$.
Then for any $\phi \in \elltwo(F)$ and $\psi \in \eeven$,
\begin{align*}
\lf \phi - S_e \phi, \psi \rf &= \lf S_o \phi + \ef[\phi(D)], \psi \rf = \ef[(S_o \phi)(D)\psi(D)] \\
&= -\ef[(S_o \phi)(1-D)\psi(1-D)] = - \lf S_o \phi, \psi \rf
\end{align*}
Then we must have $\lf S_o \phi, \psi \rf = 0$.
The 2nd equality above follows since $\psi \in \eeven \perp \linearspan(1)$.
The 3rd equality using $S_o \phi \in \eodd$ and $\psi \in \eeven$.
The final equality since $D \sim 1-D$.
This shows $\phi - S_e \phi \perp \eeven$, so $S_e \phi = \peven \phi$.
A similar argument shows $S_o \phi = \podd \phi$.
This finishes our proof of the projection operators.
Next we show the direct sum.
By the above work, if $\phi \in \ltwonot$, then $\phi = S_e \phi + S_o \phi$, so $\ltwonot = \eeven + \eodd$.
Orthogonality follows by the same argument in the display above, so $\ltwonot = \eeven \oplus \eodd$.
Finally, we characterize the operator $U_G = U$.
We have $(D_1, D_2) = (Z, 1-Z)$ for $Z \sim \unif[0, 1]$.
Then $U\phi(x) = E_G[\phi(\done) | \dtwo=x] = E[\phi(Z) | 1-Z=x] = \phi(1-x)$.
For any $\phi \in \eeven$, we have $U\phi(x) = \phi(1-x) = \phi(x)$, so $\lambda = 1$.
For any $\phi \in \eodd$, we have $U\phi(x) = \phi(1-x) = -\phi(x)$, so $\lambda = -1$.
This finishes the proof.
\end{proof}

\begin{proof}[Proof of Theorem \ref{thm:av-variance}]
By Lemma \ref{lem:av-operator}, $\ltwonot = E_{even} \oplus E_{odd}$, a direct sum of eigenspaces of $U_G$ with dispersions $\dispg(\phi) = -1$ and $\dispg(\phi) = 1$ respectively.
Then the claim in Equation \eqref{eqn:av-efficiency} follows directly from Theorem \ref{thm:eigenspace-decomposition}.
Moreover, by Corollary \ref{cor:efficiency-convex} and the variance identities in Lemma \ref{lem:var-components}, we have
\begin{align*}
n \var_G(\est) &= 2 \viid(s^e) - \vmatch(s^e) + \vmatch(s^o) = 2 \viid(s^e) - [\viid(s^e) - c(s^e)] + \vmatch(s^o) \\
&= \viid(s^e) + c(s^e) + \vmatch(s^o) = 2 \vg(s^e) + \vmatch(s^o).
\end{align*}
This finishes the proof.
\end{proof}

\end{document}